\documentclass[review,onefignum,onetabnum]{siamonline220329}

\usepackage{amsfonts}
\usepackage{graphicx}
\usepackage{algorithmic}

\usepackage{amsmath,amssymb,bm}
\usepackage{dsfont}
\usepackage{mathrsfs}
\usepackage{mathtools}

\usepackage{hyperref,bookmark}

\hypersetup{
  bookmarksnumbered,
  colorlinks=true,
  allcolors=blue,
  pdfauthor={Romain Ait Abdelmalek-Lomenech, Julien Bect, Emmanuel Vazquez},
  pdftitle={Bayesian Active Learning of (small) Quantile Sets through Expected Estimator Modification},
  pdfsubject={},
  pdfkeywords={}}

\usepackage{enumitem}
\setlist[enumerate]{leftmargin=.5in}
\setlist[itemize]{leftmargin=.5in}

\newsiamremark{remark}{Remark}
\newsiamremark{hypothesis}{Hypothesis}
\crefname{hypothesis}{Hypothesis}{Hypotheses}
\newsiamthm{claim}{Claim}

\usepackage{xparse}

\NewDocumentCommand{\citep}{o o m}{%
\IfValueT{#1}{(%
  \IfValueT{#1}{#1 }%
  \cite{#3}%
  \IfBlankTF{#2}{}{, #2}%
  )}%
  \IfNoValueT{#1}{\cite{#3}}%
}
\newcommand{\aref}[1]{%
  \hyperref[#1]{Appendix~\ref*{#1}}}%
  
\usepackage{changepage}
\renewenvironment{adjustwidth}[2]%
  {\ignorespaces}%
  {\ignorespacesafterend}

\headers{Bayesian Active Learning of (small) Quantile Sets through Expected Estimator Modification}{R. Ait Abdelmalek-Lomenech, J. Bect and E. Vazquez}

\title{Bayesian Active Learning of (small) Quantile Sets through Expected Estimator Modification\thanks{Submitted to the editors DATE.
\funding{This work was funded by the french National Research Agency (ANR) in the context of the project SAMOURAI (ANR-20-CE46-0013).}}}

\author{Romain Ait Abdelmalek-Lomenech\footnote[2]{Université Paris-Saclay, CNRS, CentraleSupélec, Laboratoire des Signaux et Systèmes.} , Julien Bect\footnotemark[2] \footnote[4]{Corresponding author. Email: \email{julien.bect@centralesupelec.fr}} , Emmanuel Vazquez\footnotemark[2]}

\usepackage{amsopn}

\newcommand{\ProbExpectLetter}[1]{\mathsf{#1}}

\renewcommand{\P}{\ProbExpectLetter{P}}
\newcommand{\Ps}{\P_{S}}

\newcommand{\Pn}{\P_n}

\newcommand{\E}{\ProbExpectLetter{E}}
\newcommand{\En}{\E_n}
\newcommand{\Enn}{\E_{n+1}}

\newcommand{\X}{\mathds{X}}
\renewcommand{\S}{\mathds{S}}
\newcommand{\Uset}{\mathds{U}}
\newcommand{\Rset}{\mathds{R}}

\newcommand{\Gf}{\Gamma(f)}

\newcommand{\Gxi}{\Gamma(\xi)}

\newcommand{\Hn}{\mathcal{H}_n}
\newcommand{\Hnn}{\mathcal{H}_{n+1}}
\newcommand{\Hnr}{\mathcal{H}_{n+r}}

\newcommand{\In}{\mathcal{I}_n}
\newcommand{\mn}{\mu_n}
\newcommand{\sn}{\sigma_n}
\newcommand{\hGn}{\widehat\Gamma_n}

\newcommand{\hGnr}{\widehat\Gamma_{n+r}}
\newcommand{\tGn}{\tilde\Gamma_n}

\newcommand{\tGnr}{\tilde\Gamma_{n+r}}
\newcommand{\tS}{\widetilde\S}
\newcommand{\tX}{\widetilde\X}
\newcommand{\tUset}{\widetilde{\mathds{U}}}

\newcommand{\dx}{\mathrm{d}x}

\newcommand{\du}{\mathrm{d}u}
\newcommand{\dPs}{\mathrm{d}\Ps}
\newcommand{\one}{\mathds{1}}

\DeclareMathOperator{\Var}{Var}

\newlength{\toto}

\usepackage{psfrag}

\ifpdf
\hypersetup{
  pdftitle={Bayesian Active Learning of (small) Quantile Sets through Expected Estimator Modification},
  pdfauthor={R. Ait Abdelmalek-Lomenech, J. Bect, E. Vazquez}
}
\fi

\usepackage[normalem]{ulem}

\begin{document}

\maketitle

\begin{abstract}
  Given a multivariate function taking deterministic and uncertain inputs, we consider the problem of estimating a quantile set: a set of deterministic inputs for which the probability that the output belongs to a specific region remains below a given threshold. To solve this problem in the context of expensive-to-evaluate black-box functions, we propose a Bayesian active learning strategy based on Gaussian process modeling. The strategy is driven by a novel sampling criterion, which belongs to a broader principle that we refer to as \emph{Expected Estimator Modification} (EEM). More specifically, the strategy relies on a novel sampling criterion combined with a sequential Monte Carlo framework that enables the construction of batch-sequential designs for the efficient estimation of small quantile sets. The performance of the strategy is illustrated on several synthetic examples and an industrial application case involving the ROTOR37 compressor model.
\end{abstract}



\section{Introduction}

Numerous scientific or industrial fields use \emph{numerical simulators} to model complex physical systems. Given such a numerical model, a fundamental challenge in uncertainty quantification and risk management is to estimate the set of input parameters for which the outputs satisfy prescribed properties. This problem is often called \emph{set inversion}~\citep{jaulin:1993:automatica} in the literature. A typical example of such a problem is the estimation of the excursion set 
\begin{align*}
\Lambda(f) = \{u \in \Uset \, : \, f(u) \ge T\},
\end{align*}
where $f\,:\, \Uset \mapsto \Rset$ represents a deterministic numerical simulator and $T \in \Rset$ is a given threshold.

\medbreak

Recently, formulations of the set inversion problem in which the simulator takes two types of inputs—deterministic and uncertain (or stochastic)—have been introduced. Such robust formulations arise in various applications, including nuclear safety \citep{chevalier:2013:phdthesis, marrel:2022:icscream}, %
flood defense optimization \citep{richet:2019:inversion}, %
and pollution control systems \citep{elamri:2023:set-inversion}.

\medbreak

Following the suggestions of \cite{leriche:2008:hdr} and \cite{richet:2019:inversion}, and the preliminary work of \cite{ait:2024:qsi}, we focus in this article on a specific robust formulation of the set inversion problem, called \emph{quantile set inversion} (QSI). Given a continuous black-box function
$f \, : \, \Uset = \X\times\S \to \Rset^q$, where $\X$ and~$\S$ are bounded
subsets of $\Rset^{d_\X}$ and $\Rset^{d_\S}$, corresponding to the
sets of possible values for the deterministic and uncertain input variables, our objective is to estimate the \emph{quantile set}
\begin{equation}
  \label{eem:eq:gammaf}
  \begin{split}
  \Gf & = \{x\in\X\,:\,\P(f(x,S) \in C) \le \alpha\}\\
     & = \left\{x\in\X \, : \, \int_\S \mathds{1}_C(f(x,s))\,\dPs(s) \le \alpha\right\},
  \end{split}
\end{equation}
where~$S$ is a random vector with known distribution~$\Ps$ on~$\S$, $C\subset\Rset^q$ is a given region of the output domain, and $\alpha \in (0,1)$.

\begin{remark}
By convention, when the simulator is stochastic and its evaluation yields noisy observations, $f(x,s)$ represents the true value of the underlying function at the point $(x,s)$, i.e., the expected output conditional on $(x,s)$.
\end{remark}
\medbreak

Our goal is thus to estimate the set of deterministic inputs $x \in \X$ for which the output $f(x,S)$ falls outside a critical region $C$---e.g., a failure or rejection domain---with sufficiently high probability under $\Ps$. %
This formulation arises naturally in \emph{Reliability-Based Design Optimization (RBDO)} \citep[see, e.g.,][]{dubourg:2011:rbdo,janusevskis:2013:simultaneous,hawchar:2018:rbdo}, where the goal is to solve an optimization problem of the form
\begin{align*}
&\underset{x\in\X}{\min}\quad h(x)\\
&\textrm{s. t.} \quad \P(f(x,S) \le 0) \le \alpha,
\end{align*} 
while identifying feasible solutions, i.e., inputs that violate the constraints with probability at most $\alpha$. With our notations, the
constraints are violated when~$f(x, s)$ belongs to the region~$C$.

\medbreak

A simple example (from \cite{ait:2024:qsi}) illustrating the QSI problem is shown in \autoref{eem:fig:example_qsi}, with $d_\X = d_\S = 1$, $C = (-\infty, 7.5]$, and $\alpha = 0.05$. The domain of uncertain inputs $\S$ is equipped with a rescaled Beta distribution that concentrates probability mass on large values. In this example, we observe that $\Gf$ consists of two relatively small, disjoint subsets of the deterministic input space $\X$. This illustrates that accurately estimating the full boundary of $f^{-1}(C)$ is unnecessary for approximating $\Gf$.

\begin{figure}[h]
  \newcommand{\ticksize}{\scriptsize}
  \psfrag{-0.02} {\ticksize $-0.02$}
  \psfrag{0}     {\ticksize $0$}
  \psfrag{0.02}  {\ticksize $0.02$}
  \psfrag{0.025} {\ticksize $0.025$}
  \psfrag{0.04}  {\ticksize $0.04$}
  \psfrag{0.05}  {\ticksize $0.05$}
  \psfrag{0.06}  {\ticksize $0.06$}
  \psfrag{0.075} {\ticksize $0.075$}
  \psfrag{0.08}  {\ticksize $0.08$}
  \psfrag{0.1}   {\ticksize $0.1$}
  \psfrag{0.12}  {\ticksize $0.12$}
  \psfrag{0.14}  {\ticksize $0.14$}
  \psfrag{0.16}  {\ticksize $0.16$}
  \psfrag{0.125} {\ticksize $0.125$}
  \psfrag{0.15}  {\ticksize $0.15$}
  \psfrag{0.2}   {\ticksize $0.2$}
  \psfrag{0.25}  {\ticksize $0.25$}
  \psfrag{0.3}   {\ticksize $0.3$}
  \psfrag{0.35}  {\ticksize $0.35$}
  \psfrag{0.4}   {\ticksize $0.4$}
  \psfrag{0.5}   {\ticksize $0.5$}
  \psfrag{0.6}   {\ticksize $0.6$}
  \psfrag{0.7}   {\ticksize $0.7$}
  \psfrag{0.8}   {\ticksize $0.8$}
  \psfrag{0.9}   {\ticksize $0.9$}
  \psfrag{1}     {\ticksize $1$}
  \psfrag{5}     {\ticksize $5$}
  \psfrag{10}    {\ticksize $10$}
  \psfrag{15}    {\ticksize $15$}
  \psfrag{20}    {\ticksize $20$}
  \psfrag{25}    {\ticksize $25$}
  \psfrag{30}    {\ticksize $30$}
  \psfrag{40}    {\ticksize $40$}
  \psfrag{50}    {\ticksize $50$}
  \psfrag{60}    {\ticksize $60$}
  \psfrag{80}    {\ticksize $80$}
  \psfrag{100}   {\ticksize $100$}
  \psfrag{150}   {\ticksize $150$}
  \psfrag{200}   {\ticksize $200$}
  \psfrag{250}   {\ticksize $250$}
  \psfrag{300}   {\ticksize $300$}
  
  \psfrag{bigbigmiscbased}{\small misc.-based}
  \psfrag{bigbigvarbased}{\small var.-based}
  \psfrag{bigbigentrbased}{\small entr.-based}
  
  \psfrag{X}{$x$}
  \psfrag{S}{$s$}

  \centering
  \psfrag{0}{\tiny $0$}
  \psfrag{0.1}{\tiny $0.1$}
  \psfrag{0.2}{\tiny $0.2$}
  \psfrag{0.3}{\tiny $0.3$}
  \psfrag{0.4}{\tiny $0.4$}
  \psfrag{0.5}{\tiny $0.5$}
  \psfrag{0.6}{\tiny $0.6$}
  \psfrag{0.8}{\tiny $0.8$}
  \psfrag{1}{\tiny $1$}
  \psfrag{2}{\tiny $2$}
  \psfrag{-2}{\tiny $-2$}
  \psfrag{4}{\tiny $4$}
  \psfrag{6}{\tiny $6$}
  \psfrag{8}{\tiny $8$}
  \psfrag{5}{\tiny $5$}
  \psfrag{10}{\tiny $10$}
  \psfrag{12}{\tiny $12$}
  \psfrag{14}{\tiny $14$}
  \psfrag{15}{\tiny $15$}
  \psfrag{16}{\tiny $16$}
  \psfrag{18}{\tiny $18$}
  \psfrag{density}{\footnotesize \raisebox{-4pt}{density}}
  \psfrag{  boundaryboundaryboundary}{\footnotesize boundary of $f^{-1}(C)$}
  \psfrag{  Gamma}{\footnotesize $\Gamma(f)$}
  \psfrag{  indicator}{\footnotesize $\mathds{1}_{\Gf} (x)$}
  \psfrag{  proba}{\footnotesize $\P(f(x,S) \in C)$}
  \psfrag{  alpha}{\footnotesize $\alpha$}
  \includegraphics[width=\textwidth,height=6.5cm]{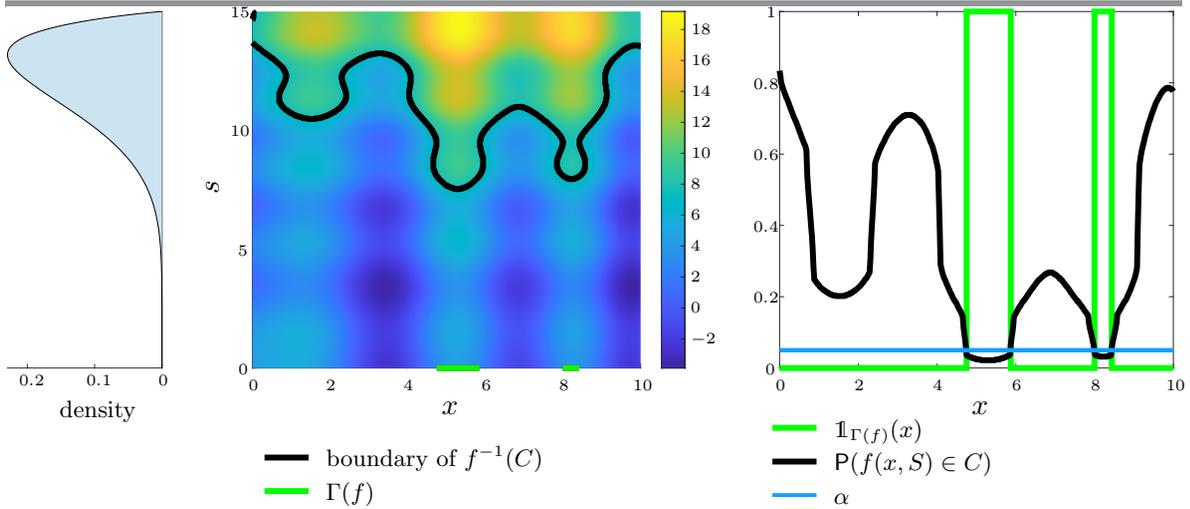}
  \caption{%
    Representation of a two-dimensional QSI problem. %
    Left: probability density function of~$\Ps$. %
    Middle: representation of~$f$, boundary of~$f^{-1}(C)$ and quantile set~$\Gamma(f)$
    associated to $C = (-\infty, 7.5]$ and $\alpha = 0.05$. %
    Right: indicator function of the quantile set,
    probability~$\P(f(x,S) \in C)$ and probability
    threshold~$\alpha$.}
  \label{eem:fig:example_qsi}
\end{figure}

\begin{remark}The QSI problem is equivalent to that of approximating the indicator function~$\one_{\Gf}$ through a classifier~$\X \to \{ 0, 1 \}$.
\end{remark}

\medbreak

In this article, we focus on the problem of estimating small quantile sets, that is, cases where $\lambda(\Gf) \ll \lambda(\X)$, with $\lambda$ denoting the Lebesgue measure on $\X$. Assuming an \emph{expensive-to-evaluate} underlying function, estimating $\Gf$ requires limiting the number of evaluations of~$f$. Recently, \cite{ait:2024:qsi} proposed a Bayesian method based on Gaussian process metamodeling and the Stepwise Uncertainty Reduction principle \citep[SUR, see, e.g.,][and references therein]{bect:2019:supermartingale} to address this task. Although this method performs efficiently on moderately difficult QSI problems, it suffers from high computational complexity. This limits its applicability to more challenging cases, particularly when estimating small quantile sets (relative to the entire deterministic input domain).

\medbreak

To address this issue, we propose a sequential Bayesian strategy based on a new approach called \emph{Expected Estimator Modification (EEM)}.
Coupled with a Sequential Monte Carlo framework (in the spirit of \cite{Li:2012:thesis,bect:2017:bss}), this strategy efficiently constructs batch designs of experiments for the estimation of such quantile sets.

\medbreak

This article is structured as follows: \autoref{eem:sec:framework} introduces the Bayesian framework considered, along with the main notations used throughout this work. \autoref{eem:sec:related} provides a concise review of previous work regarding QSI and, more generally, Bayesian methods for the estimation of excursion sets. In \autoref{eem:sec:EEM}, we introduce the EEM approach and derive a new criterion for the QSI problem. We also provide a proof of concept by applying it to moderately difficult problems. The problem of estimating small quantile sets is tackled in \autoref{eem:sec:small}, which describes a combination of the EEM strategy for QSI with a \emph{Sequential Monte Carlo (SMC)} framework, allowing the estimation of small quantile sets. The results of this strategy are illustrated on several synthetic examples and on the ROTOR37 compressor model \citep{reid:1978:rotor37}  in \autoref{eem:sec:numerical}. Finally, our conclusions are summarized in \autoref{eem:sec:conclusion}. 

\section{Bayesian framework \& notations}
\label{eem:sec:framework}

In the following, we consider an unknown continuous multivariate function $f:\Uset = \X\times\S \mapsto \Rset^q$. We adopt a Bayesian metamodeling approach by placing a Gaussian process (GP) prior $\xi \sim \text{GP}(\mu, k)$ on $f$, where $\mu$ and $k$ are respectively the mean and covariance functions of the GP. For further details on metamodeling and Gaussian processes, see \cite{rasmussen:2006:gpml,santner:2018:design}.

\medbreak

We sequentially select evaluation points using a Bayesian strategy. At each selected point $U_i \in \Uset$, we observe a possibly noisy evaluation $Z_i^{obs}$ defined by
\begin{align*} 
Z_i^{obs} = f(U_i) + \epsilon_i,
\end{align*}
\noindent
with $\epsilon_i \sim \mathcal{N}(0, \Sigma)$ independent and identically distributed. The noise covariance $\Sigma$ is a matrix of size $q \times q$, which may be identically zero.

\begin{remark}
The case $\Sigma = 0$ corresponds to a noiseless setting, as with a deterministic numerical simulator.
\end{remark}

\begin{remark}
The parameters of the GP and the noise covariance $\Sigma$ are assumed to be known in the methodological parts of this article. In \autoref{eem:sec:numerical}, the parameters of the GP are estimated using ReML \citep{stein1999interpolation}. Although this article includes no noisy examples, $\Sigma$ would be estimated jointly with the GP parameters when necessary.
\end{remark}

\medbreak

In the rest of this article, $\In = \mathcal{I}_0\cup\{(U_1, Z_1^{obs}),\dotsc, (U_n, Z_n^{obs})\}$ denotes the current information after $n_0+n$ evaluations of $f$, where $n_0$ is the size of the initial design and evaluations $\mathcal{I}_0$. For clarity, we do not distinguish between $\In$ and the sigma-algebra $\sigma(\In)$ it generates. We denote by $\Pn$ the conditional probability measure given $\In$, and $\En$ the associated expectation operator. Finally, $\mn$, $\sn$ and $k_n$ refer respectively to the posterior mean, standard deviation, and covariance of $\xi$ given $\In$.

\section{Related work}
\label{eem:sec:related}

Several formulations of the set inversion problem with uncertain inputs have been proposed \citep[see, e.g.,][]{marrel:2022:icscream,chevalier:2013:phdthesis,richet:2019:inversion,elamri:2023:set-inversion}, but the specific QSI problem has only recently been addressed \citep{ait:2024:qsi}. However, existing methods for the estimation of excursion sets can be applied in the space $\Uset = \X\times\S$ as described below.

\medbreak

For simplicity, we assume in this section that $f$ is scalar-valued and that $C = (-\infty, T)$. Defining the excursion set
\begin{align*} 
\Lambda(f) = \{(x,s)\in \X\times\S \, : \, f(x,s) \ge T\},
\end{align*}
we observe that the sets $\Gf$ and $\Lambda(f)$ are closely related:
\begin{align*} 
\Gf & = \{x \in \X \, : \, \P(f(x,S) < T) \le \alpha\}\\
    & = \{x \in \X \, : \, \P((x,S) \in \Lambda(f)) \ge 1 - \alpha\}.
\end{align*}
This relation implies that an accurate classification of the points $(x,s)$ with respect to membership in $\Lambda(f)$ allows, via Monte Carlo sampling, the classification of points $x$ with respect to $\Gf$.

\medbreak

\begin{remark}
Although the definition of $\Lambda(f)$ does not depend on the distribution $\Ps$, we retain the product notation $\X \times \S$ for clarity.
\end{remark}

\medbreak

A number of Bayesian methods have been developed to construct designs of experiments for estimating excursion sets such as $\Lambda(f)$. Most of these methods fall into one of two main categories: maximal uncertainty sampling strategies and Stepwise Uncertainty Reduction (SUR) strategies.

\medbreak

Maximal uncertainty strategies aim to sample the points $u \in \Uset$ that maximize a local uncertainty measure $G_n(u)$. For instance, the straddle heuristic \citep{bryan:2005:active} uses $G_n(u) = 1.96\sn(u) - |\mn(u) - T|$. For other strategies of this type, see \cite{cole:2023:entropy, echard:AK-MCS:2011, bichon:2008, ranjan:2008:contour, duhamel:2025:inversion}.

\medbreak

SUR strategies, by contrast, aim to minimize the \emph{expected future uncertainty} on $\Lambda(f)$. Given a global uncertainty measure $\Hn$ that depends on the current information $\In$, the next evaluation point is chosen by minimizing the criterion
\begin{align*}
J_n(u) = \En(\Hnn \mid U_{n+1} = u).
\end{align*}
Several choices for $\Hn$ have been proposed in the literature, including the integrated probability of misclassification \citep{bect:2012:sur_failureproba, chevalier:2014:fast_parallel_kriging}, the weighted integrated mean square error \citep{picheny:2010:adaptive_design}, and the integrated expected feasibility function \citep{duhamel:2023:sur-bichon}. More details on SUR methods, including sufficient conditions for convergence, can be found in \cite{bect:2019:supermartingale, stange:2024:SUR_theorie}.

\medbreak

Recently, \cite{ait:2024:qsi} proposed a method designed specifically for the QSI problem. The QSI-SUR strategy, based on the SUR principle, relies on the uncertainty measure
\begin{align*}
\Hn = \int_\X \min(\pi_n(x), 1 - \pi_n(x)) \, \dx,
\end{align*}
where $\pi_n(x) = \Pn(x \in \Gxi)$ is the probability that $x$ belongs to the random quantile set $\Gxi$ induced by the GP $\xi$. In the same spirit as \cite{bect:2012:sur_failureproba, chevalier:2014:fast_parallel_kriging}, the method selects the point that minimizes the expected future integrated probability of misclassification of $x$ with respect to $\Gxi$.

\medbreak

Empirical results indicate that, for a given evaluation budget, this criterion yields more accurate or at least comparable estimates of $\Gf$ than methods focusing on the estimation of $\Lambda(f)$. However, this improved performance comes at the cost of high computational complexity in its current implementation\footnote{\url{https://github.com/stk-kriging/contrib-qsi}}. Since no closed-form expression is available, the QSI-SUR criterion must be approximated using conditional GP simulations. As a result, it is difficult to use this method for constructing batch designs of experiments. Moreover, the criterion is approximated and optimized on a finite grid constructed using importance sampling. When $\lambda(\Gf) \ll \lambda(\X)$, this approach fails to allocate points effectively in relevant areas, which limits its suitability for estimating small quantile sets.

\section{QSI through Expected Estimator Modification}
\label{eem:sec:EEM}

\subsection{The QSI-EEM criterion}
\label{eem:sec:qsi-eem}

As discussed in \autoref{eem:sec:related}, the QSI-SUR criterion suffers from high computational complexity. To address instances involving small quantile sets, we propose in this section a computationally lighter approximation, based on a new family of criteria, termed \emph{Expected Estimator Modification} (EEM).

\medbreak

Let $J_n(u_{n+1}, \dotsc, u_{n+r})$ denote the QSI-SUR criterion for a batch of size $r$:
\begin{equation}
J_n(u_{n+1}, \dotsc, u_{n+r}) = \En\left( \Hnr \,\middle|\, U_{n+i} = u_{n+i},\ i=1,\dotsc,r \right),
\end{equation}
where $\Hnr = \int_\X \min(\pi_{n+r}(x), 1 - \pi_{n+r}(x))\,\dx$, and $\pi_{n+r}(x) = \P_{n+r}(x \in \Gxi)$. The following equivalence holds:

\medbreak

\begin{proposition}
\label{eem:prop:equiv_sur_gain}
Let $\Gamma^{\mathrm{B}}_n$ denote the Bayesian estimator of $\Gf$ given $\In$, defined by
\begin{equation}
\Gamma^{\mathrm{B}}_n = \bigl\{x \in \X \, : \, \pi_n(x) > \frac{1}{2} \bigr\},
\end{equation}
and let $\Delta$ denote the symmetric difference operator, defined for sets $A$ and $B$ as $A \Delta B = (A \setminus B) \cup (B \setminus A)$. Then
\begin{equation}
\underset{u\in\Uset^r}{\arg\min}\, J_n(u) = \underset{u\in\Uset^r}{\arg\max} \, G_n(u),
\end{equation}
where
\begin{equation}
G_n(u) = \int_\X \En\left( \left|2\pi_{n+r}(x) - 1\right| \cdot \mathds{1}_{\Gamma^{\mathrm{B}}_{n+r} \Delta \Gamma^{\mathrm{B}}_n}(x) \,\middle|\, U_{n+i} = u_{n+i},\ i=1,\dotsc,r \right) \,\dx.
\end{equation}
\end{proposition}

\begin{proof}
See \aref{eem:app:proof_sur_gain}.
\end{proof}

\medbreak

From \autoref{eem:prop:equiv_sur_gain}, we observe that minimizing the QSI-SUR criterion is equivalent to maximizing an expected weighted symmetric difference between Bayesian estimators. However, the computation of $G_n$ requires evaluating the posterior probabilities $\pi_n$ and $\pi_{n+r}$, both to determine the classification sets $\Gamma^{\mathrm{B}}_n$ and $\Gamma^{\mathrm{B}}_{n+r}$ and to compute the weights $|2\pi_{n+r}(x) - 1|$. These posterior quantities cannot be computed analytically and must be approximated using conditional simulations of the Gaussian process, which incurs significant computational cost.

\medbreak

To reduce the computational cost, we propose two simplifications. First, we replace the weights $|2\pi_{n+r}(x)-1|$, which lie in $[0,1]$, by the constant value~$1$. Second, we substitute the Bayesian set estimator $\Gamma_n^{\mathrm{B}}$ by a more tractable surrogate $\hGn$ (see below). We then obtain a rough approximation of $G_n$ as
\begin{equation}
  G_n(x) \approx\En(d(\hGnr,\hGn)\, | \, U_{n+i} = u_{n+i}, \, i=1,\dotsc,r),
\end{equation}
where
\begin{align*}
  d(\hGnr,\hGn) = \lambda(\hGnr \Delta \hGn).
\end{align*}

The function $d(\cdot, \cdot)$ defines a pseudo-metric on the Lebesgue $\sigma$-algebra of~$\X$, known as the Fréchet–Nikodým–Aronszajn distance \citep{nikodym:1930:set_distance, marczewski:1958:set_distance}. Under this approximation, $G_n$ measures the expected change between two estimators of the target set, and can thus be interpreted as an instance of EEM.

\medbreak

Regarding the choice of estimators $\hGn$, we introduce a class of plug-in estimators obtained by replacing the process $\xi$ with its posterior mean $\mu_n$. Specifically, define
\begin{equation}
  \begin{aligned}
  \hGn & = \{x\in\X \, : \, \Pn(\mu_n(x,S) \in C) \le \alpha\}\\
      & = \bigl\{x\in\X \, : \, \int_\S \mathds{1}_C(\mu_n(x,s))\,\dPs(s) \le \alpha\bigr\}.
  \end{aligned}
\end{equation}
These estimators are particularly appealing from a computational perspective due to the following property:

\begin{proposition}
  \label{eem:prop:sample_estim}
  Let $\xi$ be a univariate GP, $Z$ be a standard Gaussian random variable and $\overset{\mathrm{d}}{=}$ denote equality in distribution. Given $\In$ and a set $U = \{U_{n+1},\dotsc, U_{n+r}\}$ of candidate points, the estimator $\hGnr$ defines a closed random set \citep{molchanov:2005:random_sets} such that:
\begin{equation} 
  \begin{split}
  \hGnr \overset{\mathrm{d}}{=} \left\{x \in \X \, : \, \int_\S \mathds{1}_C\left(\mu_n(x,s) + \kappa_n(x,s)Z\right)\, \dPs(s) \le \alpha\right\},
  \end{split}
\end{equation}
where the function $\kappa_n(x,s)$ is given by
\begin{align*}
  \kappa_n(x,s) = \sqrt{k_n(U, (x, s))\, \Sigma_n(U)\, k_n(U, (x, s))^t},
\end{align*}
with $\Sigma_n(U)$ the posterior covariance matrix of the vector $(\xi(U_{n+1}),\dotsc,\xi(U_{n+r}))$, and
\begin{align*}
  k_n(U, (x,s)) = (k_n(U_{n+1}, (x,s)),\dotsc, k_n(U_{n+r}, (x,s))).
\end{align*}
\end{proposition}

\begin{proof}
See \aref{eem:app:proof_approx}.
\end{proof}
\medbreak
\begin{remark}
  For simplicity of notations, we have assumed in \autoref{eem:prop:sample_estim} that $\xi$ is a univariate Gaussian process. However, this result can easily be generalized to multivariate GPs.
\end{remark}

Unlike the Bayesian estimator $\Gamma_n^{\mathrm{B}}$, which requires Monte Carlo sampling of the process $\xi$ under the posterior $\Pn$, the estimator $\hGn$ can be computed using simpler operations. Indeed, \autoref{eem:prop:sample_estim} shows that once $\mu_n$ has been computed, realizations of $\hGnr$—conditioned on a hypothetical set of evaluation points—can be generated by drawing a standard Gaussian random vector.
\medbreak

Based on the approximations above, we now define the \emph{QSI-EEM} criterion for a batch of size $r$:
\begin{equation}
\Psi_{n,r}(u_{n+1},\dotsc,u_{n+r}) = \En(\lambda(\hGnr\Delta\hGn)\, | \, U_{n+i} = u_{n+i}, i = 1,\dotsc,r).
\end{equation}
This criterion is to be maximized, and defines a batch Bayesian strategy where the evaluation points are selected as
\begin{equation}
(U_{n+1},\dotsc,U_{n+r}) \in \underset{(u_{n+1},\dotsc, u_{n+r}) \in \Uset^r}{\arg\max} \, \Psi_{n,r}(u_{n+1},\dotsc, u_{n+r}).
\end{equation}
Heuristically, the criterion $\Psi_{n,r}$ can be viewed through the lens of sensitivity analysis. It quantifies the expected modification of the estimator when new points are evaluated, and may be interpreted as a sensitivity index of the estimation with respect to the chosen batch \citep{baucells:2013:sensi,daveiga:2015:sensi}.

\begin{remark}
 The maximization of the criterion $\Psi_{n,r}$ defines a deterministic decision rule. In particular, the selected batch $(U_{n+1},\dotsc, U_{n+r})$ is measurable with respect to the current information $\In$, i.e., there exists a function $\phi_n$ such that $(U_{n+1},\dotsc, U_{n+r}) = \phi_n(\In)$.
\end{remark}

\subsection{Approximation of the QSI-EEM criterion}
\label{eem:sec:QSI-EEM-approx}

Although the QSI-EEM criterion is simpler than the original QSI-SUR criterion, it still lacks a closed-form expression. However, it can be efficiently approximated using a finite discretization of the input spaces.

\medbreak

Let $\tX$ be a finite set of points sampled uniformly from $\X$, and let $\tS = \{(s_j, w_S(s_j))\}_j$ be a weighted sample from $\S$, defining a discrete measure
\begin{align*}
\tilde{\P}_\S = \sum_{s \in \tS}w_S(s)\delta_{\{s\}},
\end{align*}
that approximates $\Ps$. Given the posterior mean $\mn$ computed from $\In$, we define the approximate deterministic estimator
\begin{equation}
\tGn = \Bigl\{x \in \tX \, : \, \sum_{s \in \tS}w_S(s)\mathds{1}_C(\mn(x,s)) \le \alpha\Bigr\}.
\end{equation}
In addition (assuming a univariate GP $\xi$ for simplicity), as in \autoref{eem:prop:sample_estim}, given $\In$ and a batch $U = (U_{n+1},\dotsc, U_{n+r})$, we have
\begin{equation}
\tGnr \overset{\mathrm{d}}{=} \Bigl\{x \in \tX \, : \, \sum_{s \in \tS}w_S(s)\mathds{1}_C(\mn(x,s)+\kappa_n(U, x,s)Z) \le \alpha \Bigr\},
\end{equation}
where $Z$ is a standard Gaussian random variable. 

\medbreak
Let $\tGnr(z)$ denote the set obtained for a realization $Z = z$. Using a Gauss-Hermite quadrature $(z_j, w_Z(z_j))$, the criterion $\Psi_{n,r}$ can be approximated by
\begin{equation}
  \label{eem:eq:psi_tilde}
  \tilde{\Psi}_{n,r}(u_{n+1},\dotsc,u_{n+r}) = \frac{1}{|\tX|}\sum_j w_Z(z_j)\sum_{x\in\tX}|\mathds{1}_{\tGn}(x)-\mathds{1}_{\tGnr(z_j)}(x)|,
\end{equation}
where $|\tX|$ denotes the cardinality of the design grid $\tX$.

\begin{remark} The quantity $\tilde{\Psi}_{n,r}$ can be
    interpreted as a weighted average $\ell_1$ distance between two
    binary classifiers over the discrete grid $\tX$: one corresponding
    to the current plug-in estimator $\tGn$, and the other to a
    perturbed version $\tGnr(z_j)$ based on a simulated posterior
    modification. It quantifies the average classification
    disagreement induced by a potential batch of evaluations.
\end{remark}

\subsection{Results on moderately difficult problems}
\label{eem:sec:sanity_check}
Before addressing the challenging case of small quantile sets, we assess the relevance of the proposed criterion on moderately difficult QSI problems.

\medbreak

As a benchmark, we compare the accuracy of the proposed strategy with that of the QSI-SUR criterion \citep{ait:2024:qsi}, as well as with several alternative methods targeting the estimation of $\Lambda(f)$ in the joint input space $\X\times\S$ (see \autoref{eem:sec:related}). The experiments are conducted on synthetic test cases from \cite{ait:2024:qsi}, using the same implementation\footnote{\url{https://github.com/stk-kriging/contrib-qsi}} and parameters\footnote{\url{https://github.com/stk-kriging/qsi-paper-experiments}}, with Matlab v2022a and the STK toolbox v2.8.1 \citep{STK}. Implementation details and test case descriptions are provided in \aref{eem:app:details_implementation} and \aref{eem:app:sec4_examples}, respectively.

\medbreak

\begin{figure}
  \begin{adjustwidth}{-1cm}{-1cm}
  \centering
  \includegraphics[width=\toto]{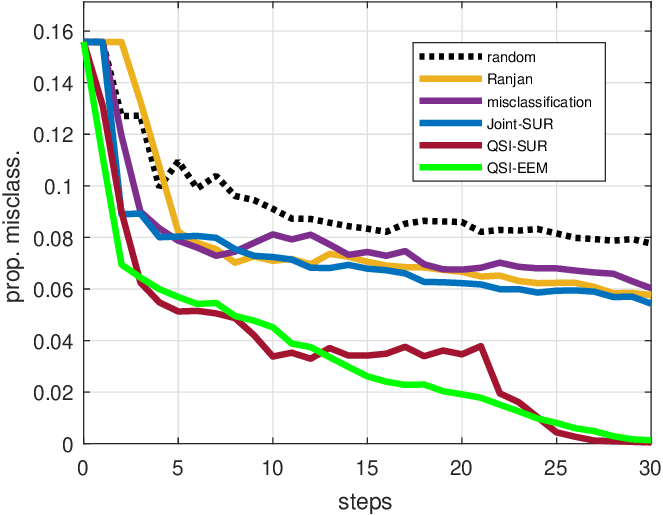}
  \hspace*{5mm}
  \includegraphics[width=\toto]{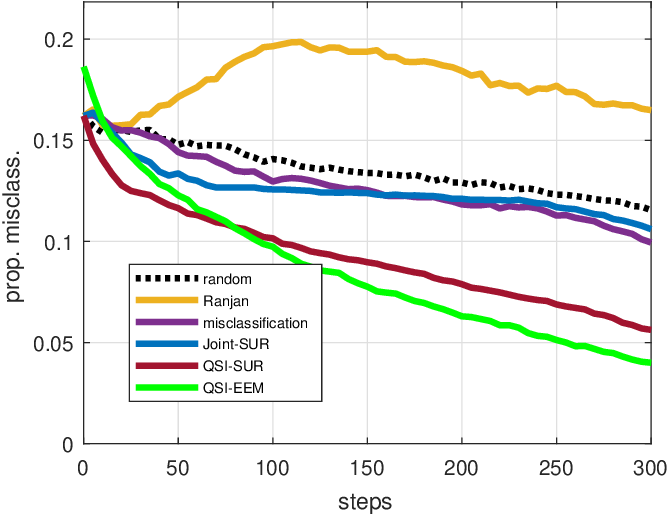}\\[5mm]
  \includegraphics[width=\toto]{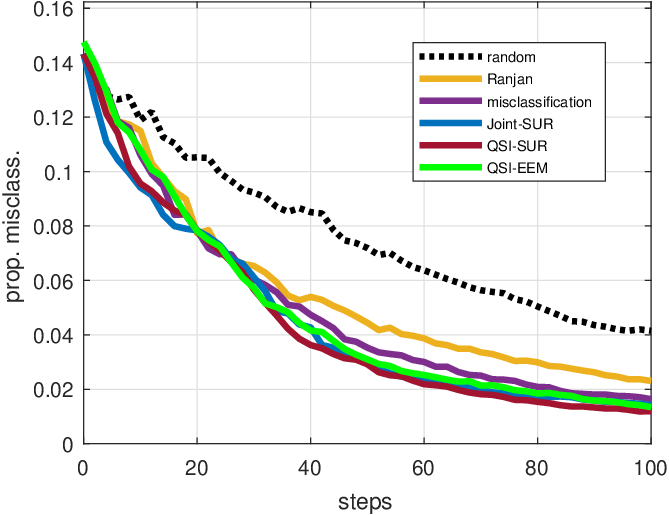}
  \caption{%
    Median of the proportion of misclassified points vs.\;number of
    steps, for 100~repetitions of the algorithms on the test
    functions~$f_1$ (top left), $f_2$ (top right) and~$f_3$
    (bottom). }
  \label{eem:fig:results_comp_EEM_SUR}
  \end{adjustwidth}
\end{figure}

\begin{figure}
  \begin{adjustwidth}{-1cm}{-1cm}
  \centering
  \includegraphics[width=\toto]{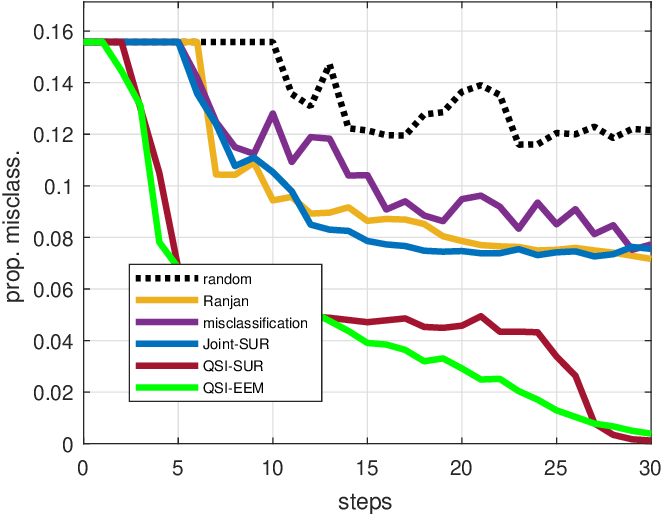}
  \hspace*{5mm}
  \includegraphics[width=\toto]{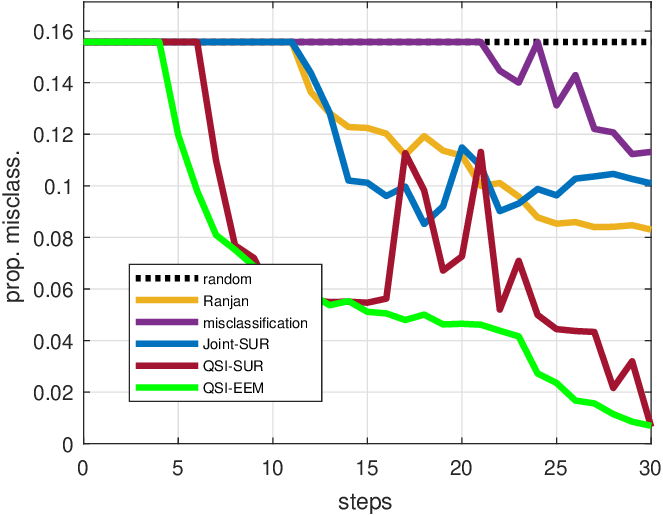}
  \caption{Quantiles of order $0.75$ (left) and $0.95$ (right) of the
    proportion of misclassified points vs.\;number of steps, for
    100~repetitions of the algorithms on the test function~$f_1$. }
  \label{eem:fig:results_comp_f1_quant_EEM_SUR}
  \end{adjustwidth}
\end{figure}

\begin{table}[ht]
  \centering
\begin{tabular}{l|c | c | c | c| c |}
  \cline{2-6}
  & Ranjan        & misclass.        & Joint-SUR      & QSI-SUR & QSI-EEM \\
  \hline
  \multicolumn{1}{|l|}{$f_1$}    &0.15 & 0.14 & 3.74 &  3.77 & 0.99\\
  \hline
  \multicolumn{1}{|l|}{$f_2$}    & 0.29 & 0.23 & 7.01  & 5.91 & 2.58  \\
  \hline
  \multicolumn{1}{|l|}{$f_3$}    & 0.30 & 0.23 & 6.57 & 5.35  & 2.44\\
  \hline
\end{tabular}
\caption{%
  Runtime (in seconds) to complete the first step. %
  Average over $10$~runs.}
\label{eem:table:first_time}
\end{table}

\begin{table}[ht]
  \centering
  \begin{tabular}{l|c | c | c | c| c|}
    \cline{2-6}
    & Ranjan        & misclass.        & Joint-SUR      & QSI-SUR & QSI-EEM \\
    \hline
    \multicolumn{1}{|l|}{$f_1$}    &1 & 0.81 & 19.26 &  21.94 & 5.67\\
    \hline
    \multicolumn{1}{|l|}{$f_2$}    & 1 & 0.84 & 4.96  & 9.45 & 4.27 \\
    \hline
    \multicolumn{1}{|l|}{$f_3$}    & 1 & 0.83 & 10.02 & 11.05 & 4.68 \\
    \hline
  \end{tabular}
  \caption{%
    Normalized total runtime (in seconds). %
    Average over $10$~runs.}
  \label{eem:table:total_time}
\end{table}

When comparing criteria under a similar implementation, the proposed strategy offers a clear computational advantage (see \autoref{eem:table:first_time} and \autoref{eem:table:total_time}). Most of the computational cost in our implementation---shared with that of QSI-SUR---arises from the sampling of the approximation grid $\tX\times\tS$ via importance sampling, using a density derived from the misclassification probability (see \aref{eem:app:details_implementation} for details). In particular, selecting points in $\X$ requires costly conditional Gaussian simulations to approximate $\pi_n(x) = \Pn(x \in \Gamma(\xi))$. As shown in \autoref{eem:fig:comp_time}, evaluating the criterion on a fixed approximation grid rather than a sampled one leads to a computational complexity of $O(M^2)$, where $M = |\tX \times \tS|$. In contrast, the QSI-SUR criterion exhibits $O(M^3)$ complexity. This difference stems from the fact that QSI-SUR requires conditional Gaussian simulations involving a Cholesky factorization of the full covariance matrix over $\tX \times \tS$, whereas QSI-EEM avoids this operation entirely.

\begin{figure}
  \center
\includegraphics[width=14cm]{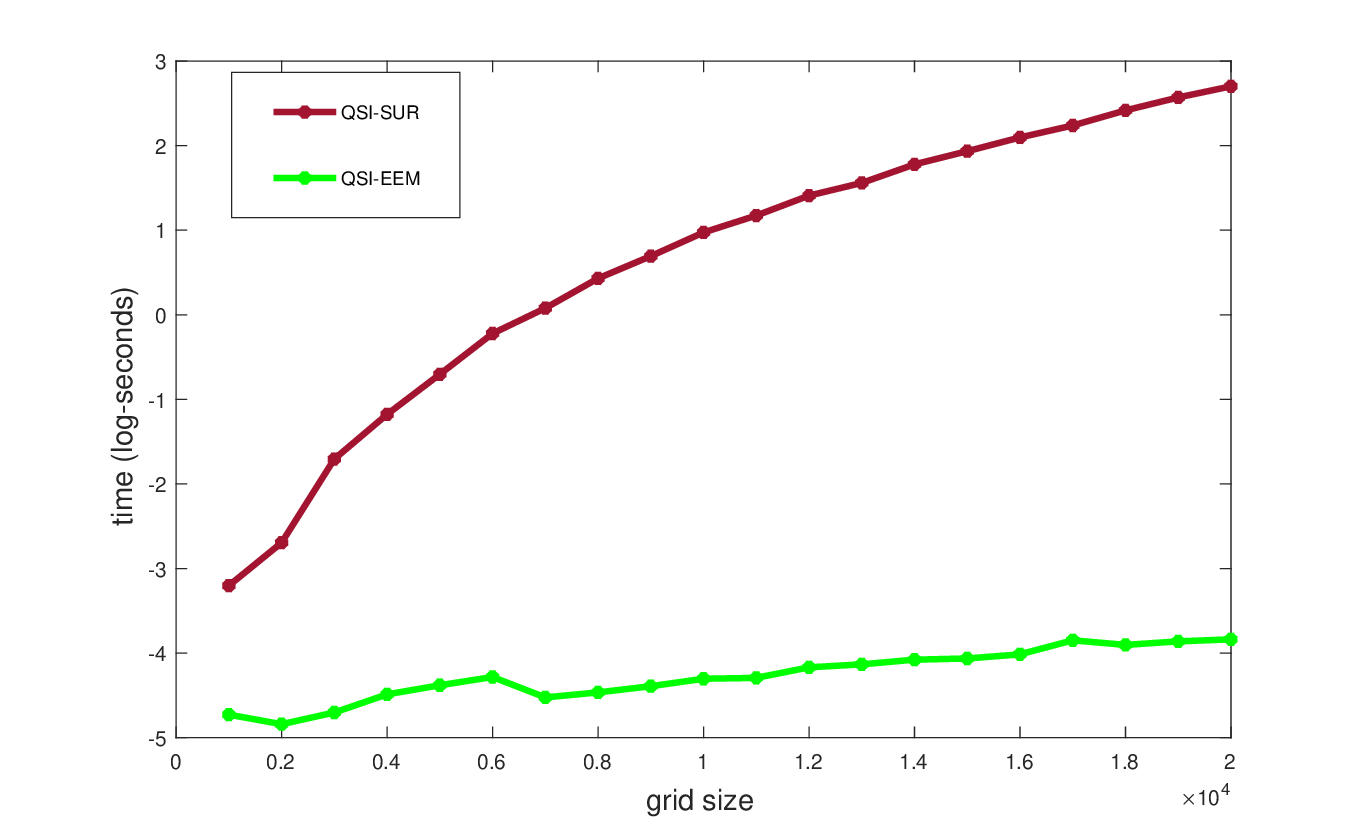}
\caption{Comparison of the time (in log-seconds) needed to evaluate the criterion at one points given the size of the approximation grid $|\tX\times\tS|$.}
\label{eem:fig:comp_time}
\end{figure}

\medbreak

More results regarding these experiments can be found in \autoref{eem:SM:sec:results-experiences}.

\subsection{Discussion on the concept of EEM and its link with the SUR principle}
\label{eem:sec:discuss-EEM}

We now discuss how the concept of EEM extends naturally to a broader class of Bayesian design strategies.

\medbreak

As before, let $f$ be an unknown function modeled by a Gaussian process $\xi$, and let $\theta(f) \in \Theta$ denote a quantity or object of interest. Let $\{\hat{\theta}_n\}$ be a sequence of estimators of $\theta(f)$, and let $d : \Theta \times \Theta \to \Rset^+$ be a divergence. An EEM strategy (with batch size $1$) targeting $\theta(f)$ selects the next evaluation point as

$$
U_{n+1} \in \underset{u \in \Uset}{\arg\max} \, \En(d(\hat{\theta}_{n+1}, \hat{\theta}_n) \mid U_{n+1} = u).
$$

Several Bayesian design strategies proposed in the literature can be reinterpreted within the EEM framework. This includes, in particular, several instances of the SUR principle.

\medbreak

For example, in Bayesian optimization, the classical Expected Improvement (EI) criterion \citep{mockus:1978:seeking}
$$
U_{n+1} \in \underset{u \in \Uset}{\arg\max}\, \En((\xi(u)-M_n)_+),
$$
with $M_n = \max\{\xi(u) \, : \, \sigma_n(u) = 0\}$, can be seen as an EEM strategy with $\hat{\theta}_n = M_n$ and divergence $d(\hat{\theta}_{n+1}, \hat{\theta}_n) = |\hat{\theta}_{n+1} - \hat{\theta}_n|$.

\medbreak

In the context of metamodel fitting, the Integrated Mean Square Error (IMSE) strategy
$$
U_{n+1} \in \underset{\breve{u} \in \Uset}{\arg\min}\, \int_\Uset\sigma_{n+1}^2(u) \, \du \quad \mid \quad U_{n+1} = \breve{u}
$$
can be expressed as an EEM strategy with $\hat{\theta}_n = \mu_n$ and divergence
$$
d(\hat{\theta}_{n+1}, \hat{\theta}_n) = \int_\Uset(\hat{\theta}_{n+1}(u)-\hat{\theta}_n(u))^2\, \du.
$$

\medbreak

A similar reformulation applies to the estimation of a probability of failure $\P_\Uset(\Lambda(f))$, where $\Lambda(f)$ is a critical excursion set and $\P_\Uset$ a probability measure on $\Uset$. The SUR-type criterion
$$
U_{n+1} \in \underset{u \in \Uset}{\arg\min}\, \En\left(\Var_{n+1}(\P_\Uset(\Lambda(\xi))) \mid U_{n+1} = u\right),
$$
where $\Lambda(\xi)$ is the random excursion set under the GP model, is equivalent to an EEM strategy with
$$
\hat{\theta}_n = \En(\P_\Uset(\Lambda(\xi))), \quad 
d(\hat{\theta}_{n+1}, \hat{\theta}_n) = (\hat{\theta}_{n+1} - \hat{\theta}_n)^2.
$$

\medbreak

The mathematical derivations of these equivalences are postponed to \aref{eem:app:eq_eem_sur}.

\medbreak
\begin{remark}
The estimator used in our formulation of expected improvement, following \cite{bect:2019:supermartingale}, differs from the more common one $\tilde{M}_n = \max\{f(U_i) \, : \, i\le n\}$. This modification handles the degenerate Gaussian case, where it may occur that $M_n > \tilde{M}_n$, allowing the strategy to be formulated consistently within the SUR framework.
\end{remark}

\medbreak

However, establishing EEM as a general-purpose principle requires some caution. Consider the simple case where $f = \theta$ is constant, defined over a two-point domain $\Uset = \{u_1, u_2\}$. Assume noisy observations of the form $Z_i = f(U_i) + \delta(U_i)\epsilon_i$, where $\epsilon_i \sim \mathcal{N}(0, 1)$ and $\delta > 0$. If we estimate $\theta$ with $\hat{\theta}_n = \frac{1}{n}\sum_{i=1}^n{Z_i}$ and use the divergence $d(\hat{\theta}_{n+1}, \hat{\theta}_n) = (\hat{\theta}_{n+1}-\hat{\theta}_{n})^2$, then
$$
\En(d(\hat{\theta}_{n+1}, \hat{\theta}_n) \mid U_{n+1} = u) = \frac{1}{(n+1)^2}\left[\E((\hat{\theta}_n-\theta)^2) + \delta^2(u)\right].
$$
As a result, the next point is selected by
$$
U_{n+1} \in \underset{u \in \{u_1, u_2\}}{\arg\max} \, \delta^2(u).
$$

In other words, the criterion favors the most noisy evaluations to estimate the mean. This illustrates that EEM strategies are not universally applicable: additional structure or constraints on the estimator $\hat{\theta}_n$ or the divergence $d$ may be necessary.

\section{Estimation of small quantile sets}
\label{eem:sec:small}

\subsection{Motivation and principle}

As discussed in \autoref{eem:sec:QSI-EEM-approx}, the QSI-EEM criterion does not admit a closed-form expression, to the best of our knowledge. As a result, both the Gaussian expectation and the integral arising in the criterion must be approximated. As shown in \autoref{eem:sec:sanity_check}, the QSI-EEM strategy, combined with the importance sampling scheme from \cite{ait:2024:qsi}, provides accurate estimations of quantile sets of moderate size.

\medbreak

However, this approach becomes ineffective when the target quantile set is small. Consider, for instance, a quantile set with size $\lambda(\Gf)\sim 10^{-8}\cdot\lambda(\X)$. 
A first difficulty is that for such small sets, it is common for the prior $\P_0$ to satisfy $\P_0(x \in \Gxi) \approx 0$ for all $x \in \X$. As a result, under this approximation, $\tilde{\Psi}_{0,r}$ becomes numerically zero for all $x \in \X$, making it unusable as a sampling criterion.
Second, an accurate approximation of the criterion requires a set $\tX_n \subset \X$ containing points whose classification may change following a new batch of evaluations. Given the very small size of $\Gf$ (and thus of $\hGn$), obtaining such points via importance sampling would require an impractically large number of approximations of the conditional probability $\pi_n$.

\medbreak

To address these issues, we draw inspiration from multilevel splitting \citep{Kahn:1951:Splitting} and subset simulation \citep{Au:2001:subset}, and introduce a Sequential Monte Carlo (SMC) framework \citep{delmoral:2006:sequential_mc_samplers,cerou:2012:smc}. Following the ideas of \cite{Li:2012:thesis, bect:2017:bss}, the core idea is to estimate a sequence of nested quantile sets:
$$
\Gamma^1(f) \supset \cdots \supset \Gamma^k(f) \supset \cdots \supset \Gamma^K(f) = \Gamma(f).
$$

Combining the SMC framework with the QSI-EEM criterion yields a strategy alternating between two phases: an \emph{estimation} phase and a \emph{sampling} phase. 
At iteration $k$, starting from $n$ evaluations of $f$, we define an intermediate target $\Gamma^k(f)$ and add evaluation points using the QSI-EEM criterion until the set is estimated with sufficient accuracy. Next, in the resampling phase, the set of points in $\X$ is concentrated around $\Gamma^k(f)$. These two phases are repeated until a satisfactory approximation of the target quantile set $\Gf$ is achieved. The next section provides details on this procedure.

\subsection{Coupling sequential Monte Carlo and QSI-EEM acquisition}

The strategy begins by conditioning the GP prior $\xi$ using $n_0$ points from $\Uset$, typically chosen as a space-filling design.As discussed in \autoref{eem:sec:QSI-EEM-approx}, a finite subset $\tS \subset \S$ with associated weights is selected to approximate the measure $\Ps$. An initial finite set $\tX_0$ is sampled from $\X$. The following steps are repeated for each stage $k = 1, \dotsc, K$. Let $n_k = n_0 + \sum_{j=1}^k N_j r$ denote the total number of evaluations performed up to stage $k$, where $N_j$ is the number of batches sampled during stage $j$. Assume that $\Gamma^{k-1}(f)$ has been satisfactorily estimated after $n_{k-1}$ evaluations of $f$. By convention, we set $\Gamma^0(f) = \X$.

\medbreak

The first step is to define a new intermediary quantile set $\Gamma^{k}(f)$ to be estimated. A decreasing sequence of quantile sets can be defined by setting
\begin{align*}
\Gamma^{k}(f) = \{x \in \X \, : \, P(f(x,s) \in C_ {k}) \le \alpha\},
\end{align*}
where $(C_k)_k$ denotes an increasing sequence of subsets of $\Rset^q$.
For now, assume that the sequence $(C_k)_k$ is given. The details regarding the construction of this sequence are given in \autoref{eem:sec:density}.

\medbreak

\begin{remark}
From this point onward, all criteria and estimators refer to their approximated versions, constructed according to the scheme described in \autoref{eem:sec:QSI-EEM-approx}.
\end{remark}

\medbreak

Given a batch size $r$, $N_k$ batches of $r$ evaluation points are then sampled sequentially using the QSI-EEM criterion targeting $\Gamma^k(f)$:
\begin{equation}
\Psi^k_{n,r}(u_{n+1},\dotsc,u_{n+r}) = \En(\lambda(\hGnr^k\Delta\hGn^k)\, | \, U_{n+i} = u_{n+i}, i = 1,\dotsc,r).
\end{equation}
This criterion is approximated on $\tUset_k = \tX_k \times \tS$, and batches are added until a suitable estimator $\widehat{\Gamma}^k_{n_k}$ of $\Gamma^k(f)$ is obtained, where
\begin{equation}
  \widehat{\Gamma}^k_{n_k} = \{x \in \X \, : \P(\mu_{n_k}(x,S) \in C_k) \le \alpha\}.
\end{equation} 
The number of new batches $N_k$ is not fixed in advance; it is determined by a stopping condition described in \autoref{eem:sec:stop}.

\medbreak

Once this phase is complete, the algorithm proceeds to the sampling step.
Let $q_{n_k}^{k}$ be a target density, built from the information $\mathcal{I}_{n_k}$, and concentrated around $\widehat{\Gamma}^k_{n_k}$ and its boundary.
Particles in $\tX_k$ are resampled using residual resampling \citep{douc:2005:resampling, kuptametee:2022:resampling} with weights proportional to $q_{n_k}^k(x)/q_{n_{k-1}}^k(x)$, and then moved using an adaptive Metropolis-Hastings algorithm (see \aref{eem:app:MH}) to define the next set $\tX_{k+1}$. A new target set $\Gamma^{k+1}(f)$ is then defined as described in \autoref{eem:sec:density}.

\begin{remark}
Note that an alternative approach consists in defining decreasing quantile sets using a decreasing sequence of thresholds $(\alpha_k)_k$. However, this is inconvenient in practice: since the quantization of $\Ps$ has finite support by construction, the empirical average $\sum_{s \in \tS}w_S(s)\mathds{1}_C(\mn(x,s))$ takes only finitely many values. This creates difficulties when attempting to calibrate a target quantile set of given relative size using the density described in \autoref{eem:sec:density}.
\end{remark}
\medbreak
A summary of the algorithm is given in \autoref{eem:table:code}. The choice of the target density, the construction of the sequence of quantile sets, and the stopping criterion are discussed in the next subsection. 
\begin{table}
 
  \caption{Summary of the strategy}
  \label{eem:table:code}

  \noindent\rule{\textwidth}{0.4pt}
  \begin{enumerate}
  \item {\bf Initialization}
    \begin{enumerate}
      \item Set a batch size $r$.
      \item Condition $\xi$ on an initial space-filling design of size $n_0$.
      \item Choose a collection $\tS \in \S$ and associated weights.
    \item Sample a finite set $\tX_1$ in $\X$.
    \end{enumerate}
    \medbreak
    
  \item {\bf Repeat} until a satisfactory approximation of $\Gf$ is reached.
    \begin{enumerate}
    \item {\bf Estimation} phase
      \begin{itemize}
        \item Set $\Gamma^k(f)$ according to \autoref{eem:sec:density}.
        \item Add $N_k$ new batches of evaluations points according to the approximated QSI-EEM criterion on $\tUset_k=\tX_k\times\tS$.
        \item Stop when stopping condition is satisfied (see \autoref{eem:sec:stop}).
      \end{itemize}
    \item {\bf Sampling} phase
      \begin{flushleft}
        \begin{itemize}
        \item Resample points in $\tX_k$ according to $\sum_{x \in \tX_k}\frac{q_{n_{k+1}}^k(x)}{q_{n_k}^k(x)}\delta_x$ (residual resampling).
        \item Define $\tX_{k+1}$ by applying the adaptive Metropolis-Hastings algorithm (\aref{eem:app:MH}) to the resampled points, with target density $q_{n_{k+1}}^k$.
        \end{itemize}
      \end{flushleft}
    \item Increment~$k$.
    \end{enumerate}
    
  \end{enumerate}

  \noindent\ignorespaces\rule{\textwidth}{.4pt}%

\end{table}

\subsection{Target density and choice of the sets $C_k$}
\label{eem:sec:density}

Following \cite{dubourg:2013:metamodel_IS, Li:2012:thesis, bect:2017:bss}, a natural candidate for the target density is
\begin{align*}
q_n^k(x) \propto \Pn(x\in\Gamma^k(\xi)),
\end{align*}
where $\propto$ denotes proportionality.
However, using this density would eliminate the computational advantage over \cite{ait:2024:qsi}, as discussed in \autoref{eem:sec:sanity_check}. Indeed, it lacks a closed-form expression and must be approximated via Monte Carlo simulation of conditional sample paths of $\xi(x, \cdot)$, which is computationally costly.

\medbreak

To retain simplicity and interpretability, we instead propose to use
\begin{align*}
q_n^k(x) \propto \mathds{1}_{\hGn^{k,+}}(x),
\end{align*}
where $\hGn^{k,+}$ is a relaxed estimation of $\Gamma^k(f)$ at step $n$, satisfying $\hGn^k \subset \hGn^{k,+}$.

\medbreak

Such a relaxed estimation can be easily constructed from the posterior quantile functions $\mu_n^{\beta}(x,s) = \mu_n(x,s) + \sn\Phi^{-1}(\beta)$, with $0 <\beta < 1$. For instance, assuming $C_k = (-\infty, T_k)$, one can choose $\beta > 1/2$. Then, setting
\begin{align*}
  \hGn^{k,+} = \{x \in \X \, : \, \P(\mu_n^{\beta}(x,S) \in C_k) \le \alpha\},
\end{align*}
we obtain $\hGn^k \subset \hGn^{k,+}$, since $\mu_n^{\beta}(x,s) \ge \mu_n(x,s)$ for all $(x,s)$. 
In practice, $\beta$ is calibrated so that $\lambda(\hGn^{k,+}) \approx \kappa \cdot \lambda(\hGn^k)$ for a given constant $\kappa$.

\medbreak

\begin{remark}
Using $q_n^k$ as target density turns the Metropolis-Hastings random walk into a constrained Gaussian random walk on $\X$.
\end{remark}

\medbreak

The calibration of the sequence $(C_k)_k$ is straightforward. Given $n$ evaluations of $f$, and a set of points $\tX_k = \{x_k^i, \, i = 1,\dotsc,M\}$ obtained after resampling and displacement via the MCMC scheme (with target density $q_n^k$) described in \aref{eem:app:MH}, we propose to choose $C_k$ such that
\begin{equation}
\frac{1}{M}\sum_{i=1}^M \mathds{1}_{\hGn^{k,+}}(x^i_{k}) \approx\rho,
\end{equation}
with $\rho \in (0,1)$. In other words, the next quantile set is defined so that a proportion $\rho$ of our particles in $\X$ lies in the relaxed set $\hGn^{k,+}$. Thanks to the resampling and movement steps based on $q_n^k$, this ensures that the particle set covers both the interior and the neighborhood of the boundary of the new target quantile set.

\begin{remark}
The number $K$ of quantile sets to estimate is not known in advance. Indeed, the sequence $(C_k)_k$ is constructed adaptively, based on past evaluations of the function $f$.
\end{remark}

\medbreak

\begin{remark}
The construction of $C_k$ and $\hGn^{k,+}$ is problem-specific. The scalar case with $C_k = (-\infty, T_k)$ has been used here for simplicity, but the method readily generalizes to more complex settings, as illustrated in \autoref{eem:sec:rotor}.
\end{remark}

\subsection{Stopping criterion}
\label{eem:sec:stop}

Given $n$ evaluations of $f$ and a quantile set $\Gamma^k(f)$, the number $N_k$ of additional batches needed to obtain a satisfactory approximation is not known in advance. It is therefore necessary to assess the quality of the current estimate in order to decide when to stop adding evaluation points.

\medbreak

We define the expected relative estimation error as
\begin{align*}
\Xi_n^k = \frac{\En(\lambda(\Gamma^k(\xi)\Delta\hGn^k))}{{\lambda(\hGn^k)}}.
\end{align*}
This quantity, however, does not admit a closed-form expression. Similarly to the QSI-EEM criterion, we approximate it as
\begin{align*}
\Xi_n^k \approx \frac{\Psi^k_{n, |\mathds{O}_n|}(\mathds{O}_n)}{\lambda(\hGn^k)},
\end{align*}
where $\mathds{O}_n \subset \Uset$ is a large set of points selected to capture regions of high residual uncertainty. In practice, as detailed in \autoref{eem:sec:implementation}, $\mathds{O}_n$ is obtained by sampling (without replacement) points according to a distribution proportional to $\min(p_n^k(u), 1 - p_n^k(u))$, where $p_n^k(u) = \Pn(\xi(u) \in C_k)$.

\medbreak

The stopping condition is thus defined as
\begin{equation}
  \frac{\Psi_{n, |\mathds{O}_n|}(\mathds{O}_n)}{\lambda(\hGn^k)} \le \tau_k,
\end{equation}
where $\tau_k \in (0,1)$ denotes a stage-dependent tolerance parameter.

\begin{remark}
The tolerance threshold $\tau_k$ may vary across stages of the algorithm, hence the subscript $k$. In particular, as done in the next section, it is reasonable to use a lower tolerance when estimating the final target quantile set.
\end{remark}

\section{Numerical results}
\label{eem:sec:numerical}

\subsection{Implementation}
\label{eem:sec:implementation}

\paragraph{Gaussian process prior} 
The unknown function $f$ is modeled by a multivariate Gaussian process with independent components $\xi_1,\dotsc,\xi_q$ ($q = 1$ in \autoref{eem:sec:synthetic_functions}, $q = 2$ in \autoref{eem:sec:rotor}). Each component is assumed to be a Gaussian process with unknown constant mean and anisotropic Matérn covariance function \citep[see, e.g.,][]{chiles:2012:geostats}. Covariance parameters are re-estimated after each new batch of evaluations using restricted maximum likelihood (ReML) \citep{stein1999interpolation}. The regularity parameter $\nu$ is restricted to $\{1/2, 3/2, 5/2, +\infty\}$ to take advantage of the simplified closed-form expressions of the Matérn kernel for half-integer values. Note that the limiting case $\nu = +\infty$ corresponds to the Gaussian (or squared exponential) covariance function. A nugget effect of $10^{-6}$ is added to the covariance matrix to mitigate numerical instabilities due to ill-conditioning.

\medbreak

\paragraph{Initial design}
The GPs are initially conditioned using a design of size $n_0 = 10(d_\X + d_\S)$, in line with the heuristic of \cite{loeppky2009choosing}, sampled via a maximin Latin hypercube design. The maximin design is approximated by selecting the best among $1000$ independent Latin hypercube samples, based on their minimum pairwise distance.

\medbreak

\paragraph{Parameters of the SMC algorithm}
The initial set $\tX_0$ is built as a random Latin hypercube sample of size $|\tX_0| = 250$. At each stage $k$, the relaxed estimator used to define $q_n^k$ is calibrated with $\kappa \approx 1.1$, and the next quantile set $\Gamma^{k+1}(f)$ is defined such that $\rho \approx 0.35$ (see \autoref{eem:sec:density}).
Particles are resampled using residual resampling \citep{hol:2006:resampling}, and then moved 25 times using the adaptive SMC kernel described in \aref{eem:app:MH}.

\medbreak

\paragraph{Approximation of the criterion}
The QSI-EEM criterion is approximated on a grid $\tX_k^{\mathrm{rand}}\times\tS$, where $\tX_k^{\mathrm{rand}}$ is a random sample of 100 particles drawn uniformly (without replacement) from the current population and $\tS$ is obtained by applying an inverse transformation (with respect to $\Ps$) to the first $2^9 = 512$ points of a Sobol sequence in $[0,1]^{d_\S}$.
Since in our examples the components of $S \sim \Ps$ are independent, the transformation is applied coordinate-wise using the inverse cumulative distribution functions of the marginals of $\Ps$.

\begin{remark}
Such a transformation exists under general conditions \citep{sklar1959fonctions}, and can be constructed via methods such as the Rosenblatt or Nataf transforms \citep{lebrun:2009:transform, Melchers:2017:structural}.
In recent machine learning research, when a large sample from $\Ps$ is available, such transformations can be approximated using normalizing flows \citep[see, e.g.,][]{Kobyzev2020NormalizingFA}.
\end{remark}

\medbreak
\paragraph{Optimization of the criterion}
To reduce the computational cost of optimizing over $\Uset^r$, we adopt a sequential strategy inspired by \cite{sacks:1989:optim_seq,chevalier:2014:fast_parallel_kriging}.
For $b = 1,\dotsc, r$, the $(n + b)$-th evaluation point is selected by optimizing the partial criterion
\begin{align*}
U_{n+b} \in \underset{u \in \Uset}{\arg\min}\, \Psi^k_{n,b}(U_{n+1},\dotsc,U_{n+b-1}, u).
\end{align*}
Each partial criterion is optimized numerically with a continuous solver, initialized at
\begin{align*}
  U_{n+b}^{start} \in \underset{u^{start} \in \Uset^{\circ}_b}{\arg\min}\, \Psi^k_{n,b}(U_{n+1},\dotsc,U_{n+b-1}, u^{start}),
\end{align*}
where $\Uset^{\circ}_b$ is a set of 100 points sampled without replacement from $\tUset_k = \tX_k \times \tS$, according to the misclassification probability $\min(p_n^k(u), 1 - p_n^k(u))$, with $p_n^k(u) = \Pn(\xi(u) \in C_k)$.

\medbreak

\paragraph{Stopping criterion}
Two different thresholds are used depending on whether the current target is an intermediate quantile set or the final one. The criterion is met when the estimated potential variation (see \autoref{eem:sec:stop}) satisfies $\Xi_n^k \le \tau = 1/3$ for an intermediate set, or $\Xi_n^K \le \tau = 1/5$ for the final target set. This quantity is approximated using an inducing set $\mathds{O}_n$ of 250 points sampled without replacement from $\tUset_k$, according to the weights $\min(p_n^k(u), 1 - p_n^k(u))$.

\medbreak

\paragraph{Restart procedure}
To avoid degeneration of the SMC algorithm, a restart procedure is included.
If, after adding a new batch at stage $k$, less than 10\% or more than 75\% of the particles in $\tX_k$ lie in the relaxed set $\hGn^{k,+}$ (see \autoref{eem:sec:density}), the algorithm is restarted from the earliest stage where either this condition or the stopping criterion was violated. This procedure is motivated by the fact that, by construction, $\hGn^{k,+}$ is expected to contain about $\rho = 35\%$ of the points in $\tX_k$.
Triggering the restart condition suggests a potential loss of information or misalignment in the estimation process at previous stages.

\medbreak

All numerical experiments are performed using Matlab v2022a and STK v2.8.1 \citep{STK}.

\subsection{Synthetic functions}
\label{eem:sec:synthetic_functions}

We begin with three QSI problems based on synthetic functions, all with scalar outputs and noiseless observations.
To assess the performance of our strategy, we report the number of batches required by the QSI-EEM method and the relative estimation error $\frac{\lambda(\hat{\Gamma}_N\Delta\Gf)}{\lambda(\Gf)}$.
Each experiment is repeated over 100 independent runs, with varying initial designs and batch sizes $r \in \{1, 2, 3\}$.

\medbreak

\paragraph{First case}
We consider an adaptation of the piston simulation function \citep{piston}, defined by
\begin{align*}
f_1(u) = 2\pi\sqrt{\frac{x_1}{K+x_2^2\frac{s_1x_3s_2}{x_4V}}},
\end{align*}
with $V = \frac{x_2}{2K}\left(\sqrt{A^2 + 4K\frac{s_1x_3}{x_4}s_2} - A\right)$, $A = s_1x_2 + 19.62x_1 - \frac{Kx_3}{x_2}$, and $K = 1000$.

\medbreak

This function models the cycle time (in seconds) of a piston as a function of its weight $x_1$, surface area $x_2$, initial gas volume $x_3$, gas temperature $x_4$, ambient temperature $s_1$, atmospheric pressure $s_2$, and spring coefficient $K$. 

\medbreak

The domains of definition are $\X = [30,60]\times[0.005,0.02]\times[0.002, 0.01]\times [340, 360]$ and $\S = [90000, 110000]\times[290, 296]$. The distribution $\Ps$ on $\S$ is taken as the product of two rescaled Beta distributions $\mathcal{B}(1/2,1/2)$ (i.e., arcsine distributions).
We consider $C = (-\infty, 1.12]$ and $\alpha = 0.05$.
Under this setting, the quantile set satisfies $\lambda(\Gf)\approx 10^{-5}\cdot\lambda(\X)$.

\medbreak

\paragraph{Second case}
The second example uses the 7-dimensional Trid function \citep{trid}, defined by
\begin{align*}
f_2(u) = \sum_{i=1}^7(u_i-1)^2 - \sum_{i=2}^7u_iu_{i-1},
\end{align*}
with $x = (u_1, \dotsc, u_4) \in \X = [-49, 49]^4$, $s = (u_5, u_6, u_7) \in \S = [-49, 49]^3$, and $\Ps = \mathcal{U}(\S)$ the uniform distribution on $\S$.
We take $C = [4700, +\infty)$ and $\alpha = 0.10$. This yields a quantile set of size $\lambda(\Gf)\approx 10^{-6}\cdot\lambda(\X)$.

\medbreak

\paragraph{Third case}
The final example is based on the OTL circuit function \citep{BenAri:2007:otl}, defined over $\X = [50, 150]\times[25, 70]\times[0.5, 3]\times[1.2, 2.5]\times[0.25, 1.2]$ and $\S = [-50, 300]$ by
\begin{align*}
f_3(x,s) = \frac{(h(x) + 0.74)s(x_5+9)}{s(x_5+9)+x_3} + \frac{11.35x_3}{s(x_2+9)+x_3} + \frac{0.74x_3s(x_5+9)}{(s(x_5+9)+x_3)x_4},
\end{align*}
with $h(x) = 12\frac{x_2}{x_1+x_2}$.
\medbreak
We use a Gaussian distribution for $\Ps$, with mean $175$ and variance $50$, truncated to $\S$.  We set $C = [2.65, +\infty)$ and $\alpha = 0.05$. This results in a quantile set of size $\lambda(\Gf) \approx 10^{-8}\cdot \lambda(\X)$.

\medbreak

This function models the midpoint voltage of a push-pull circuit, given resistance values $x_1,\dotsc,x_5$ and current gain $s$.
The problem can be interpreted as finding resistance values such that the voltage remains below 2.65 with 95\% probability over the distribution of current gain.

\medbreak

Experimental results are shown in \autoref{eem:fig:scatter_piston} to \autoref{eem:fig:scatter_otl}.

\begin{figure}
  \begin{adjustwidth}{-2cm}{-2cm}
  \centering
  \includegraphics[width=7.5cm]{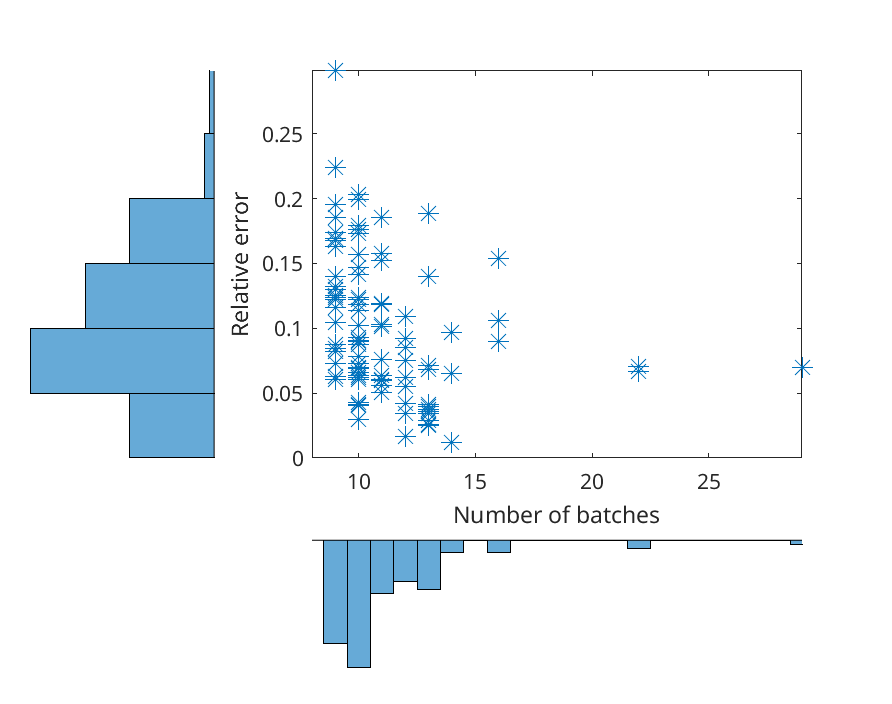}\\
  \includegraphics[width=7.5cm]{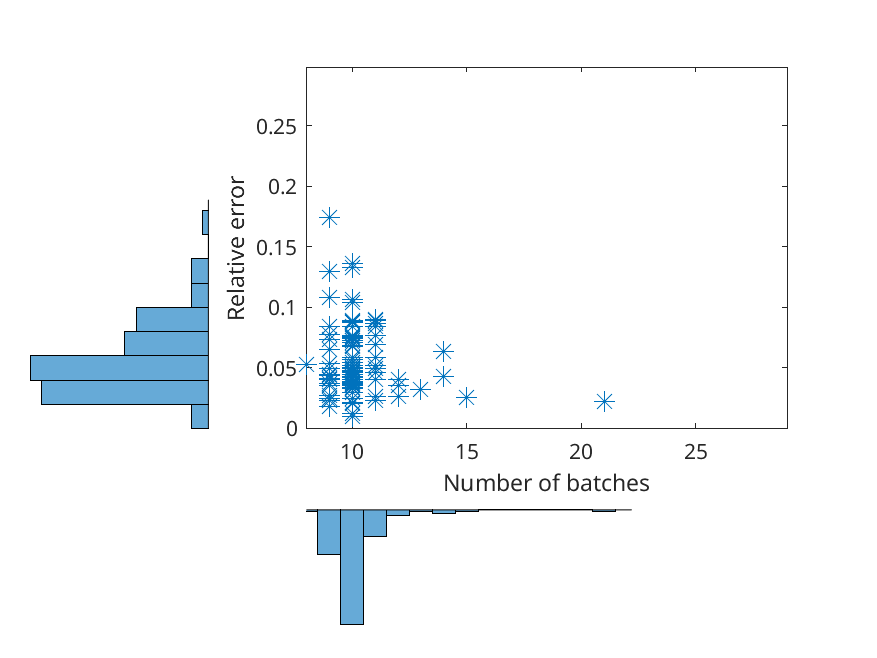}
  \includegraphics[width=7.5cm]{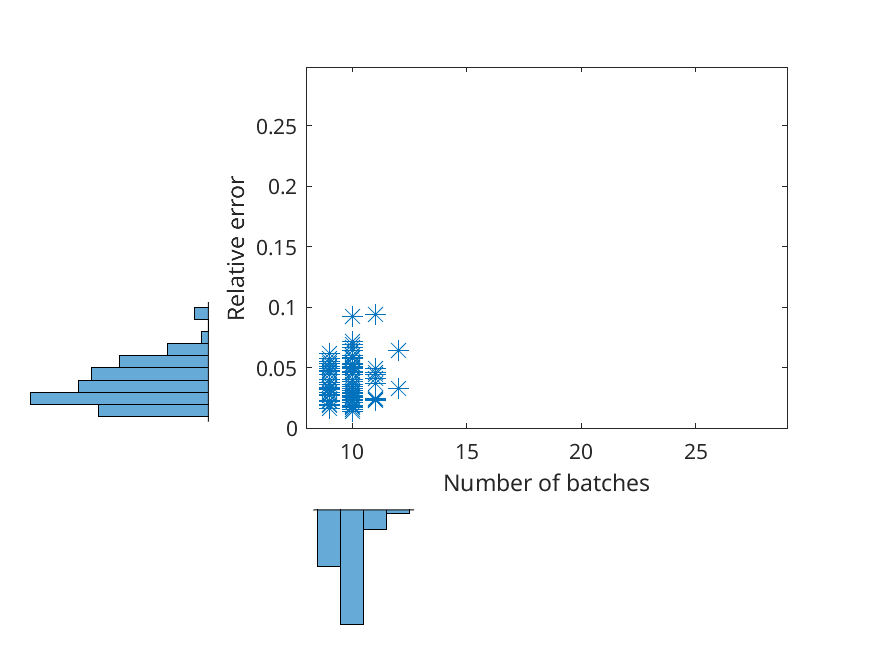}
  \caption{Scatter plots (and distributions) of the relative error against the number of batches on the piston simulation function example, for different batch sizes $r$ (top to bottom, left to right: $r=1,2,3$) and a $100$ runs of the algorithm.}
  \label{eem:fig:scatter_piston}
  \end{adjustwidth}
\end{figure}

\begin{figure}
  \begin{adjustwidth}{-2cm}{-2cm}
  \centering
  \includegraphics[width=7.5cm]{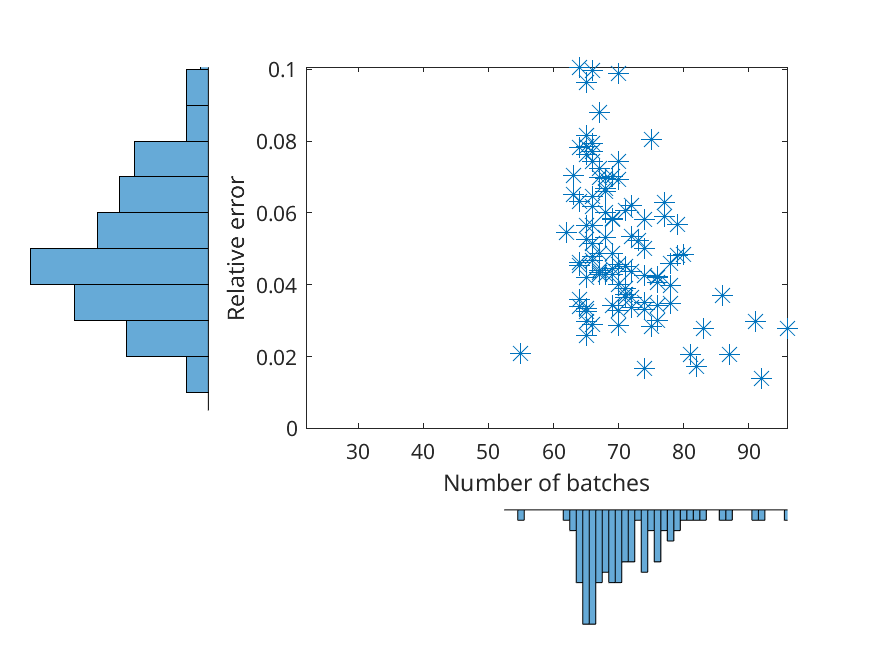}\\
  \includegraphics[width=7.5cm]{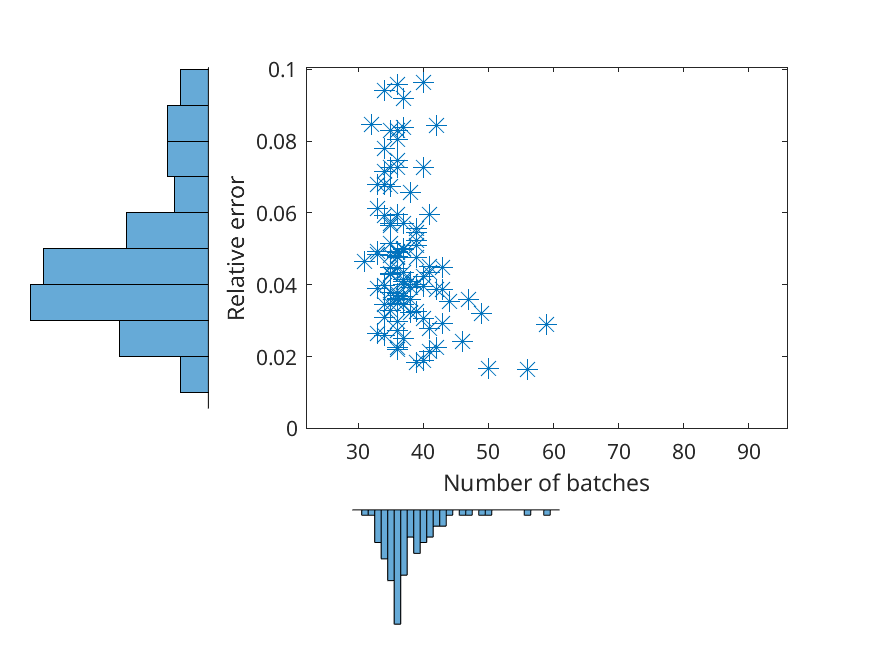}
  \includegraphics[width=7.5cm]{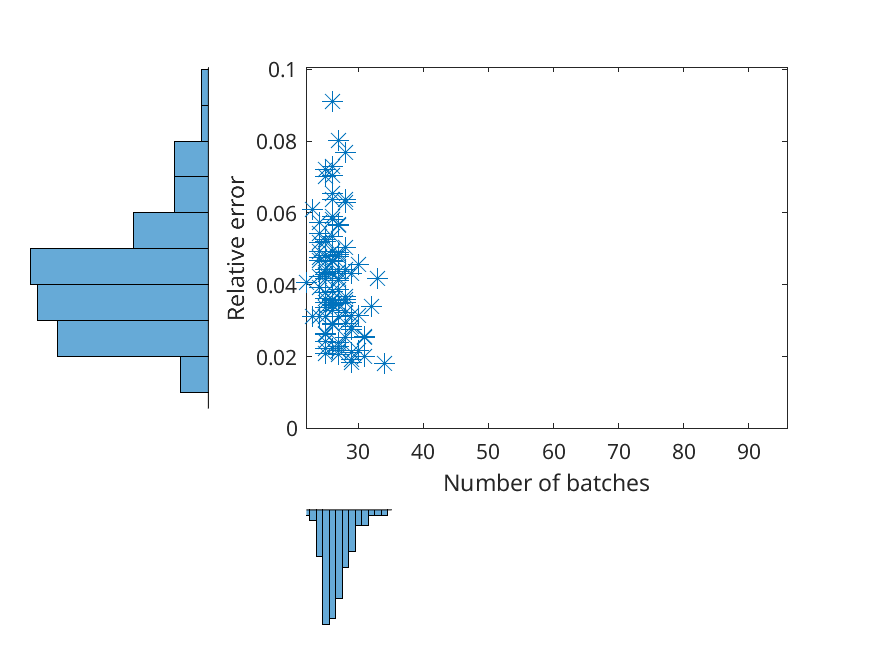}
  \caption{Scatter plots (and distributions) of the relative error against the number of batches on the $7$-Trid example, for different batch sizes $r$ (top to bottom, left to right: $r=1,2,3$) and a $100$ runs of the algorithm.}
  \label{eem:fig:scatter_trid}
  \end{adjustwidth}
\end{figure}

\begin{figure}
  \begin{adjustwidth}{-2cm}{-2cm}
  \centering
  \includegraphics[width=7.5cm]{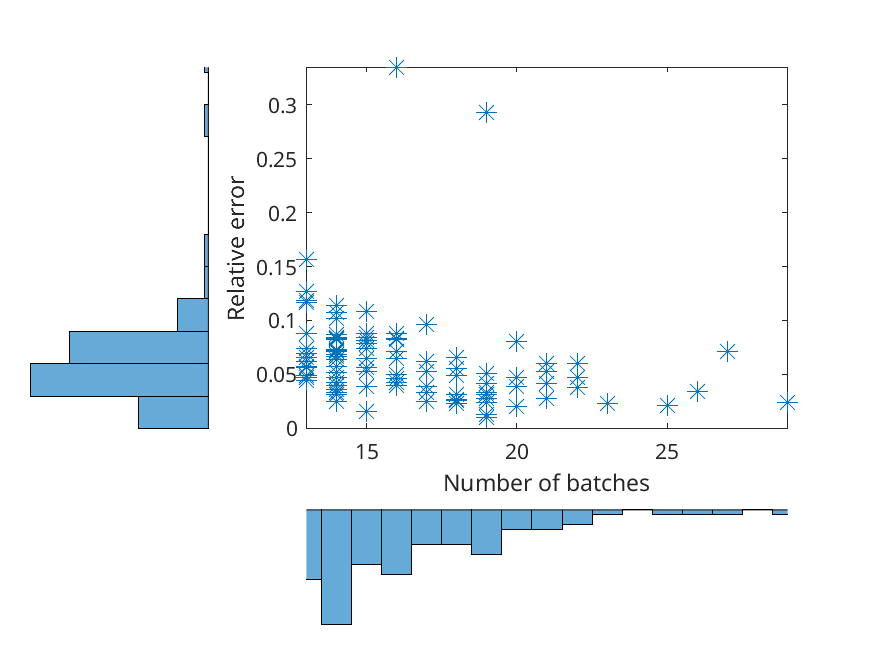}\\
  \includegraphics[width=7.8cm]{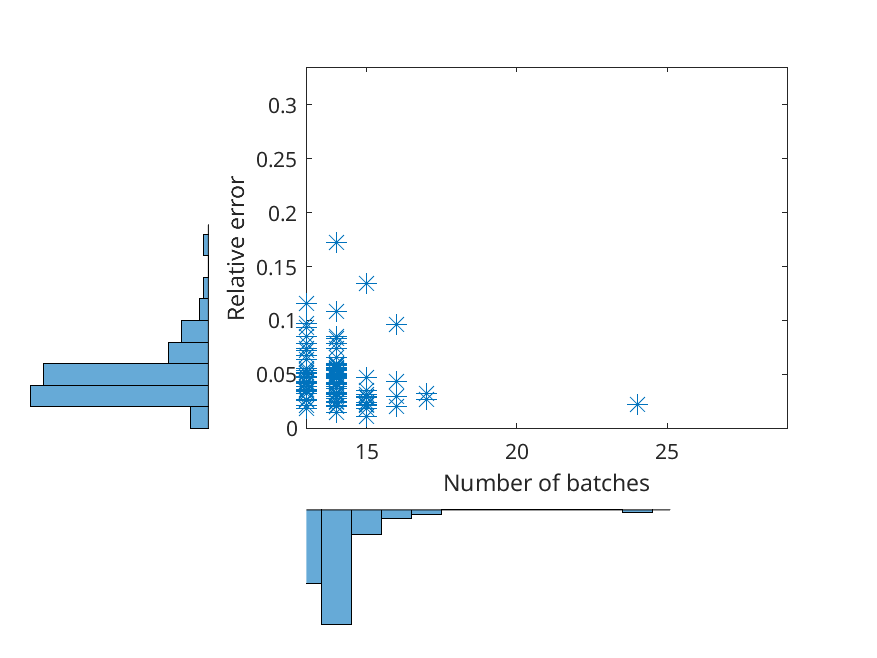}
  \includegraphics[width=7.5cm]{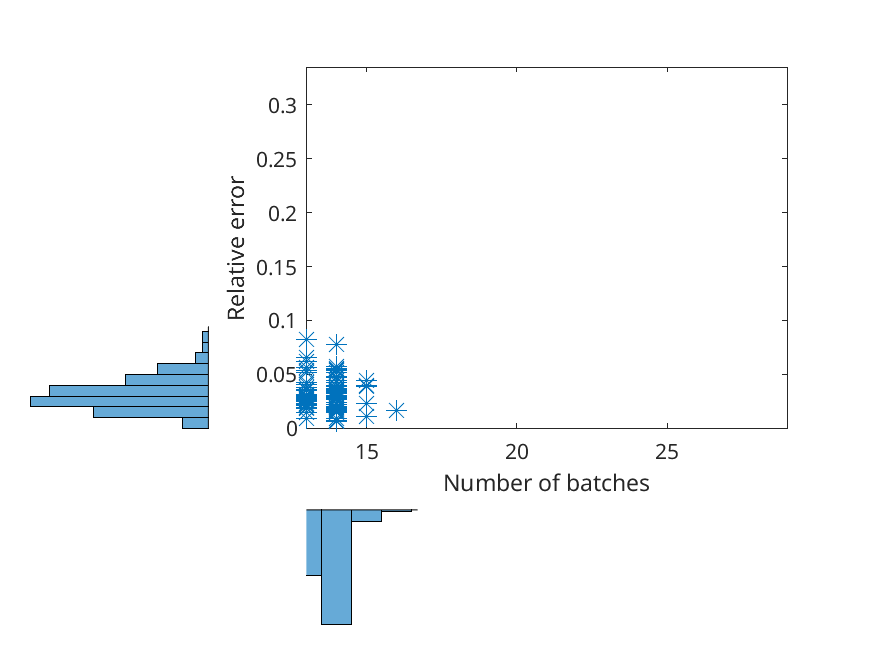}
  \caption{Scatter plots (and distribution) of the relative error against the number of batches on the OTL example, for different batch sizes $r$ (top to bottom, left to right: $r=1,2,3$) and a $100$ runs of the algorithm.}
  \label{eem:fig:scatter_otl}
  \end{adjustwidth}
\end{figure}

\medbreak

As a benchmark method, we include a comparison between our strategy and the Bayesian Subset Simulation (BSS) algorithm \citep{bect:2017:bss}, applied to the estimation of the set
\begin{align*}
\Lambda(f) = \{u \in \Uset \, : \, f(u) \in C\} \subset \X\times\S
\end{align*}
We remind, as discussed in \autoref{eem:sec:related}, that a metamodel accurately classifying the points $u \in \Uset$ with respect to membership in $\Lambda(f)$ can be used to infer membership of $x \in \X$ in $\Gf$.
We use the authors’ Matlab implementation\footnote{\url{https://github.com/stk-kriging/contrib-bss}}, with default parameters and $m = 5000$ particles in $\Uset$, adapting only the GP parameter estimation to match our QSI-EEM implementation.
The experiments are repeated over 100 independent runs, using the same initial designs as in the QSI-EEM experiments.

\medbreak

Results are shown in \autoref{eem:fig:scatter_bss} and summarized in \autoref{eem:table:piston} to \autoref{eem:table:otl}.

\medbreak 
We observe that the QSI-EEM strategy, combined with the SMC framework, provides highly accurate estimates of the quantile sets across the three synthetic examples. While the total number of evaluation points increases with batch size $r$, the number of batches queried tends to decrease slightly. The average and worst-case relative errors also decrease with batch size, likely due to the increased number of evaluations per stage. These results support, when computationally feasible, the use of parallel evaluations of the simulator to construct batch sequential designs.

\begin{figure}
  \begin{adjustwidth}{-2cm}{-2cm}
  \centering
  \includegraphics[width=7.5cm]{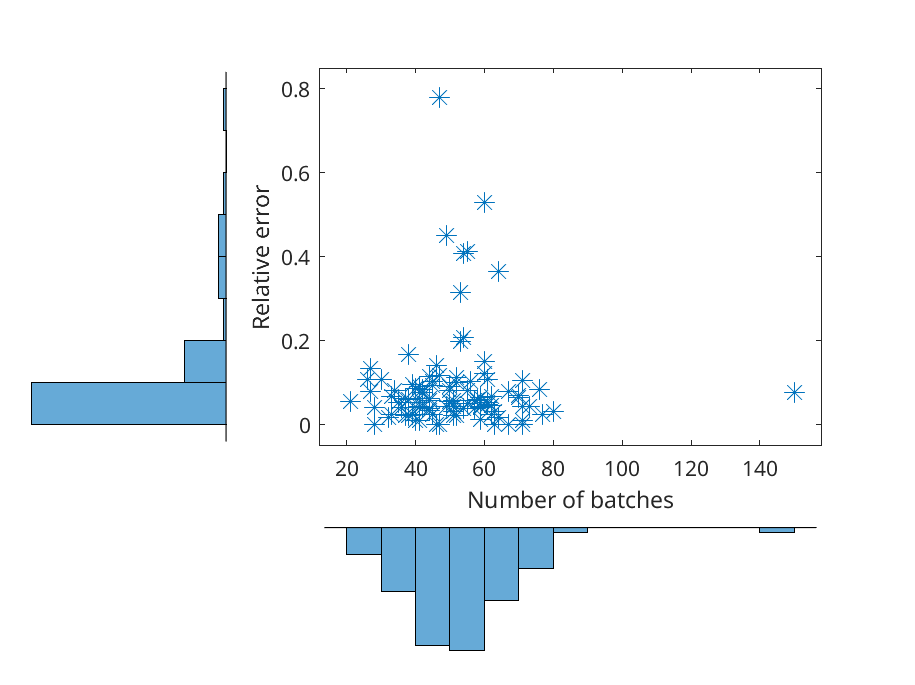}\\
  \includegraphics[width=7.8cm]{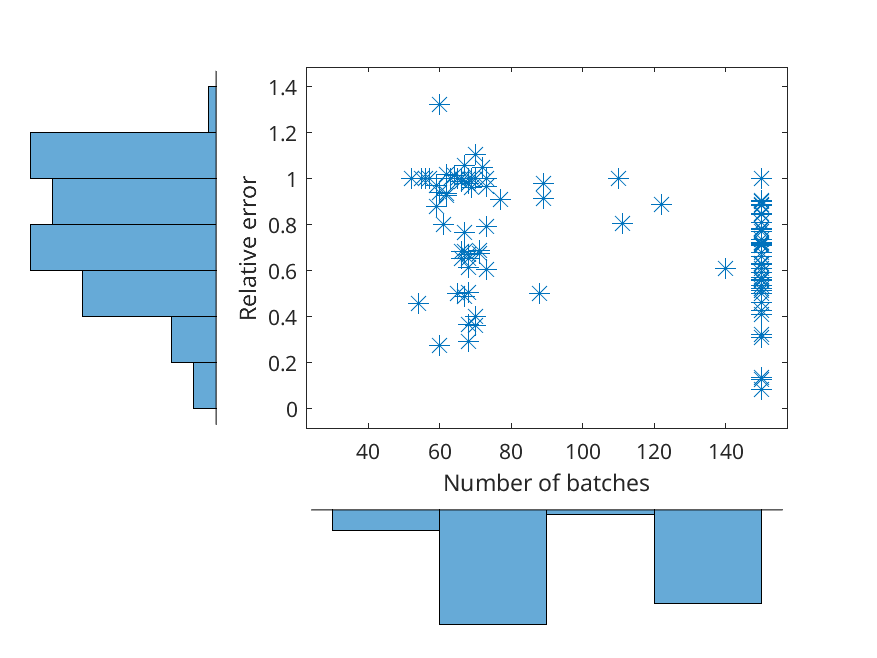}
  \includegraphics[width=7.5cm]{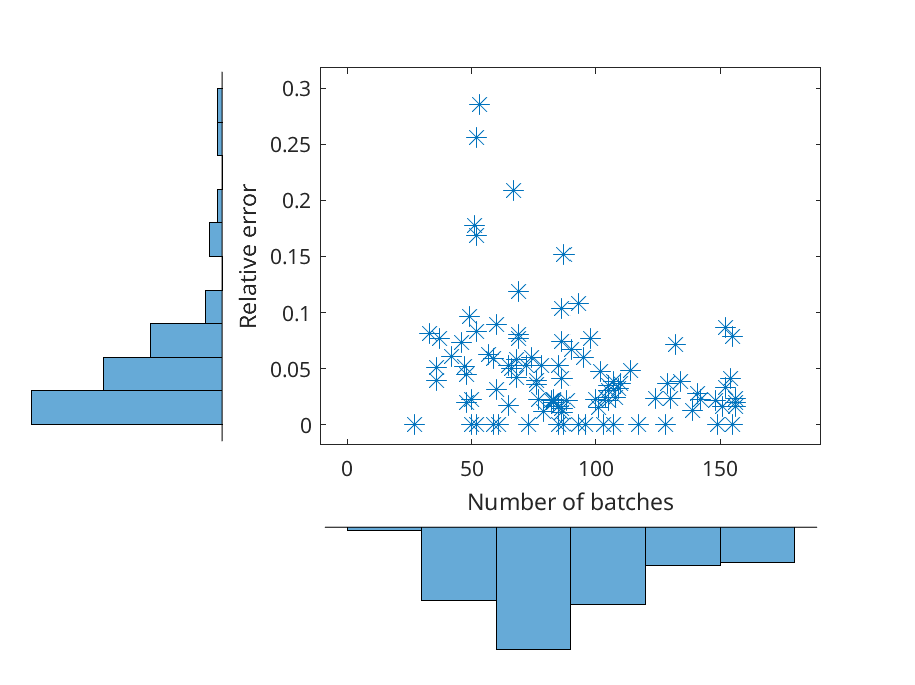}
  \caption{Scatter plots (and distributions) of the relative error against the number of batches using the BSS strategy on a $100$ independent runs. (top to bottom, left to right: $f_1$, $f_2$, $f_3$.)}
  \label{eem:fig:scatter_bss}
  \end{adjustwidth}
\end{figure}

\medbreak

The results obtained with the BSS algorithm indicate either poor performance or comparable performance at the cost of a significantly larger number of queried points.
In particular, for the second test case $f_2$, \autoref{eem:table:trid} shows that BSS requires a very large number of evaluations and often fails to construct a surrogate model capable of recovering the quantile set with acceptable accuracy. In some of the worst-case runs, the relative estimation error exceeds 1, indicating highly inaccurate estimates. This behavior can be attributed to the fact that BSS targets the entire set $\Lambda(f)$, rather than the specific regions of interest for estimating $\Gf$. Furthermore, the stopping condition in BSS is based on the estimation error of $\Lambda(f)$, which may poorly reflect the error on $\Gf$.
By contrast, BSS performs more consistently on the first and third test cases ($f_1$ and $f_3$), as shown in \autoref{eem:table:piston} and \autoref{eem:table:otl}. However, although the average estimation accuracy is comparable to that of the QSI-EEM strategy on these cases, it is achieved at the cost of a substantially higher number of evaluations.

\begin{table}[ht]
  \begin{adjustwidth}{-3cm}{-3cm}
  \centering
  \footnotesize
  \begin{tabular}{|l|c|c|c|c|}
    \hline
    & QSI-EEM ($r=1$) & QSI-EEM ($r=2$) & QSI-EEM ($r=3$) & BSS \\
    \hline
    No. of batches & 11.2 \scriptsize (9 -- 29) & 10.2 \scriptsize(8 -- 21) & 9.8 \scriptsize(9 -- 12) & 51.4 \scriptsize(21 -- 150) \\
    \hline
    No. of points & 11.2 \scriptsize(9 -- 29)  & 20.4 \scriptsize(16 -- 42) & 29.5 \scriptsize(27 -- 36) & 51.4 \scriptsize(21 -- 150)\\
    \hline
    Rel. error & 0.100 \scriptsize(0.012 -- 0.284) & 0.056 \scriptsize(0.010 -- 0.174) & 0.037 \scriptsize(0.014 -- 0.094) & 0.092 \scriptsize(0 -- 0.778) \\
    \hline
\end{tabular}
    \caption{Summary of QSI-EEM and BSS metrics for case 1. Average (min - max).}
  \label{eem:table:piston}
\end{adjustwidth}
\end{table}

\begin{table}[ht]
  \begin{adjustwidth}{-3cm}{-3cm}
  \centering
  \footnotesize
  \begin{tabular}{|l|c|c|c|c|}
    \hline
    & QSI-EEM ($r=1$) & QSI-EEM ($r=2$) & QSI-EEM ($r=3$) & BSS \\
    \hline
    No. of batches & 70.7 \scriptsize(55 -- 96) & 37.7 \scriptsize(31 -- 59) & 26.6 \scriptsize(22-- 34) & 100.5 \scriptsize(52 -- 150) \\
    \hline
    No. of points & 70.7 \scriptsize(55 -- 96)  & 75.5 \scriptsize(62 -- 118) & 79.8 \scriptsize(66 -- 102) & 100.5 \scriptsize(52 -- 150) \\
    \hline
    Rel. error & 0.050 \scriptsize(0.014 -- 0.100) & 0.047 \scriptsize(0.016 -- 0.096) & 0.041 \scriptsize(0.018 -- 0.091) & 0.751 \scriptsize(0.086 -- 1.322)\\
    \hline
\end{tabular}
    \caption{Summary of QSI-EEM and BSS metrics for case 2. Average (min - max).}
  \label{eem:table:trid}
\end{adjustwidth}
\end{table}


\begin{table}[ht]
  \begin{adjustwidth}{-3cm}{-3cm}
  \centering
  \footnotesize
    \begin{tabular}{|l|c|c|c|c|}
        \hline
        & QSI-EEM ($r=1$) & QSI-EEM ($r=2$) & QSI-EEM ($r=3$) & BSS \\
        \hline
        No. of batches  & 16.5 \scriptsize(13 -- 29) & 14.0 \scriptsize(13 -- 24) & 13.7 \scriptsize(13 -- 16) & 90.1 \scriptsize(27 -- 156)\\
        \hline
        No. of points  & 16.5 \scriptsize(13 -- 29) & 28.0 \scriptsize(26 -- 48) & 41.2 \scriptsize(39 -- 48) & 90.1 \scriptsize(27 -- 156)\\
        \hline
        Rel. error  & 0.062 \scriptsize(0.010 -- 0.334) & 0.047 \scriptsize(0.011 -- 0.172) &0.032 \scriptsize(0.007 -- 0.082) & 0.047 \scriptsize(0 -- 0.286)\\
        \hline
    \end{tabular}
    \caption{Summary of QSI-EEM and BSS metrics for case 3. Average (min - max).}
  \label{eem:table:otl}
  \end{adjustwidth}
\end{table}

\subsection{Application to ROTOR37}
\label{eem:sec:rotor}

We now consider the ROTOR37 compressor model \citep{reid:1978:rotor37}.

\medbreak

A key challenge for engine manufacturers is to design rotor compressors---which increase the density of a gas within the engine---that are both efficient and meet performance constraints with high probability. In the case of the ROTOR37 model, the goal is to control the probability that the mass flow or pressure ratio deviates from a given range around baseline values.

\medbreak

We model the simulator as a function $f : \X\times\S \to \Rset^3$, with $\X = [0,1]^{13}$ and $\S = [0,1]^5$, describing the behavior of a single-stage transonic axial compressor rotor (see \autoref{eem:fig:rotor}) under a given set of inputs. We assume that $\S$ is endowed with the uniform distribution $\mathcal{U}(\S)$. 

\medbreak

The input space $\X$ corresponds to geometric and design variables (e.g., camber, thickness) of the rotor blades, while $\S$ represents manufacturing uncertainties (e.g., blade deviations).
The output of the simulator consists of three variables: the mass flow ($f_1$), the pressure ratio ($f_2$), and the isentropic efficiency ($f_3$).
In the following, we focus exclusively on mass flow $f_1$ and pressure ratio $f_2$.

\begin{remark}
The simulator used is in fact a Gaussian metamodel of the ROTOR37 numerical simulator, based on normalized inputs and built using RobustGaSP \citep{RobustGaSP} with a large space-filling input design.
\end{remark}

\medbreak

Our objective is to estimate the set of feasible designs. More precisely, we aim to identify designs for which, under the distribution of manufacturing uncertainties, both $f_1$ and $f_2$ remain close to their baseline values with high probability.
We define $\Gf$ with $\alpha = 0.05$ and
\begin{align*}
C & = \left\{(z_1, z_2) \in \Rset^2 \, : \, \frac{|z_1 - b_1|}{|b_1|} >\mathrm{tol} \quad \text{or} \quad \frac{|z_2 - b_2|}{|b_2|} >\mathrm{tol}\right\},
\end{align*}
with baseline values $b_1 = 0.655$ and $b_2 = 0.582$, and a tolerance parameter $\mathrm{tol} = 0.18$ around these targets. This setup yields a quantile set of size $\lambda(\Gf) \approx 10^{-8} \cdot \lambda(\X)$.

\begin{figure}
  \center
  \includegraphics[width=10cm]{graphs_eem/rotor_US_archive.eps}
  \caption{Illustration of the compressor rotor modeled by the simulator (source: US National Archives).}
\label{eem:fig:rotor}
\end{figure}

\medbreak

For the SMC component of our strategy, defining suitable target densities $q_n^k$ is more involved than in the synthetic examples of \autoref{eem:sec:synthetic_functions}, due to the shape of the region $C$.
Let $\mu_n^1$ and $\mu_n^2$ be the posterior means of the two independent GPs modeling the simulator outputs, and $\sigma_n^{2,1}$, $\sigma_n^{2,2}$ the corresponding posterior variances. Let $\beta > 1/2$ be a tuning parameter.
We define the modified mean function $\mu_n^\beta = (\mu_n^{\beta,1}, \mu_n^{\beta,2})$ by:
\begin{align*}
\mu_n^{\beta, i}(u) = \begin{cases}
\min\left(\mu_n^i(u) + \sqrt{\sigma_n^{2,i}(u)}\Phi^{-1}(\beta), b_i\right) & \quad \text{if} \quad \mu_n^i(u) < (1-\mathrm{tol})b_i,\\
\max\left(\mu_n^i(u) - \sqrt{\sigma_n^{2,i}(u)}\Phi^{-1}(\beta), b_i\right) & \quad \text{if} \quad \mu_n^i(u) > (1+\mathrm{tol})b_i,\\
\mu_n^i(u) & \quad \text{otherwise},
\end{cases}
\end{align*}
for $i = 1,\, 2$. Following the procedure of \autoref{eem:sec:density}, we define $q_n^k = \mathds{1}_{\hGn^{k,+}}$, where the relaxed estimator $\hGn^{k,+}$ is given by
\begin{align*}
\hGn^{k,+} = \{x \in \X \, : \, \P(\mu_n^\beta(x, S) \in C_k) \le \alpha\}.
\end{align*}

As in the previous section, we repeat the QSI-EEM procedure over 100 independent runs using a batch size of $r = 5$, starting from different initial designs. A minor implementation change is made by setting the initial design size to $n_0 = 5(d_\X + d_\S) = 90$.

\medbreak

The results are shown in \autoref{eem:fig:scatter_rotor}. We observe that our strategy succeeds in retrieving the small quantile set of relative size on the order of $10^{-8}$ within 18 to 28 batches, requiring between 90 and 140 additional evaluations beyond the initial design. In most runs, the relative error is below or equal to 0.3, with a maximum near 0.4.
From \autoref{eem:fig:rotor_histo}, we can appreciate that, for each of the $100$ independent runs, the last batch of $5$ points sampled appears (at least marginally) to be in or close to the target quantile set. \autoref{eem:fig:rotor_histo} shows that in all 100 runs, the final batch of five points tends to lie within or close to the target quantile set. Furthermore, as seen in \autoref{eem:fig:parallel_rotor}, the estimated quantile set, although not identical to $\Gf$, is geometrically close, suggesting that the strategy captures the overall structure of the target set.


\begin{figure}
  \center
\includegraphics[width=10cm]{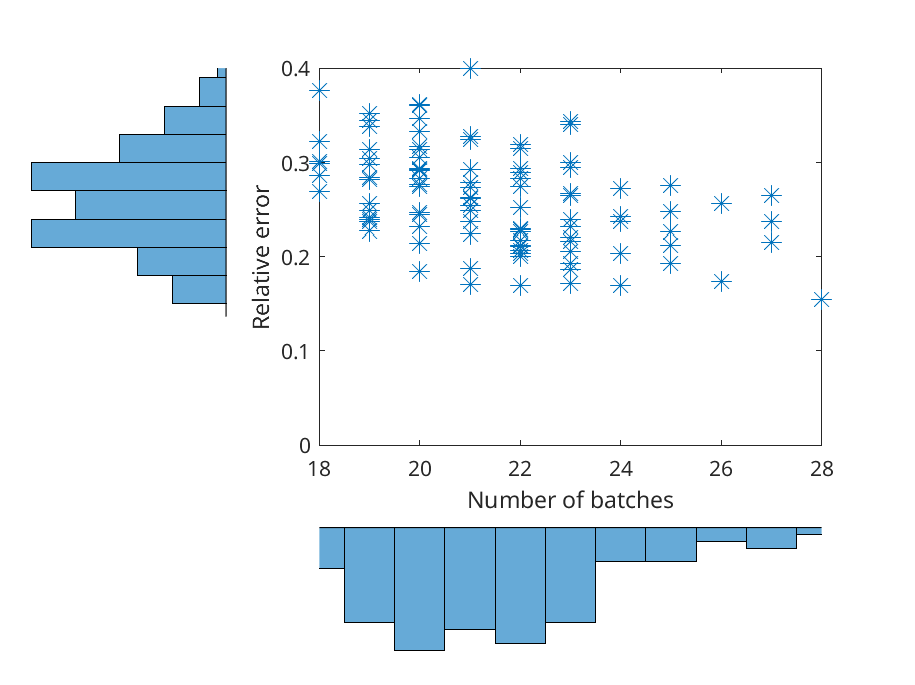}
\caption{Scatter plots (and distributions) of the relative error against the number of batches (with $r=5$) on the ROTOR37 case for a $100$ independent runs of the algorithm.}
\label{eem:fig:scatter_rotor}
\end{figure}

\begin{figure}
  \begin{adjustwidth}{-1.5cm}{-1.5cm}
    
  \centering
  \includegraphics[width=5cm]{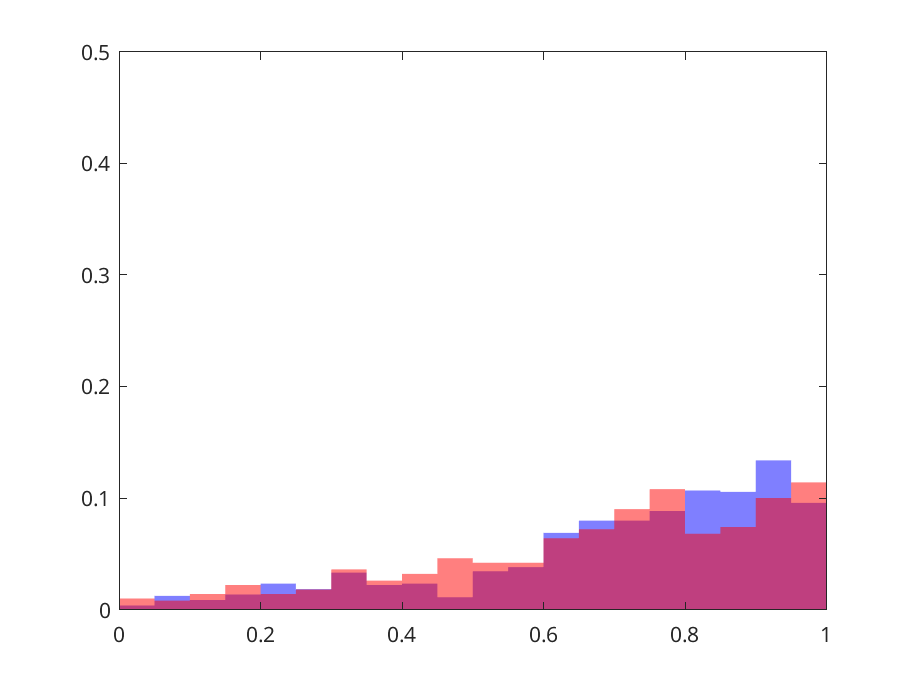}\\
  \includegraphics[width=5cm]{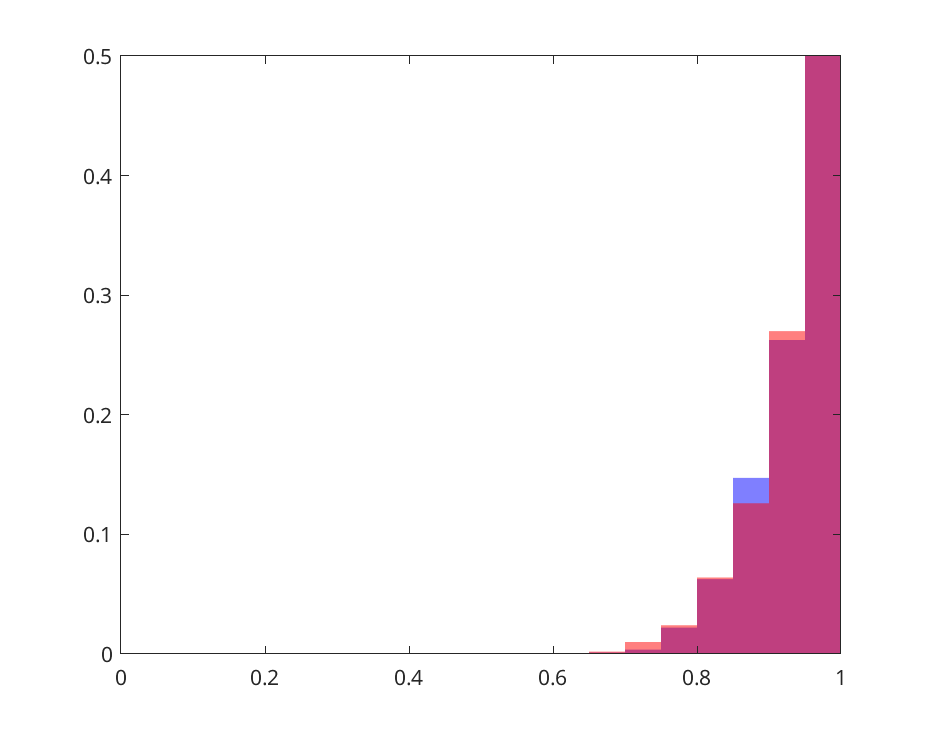}
  \includegraphics[width=5cm]{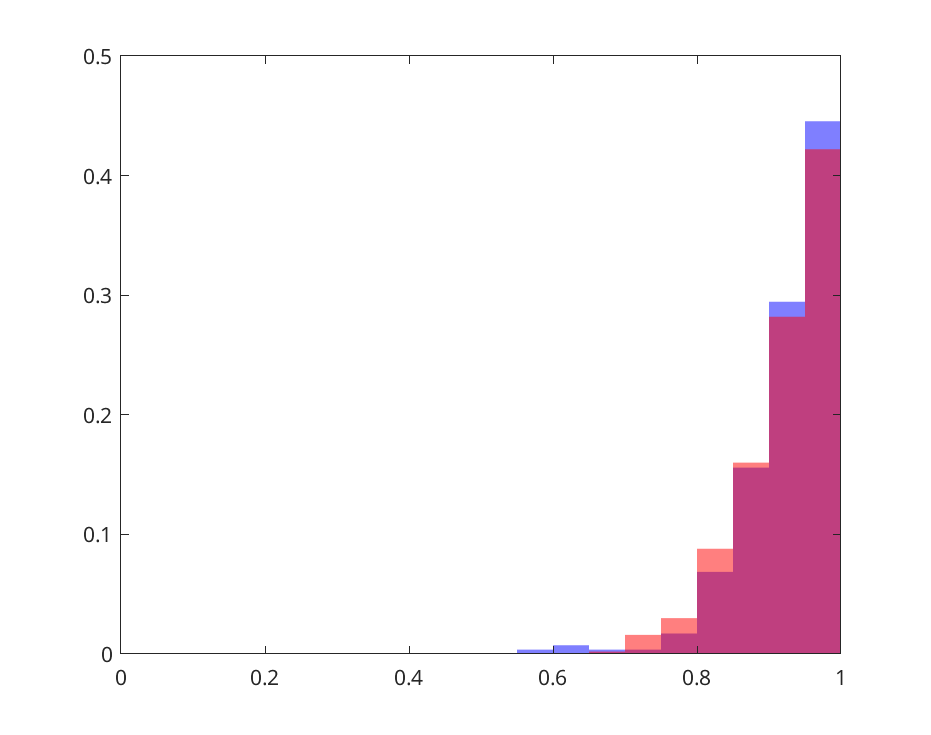}
  \includegraphics[width=5cm]{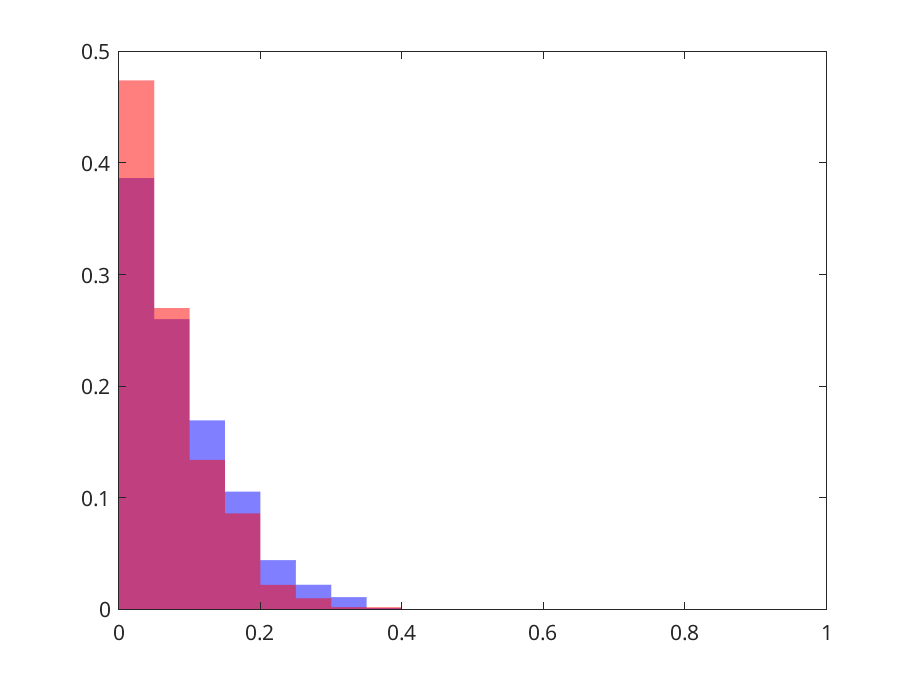}\\
  \includegraphics[width=5cm]{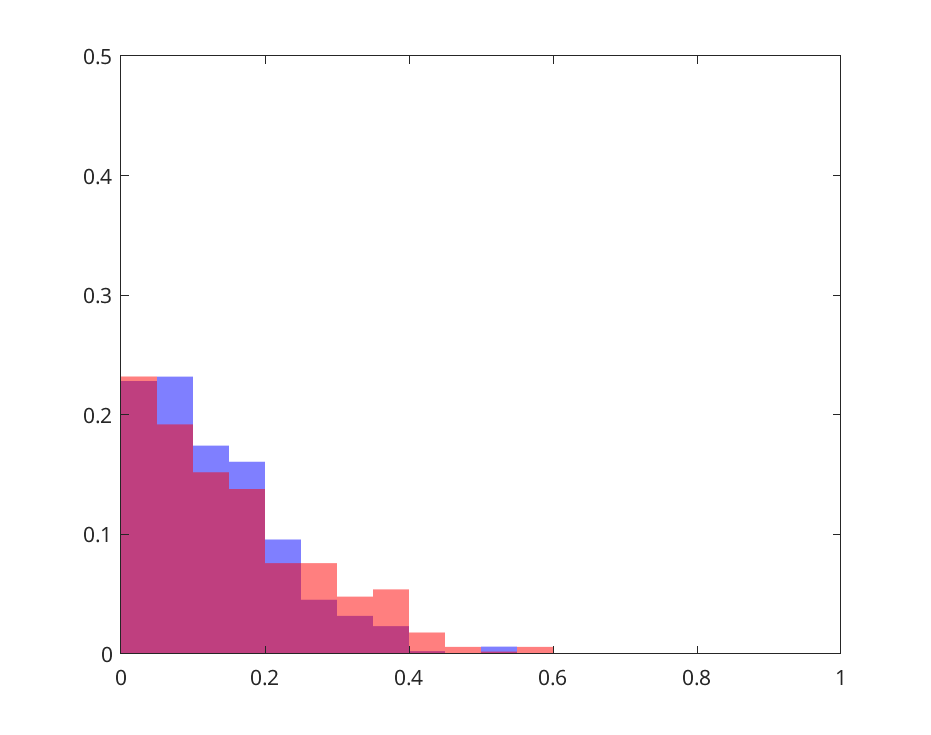}
  \includegraphics[width=5cm]{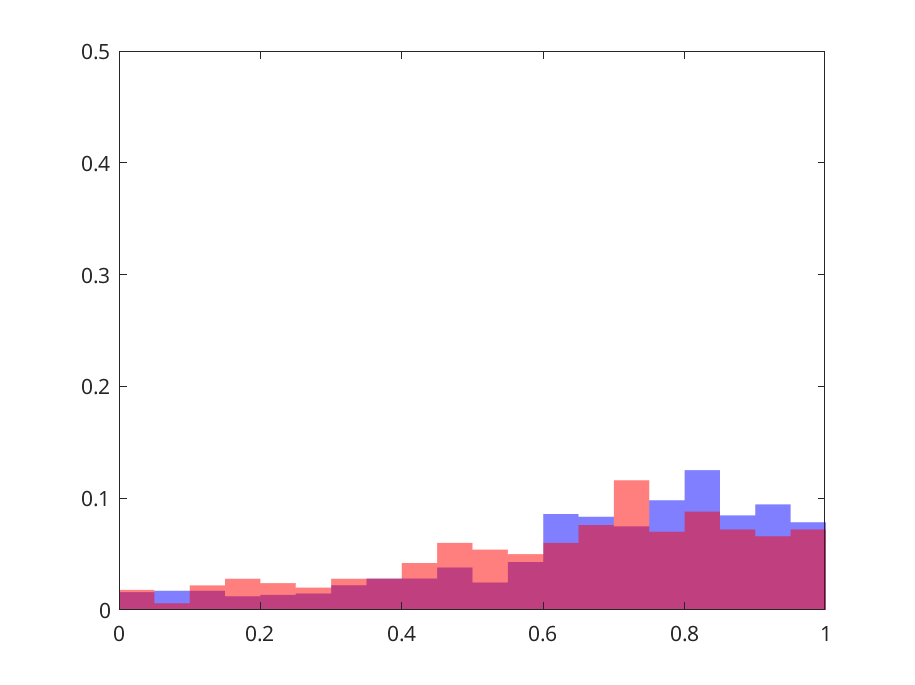}
  \includegraphics[width=5cm]{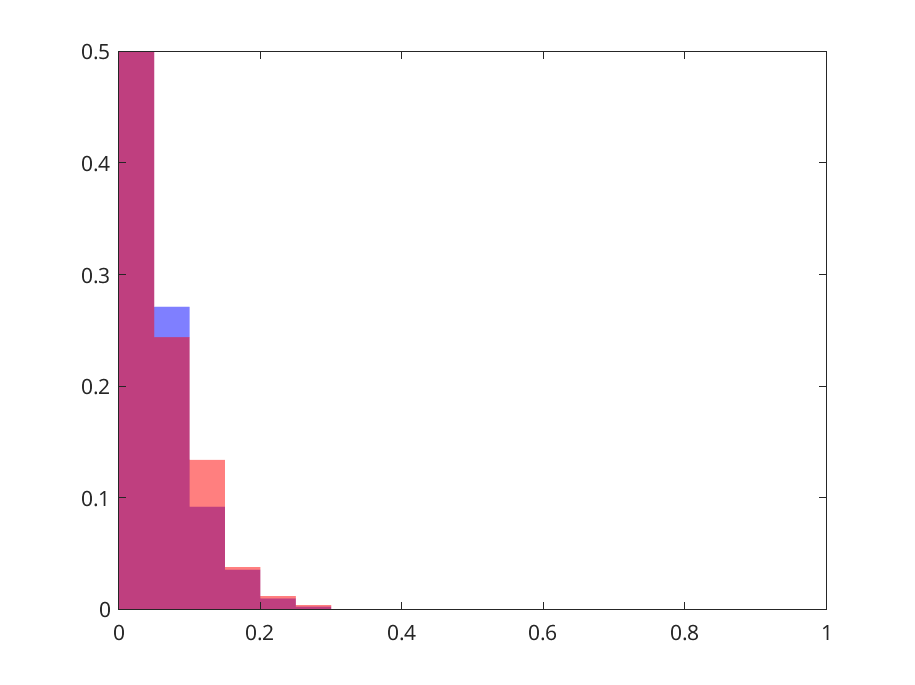}\\
  \includegraphics[width=5cm]{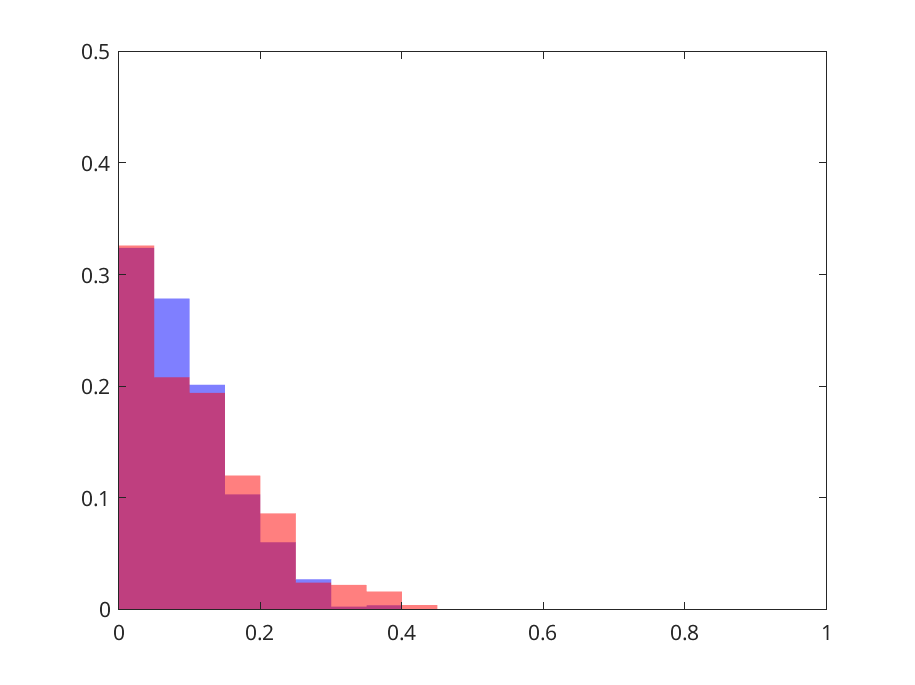}
  \includegraphics[width=5cm]{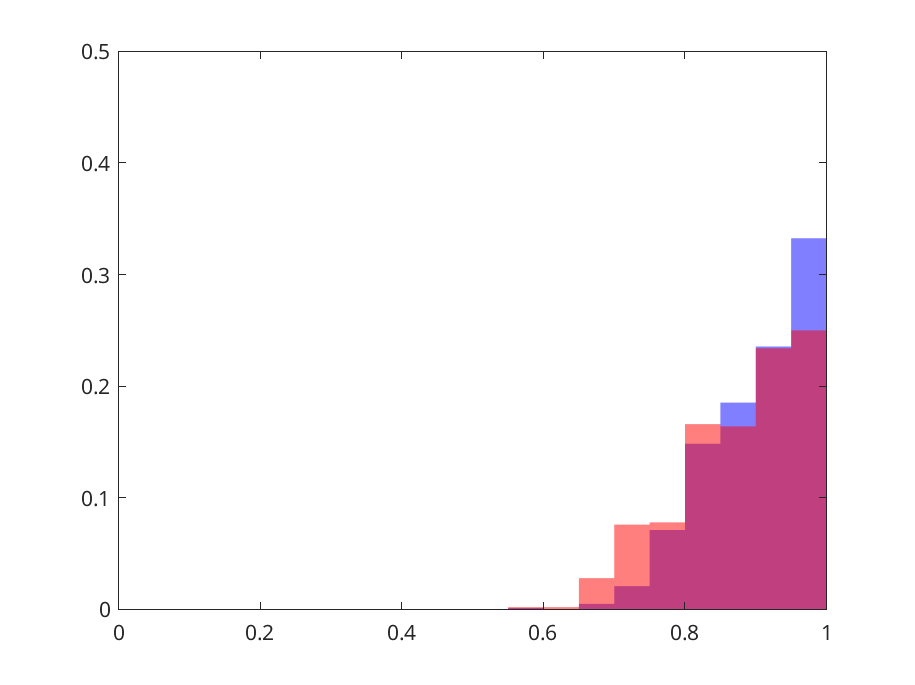}
  \includegraphics[width=5cm]{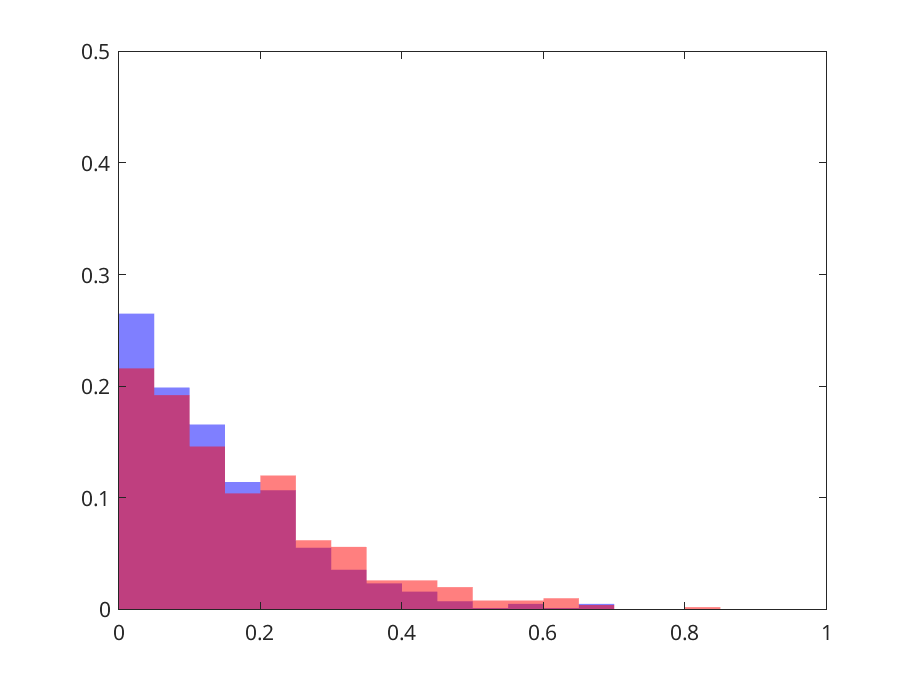}\\
  \includegraphics[width=5cm]{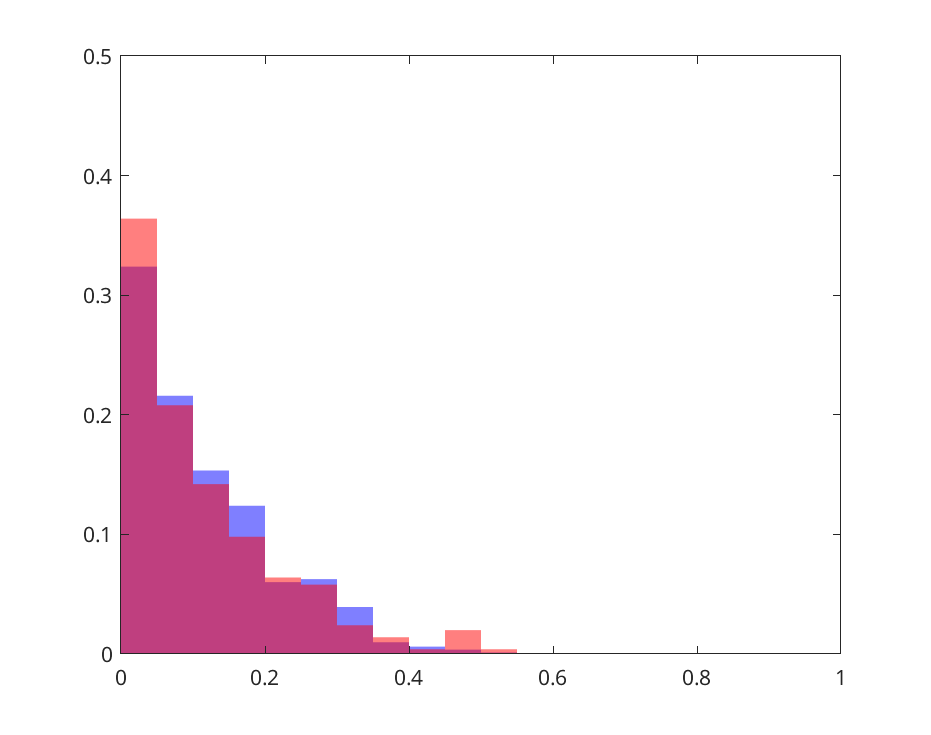}
  \includegraphics[width=5cm]{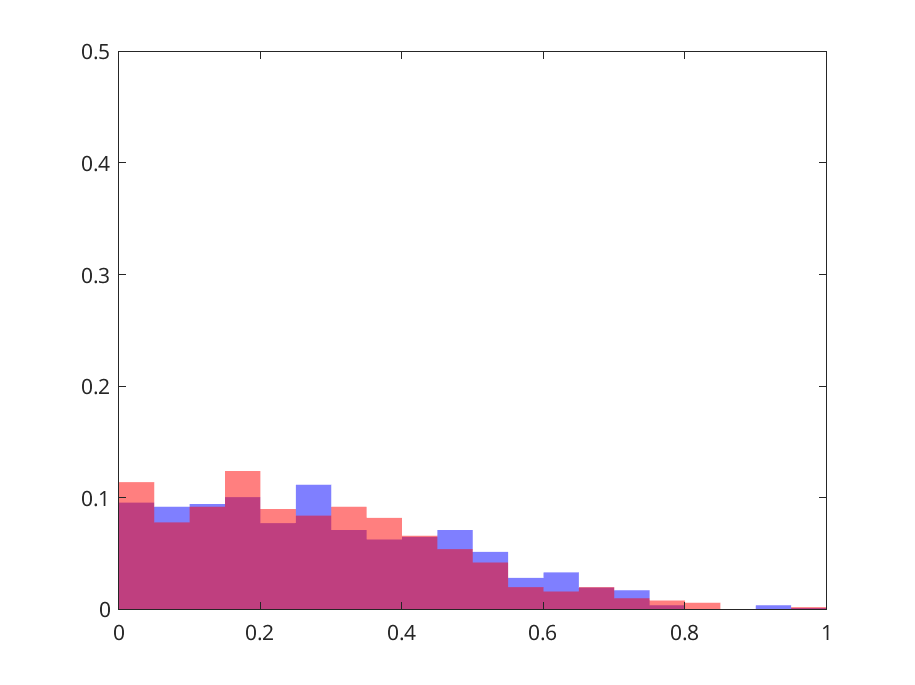}
  \includegraphics[width=5cm]{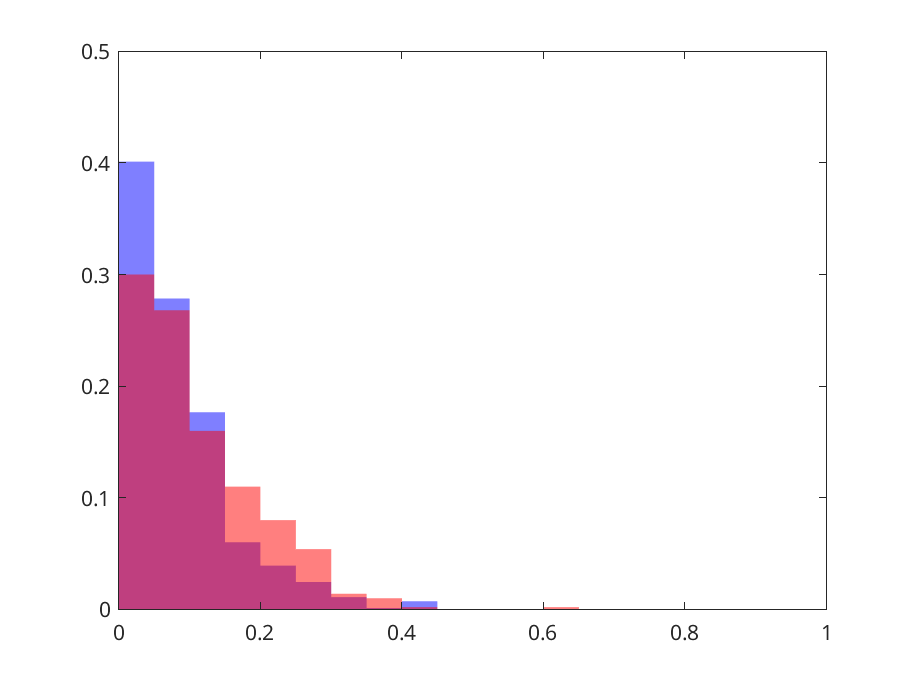}

  \caption{Marginal distributions of $1000$ randomly sampled points in $\Gf$ (blue) and the last $5$ points sampled by the QSI-EEM criterion (red) for a $100$ independent runs. Increasing coordinate number from top to bottom, left to right.}
  \label{eem:fig:rotor_histo}
\end{adjustwidth} 
\end{figure}

\begin{figure} 
  \center
  \includegraphics[width=13cm]{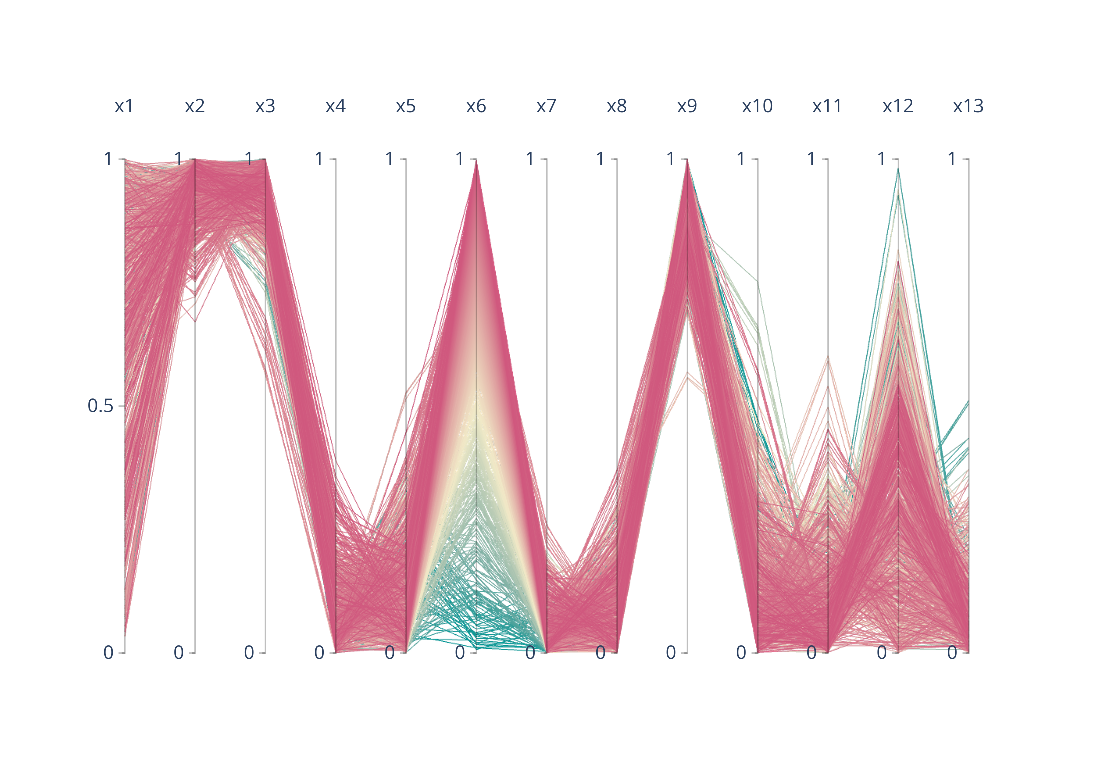}\\
  \includegraphics[width=13cm]{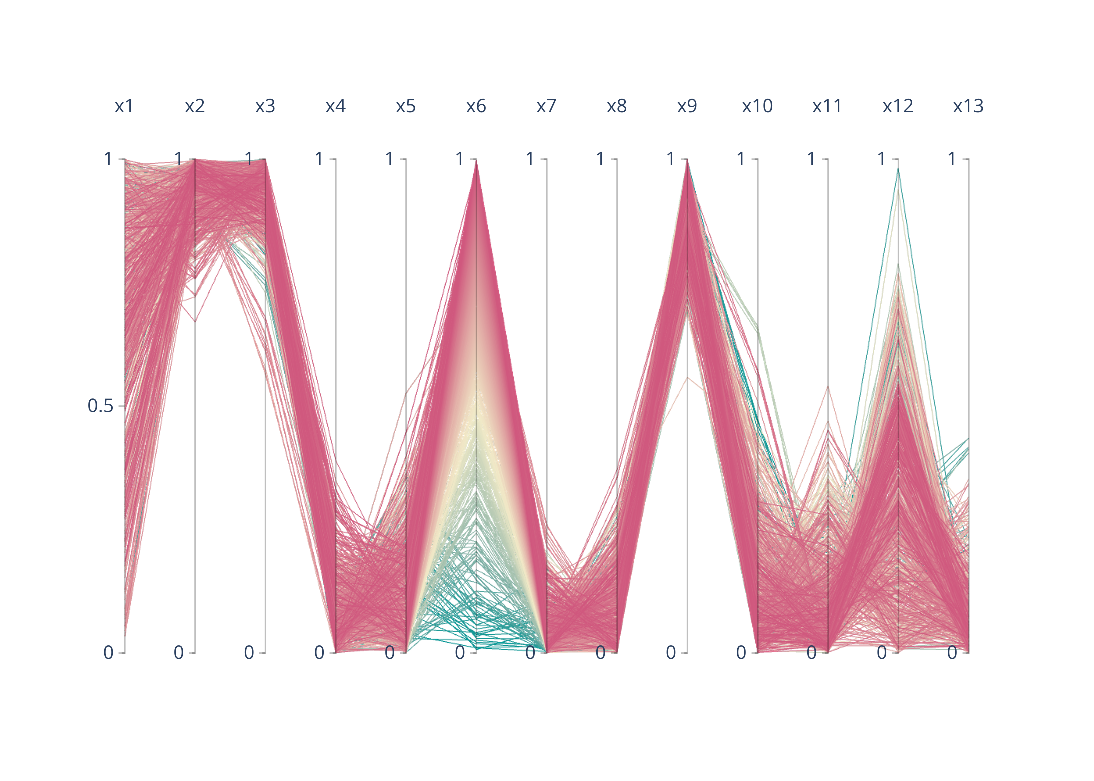}
  \caption{Parallel coordinates plots of $1000$ random points sampled in $\Gf$ (top), and $1000$ points sampled in $\hGn$ for a given run (bottom).}
  \label{eem:fig:parallel_rotor}
\end{figure}

\newpage

\section{Conclusion}
\label{eem:sec:conclusion}

We introduced a Bayesian strategy for a robust set inversion problem, known as quantile set inversion (QSI), where the objective is to estimate the set of deterministic inputs for which the probability---with respect to uncertain inputs---that the simulator output lies in a specified region is below a given threshold.

\medbreak
Through a new perspective referred to as Expected Estimator Modification (EEM), we derived a Bayesian sequential strategy, based on Gaussian process modeling of the simulator, to tackle the QSI problem. This method can be seen as a simplified approximation of the QSI-SUR strategy developed in \cite{ait:2024:qsi}. The relevance of the QSI-EEM method is demonstrated on several moderately difficult synthetic problems, where it achieves comparable performance to QSI-SUR with significantly lower computational cost.

\medbreak

We then considered more challenging QSI problems where the target set is small relative to the full deterministic input space. By combining QSI-EEM with a sequential Monte Carlo framework, we demonstrated the efficient estimation of quantile sets with relative size on the order of $10^{-6}$ to $10^{-8}$, both on synthetic examples and on the ROTOR37 compressor model. Thanks to its moderate computational complexity, the strategy supports the construction of batch designs of experiments, enabling the practical use of parallel evaluations of expensive simulators.

\medbreak
Future work will focus on extending the strategy to broader classes of QSI problems. As currently implemented, it remains limited in applicability to certain classes of QSI problems. In particular, extensions could target high-dimensional input spaces or the case of very small probability thresholds $\alpha$. On the theoretical side, further work could investigate conditions under which EEM methods are consistent or well-posed. Despite promising empirical results, the QSI-EEM strategy currently lacks formal convergence guarantees.

\section*{Acknowledgments}
The authors acknowledge the support of the French National Research Agency (ANR), which funded this work in the context of the SAMOURAI project (ANR-20-CE46-0013).
\medbreak
The authors are grateful to Sébastien Da Veiga (ENSAI, formerly Safran Tech) and Safran Tech for providing the metamodel of the ROTOR37 model used in \autoref{eem:sec:rotor}.

\newpage

\bibliographystyle{siamplain}
\bibliography{refs}

\newpage
\section*{Appendix}

\newcounter{main_figure}
\setcounter{main_figure}{\value{figure}}
\renewcommand\thefigure   {APP.\fpeval{\value{figure}-\value{main_figure}}}

\newcounter{main_equation}
\setcounter{main_equation}{\value{equation}}
\renewcommand\theequation   {APP.\fpeval{\value{equation}-\value{main_equation}}}

\newcounter{main_table}
\setcounter{main_table}{\value{table}}
\renewcommand\thetable   {APP.\fpeval{\value{table}-\value{main_table}}}

\newcounter{main_remark}
\setcounter{main_remark}{\value{theorem}}
\renewcommand\theremark   {APP.\fpeval{\value{remark}-\value{main_remark}}}

\setcounter{section}{0}
\renewcommand{\thesection}{\Alph{section}}

\section{Mathematical derivations}

\subsection{Proof of Proposition 2.1}
\label{eem:app:proof_sur_gain}

As in \cite{ait:2024:qsi}, notice that
\begin{equation*}
  \label{eem:app:eq_qsi_sur}
  \begin{split}
\Hn & = \int_\X\min(\pi_n(x),1-\pi_n(x))\,\dx\\ 
& = \int_\X\left(\pi_n(x)\mathds{1}_{\Gamma_n^{\mathrm{B},c}}(x) + (1-\pi_n(x))\mathds{1}_{\Gamma_n^\mathrm{B}}(x)\right) \, \dx\\
  \end{split}
\end{equation*}
with $\Gamma_{n}^{\mathrm{B},c}$ the complement of $\Gamma_{n}^{\mathrm{B}}$ in $\X$.
For simplicity, and without loss of generality, we assume a batch size of $r = 1$. Let us define the noisy observation of $\xi$ at a point $u \in \Uset$ as $Z_{n+1}(u) = \xi(u) + \epsilon_{n+1}$. We note
\begin{equation*}
  \begin{split}
\pi_{n+1}(x \, | u) & = \Pn(x\in\Gxi \, | Z_{n+1}(u))\\
& = \En(\mathds{1}_{\Gxi}(x) \, | Z_{n+1}(u)).
  \end{split}
\end{equation*}
We have that
\begin{equation}
  \label{eem:app:HNNU}
\begin{split}
\En(\Hnn \, | \, U_{n+1} = u) & = \int_\X \En(\min(\pi_{n+1}(x), 1-\pi_{n+1}(x))\, | \, U_{n+1}=u)\, \dx \\
& = \int_\X \En(\min(\pi_{n+1}(x \,| \, u), 1 - \pi_{n+1}(x \, | \, u)))\, \dx\\
& = \int_\X \En(\pi_{n+1}(x \,| \, u)\mathds{1}_{\Gamma_{n+1}^{\mathrm{B},c}(u)}(x) + (1-\pi_{n+1}(x \,| \, u))\mathds{1}_{\Gamma_{n+1}^{\mathrm{B}}(u)}(x))\, \dx,
\end{split}
\end{equation}
where $\Gamma_{n+1}^{\mathrm{B}}(u) = \{x \in \X \, : \, \pi_{n+1}(x \, | \, u) > 1/2\}$.
Similarly, using that
\begin{equation*}
\begin{split}
\pi_n(x)  = \En(\En(\mathds{1}_{\Gxi} (x) \, | Z_{n+1}(u))) = \En(\pi_{n+1}(x \, | \, u)),
\end{split}
\end{equation*}
we obtain 
\begin{equation}
  \label{eem:app:HNU}
  \begin{split}
\Hn = \int_\X \En(\pi_{n+1}(x \, | \, u)\mathds{1}_{\Gamma_{n}^{\mathrm{B},c}}(x) + (1-\pi_{n+1}(x \, | \, u))\mathds{1}_{\Gamma_{n}^{\mathrm{B}}}(x))\,\dx.
  \end{split}
\end{equation}
Consider the information gain $G_n(u) = \Hn - \En(\Hnn \, | \, U_{n+1} = u)$, and notice that maximizing $G_n$ is equivalent to minimizing the QSI-SUR criterion $J_n(u) = \En(\Hnn \, | \, U_{n+1}=u)$. Using \autoref{eem:app:HNNU} and \autoref{eem:app:HNU}, we have
\begin{equation*}
G_n(u) = \int_\X\En(\Delta(x \, | \, u))\,\dx,
\end{equation*}
where
\begin{equation*}
  \begin{split}
\Delta(x\,|\,u) & = \pi_{n+1}(x \, | \, u)(\mathds{1}_{\Gamma_n^{\mathrm{B},c}}(x) - \mathds{1}_{\Gamma_{n+1}^{\mathrm{B},c}(u)}(x)) + (1-\pi_{n+1}(x \, | \, u))(\mathds{1}_{\Gamma_n^\mathrm{B}}(x)-\mathds{1}_{\Gamma_{n+1}^\mathrm{B}}(x)).
  \end{split}
\end{equation*}
For all $x \in \X$, we observe that 
\begin{equation*}
  \Delta(x\, | \, u) =
  \begin{cases}
  2\pi_{n+1}(x\,|\,u) - 1 & \text{if} \quad x \in \Gamma_{n+1}^\mathrm{B}(u)\cap\Gamma_{n}^{\mathrm{B},c},\\
  1 - 2\pi_{n+1}(x \, | \, u) & \text{if} \quad x \in \Gamma_{n}^\mathrm{B}\, \cap \,\Gamma_{n+1}^{\mathrm{B}, c}(u),\\
  0 & \text{otherwise}.
  \end{cases}
\end{equation*}
 To conclude, it suffices to notice that, by definition of $\Gamma_{n+1}^B(u)$, $x \in \Gamma_{n+1}^B(u)$ if and only if $\pi_{n+1}(x \, | \, u) > 1/2$. Therefore,
\begin{equation*}
  \begin{split}
\Delta(x \, | \, u) & = 
\begin{cases}
  |2\pi_{n+1}(x \, | \, u) - 1| & \text{if} \quad (x \in \Gamma_{n+1}^\mathrm{B}(u)\cap\Gamma_{n}^{\mathrm{B},c}) \quad \text{or} \quad (x \in \Gamma_{n+1}^{\mathrm{B},c}(u)\cap\Gamma_{n}^{\mathrm{B}}),\\
  0 & \text{otherwise},
  \end{cases}\\
  & = |2\pi_{n+1}(x \, | \, u) - 1|\, \cdot \,\mathds{1}_{\Gamma_{n+1}^\mathrm{B}(u)\Delta \Gamma_{n}^\mathrm{B}}(x).
\end{split}
\end{equation*}
Hence,
\begin{equation*}
\begin{split}
G_n(u) & = \int_\X\En\left(|2\pi_{n+1}(x \, | \, u) - 1|\, \cdot \,\mathds{1}_{\Gamma_{n+1}^\mathrm{B}(u)\Delta \Gamma_{n}^\mathrm{B}}(x)\right)\,\dx\\
& = \int_\X\En\left(|2\pi_{n+1}(x) - 1|\, \cdot \, \mathds{1}_{\Gamma_{n+1}^\mathrm{B}\Delta \Gamma_{n}^\mathrm{B}}(x) \, \middle| \, U_{n+1} = u\right)\,\dx,
\end{split}
\end{equation*}
which concludes the proof.

\subsection{Proof of Proposition 2.2}
\label{eem:app:proof_approx}

  Under the measure $\Pn$, $\xi$ is also a GP with mean $\mu_n$ and covariance $k_n$ as the conditioning of a GP given a finite number of values. Hence, given a set of points $U \in \Uset$ and their observations $\xi(U) = (\xi(U_{n+1}),\dotsc, \xi(U_{n+r}))^t$, the Matheron's update rule indicates that the posterior mean of $\xi$ given $\sigma(\In\cup\{\xi(U)\})$ is
  \begin{align*}
  \mu_{n+r}(u) = \mn(u) + k_n(U,u)\Sigma_n^{-1}(U)(\xi(U)-\mn(U)).
  \end{align*}
  It is easy to see that under the measure $\Pn$, $k_n(U,u)\Sigma_n^{-1}(U)(\xi(U)-\mn(U)) \sim \mathcal{N}(0, \kappa^2_n(u))$. Hence, under $\Pn$,
  \begin{align*} \mu_{n+r}(u) \overset{\mathrm{d}}{=} \mn(u) + \kappa_n(u)Z,
  \end{align*}
  with $Z \sim \mathcal{N}(0,1)$. We conclude using that, under $\Pn$, $\mn(\cdot) + \kappa_n(\cdot)Z$ and $\mu_{n+r}$ are both GPs with similar mean and covariance functions.

  \subsection{Equivalences EEM / SUR}
  \label{eem:app:eq_eem_sur}
  
   In this section, we give the mathematical derivations of the equivalences between several SUR methods et EEM methods. Without loss of generality, we assume a batch size of $r=1$. Let us first remind the law of total variance.

\medbreak

\begin{lemma}[Law of the total variance]
  Let $\mathcal{F}_1$ and $\mathcal{F}_2$ two $\sigma$-algebras such that $\mathcal{F}_1 \subset \mathcal{F}_2$, and denote $\E_i$ and $\Var_i$ the conditional expectation and variance operators associated with $\mathcal{F}_i$.\medbreak
  For any square integrable random variable $Z$, we have
  \begin{align*}
  \Var_1(Z) = \E_1(\Var_2(Z)) + \Var_1(\E_2(Z)). 
  \end{align*}
\end{lemma} 
\medbreak
In the following, we note $G_n(u) = \Hn - \En(\Hnn \, | \, U_{n+1} = u)$ the expected gain function (given $\In$) associated to a SUR criterion $J_n(u) = \En(\Hnn \, | \, U_{n+1} = u)$. Notice that maximizing the $G_n$ is equivalent to minimizing $J_n$.

\medbreak

\paragraph{Integrated variance}
By the law of total variance, we have that 
\begin{equation}
  \begin{split}
G_n(\breve{u}) & = \int_\Uset\sn^2(u) \, \du - \int_\Uset\En(\sigma_{n+1}^2(u) \, | U_{n+1} = \breve{u})\,\du\\
& = \int_\Uset\Var_n(\Enn(\xi(u))\, | \, U_{n+1} = \breve{u})\, \du\\
& = \int_\Uset\Var_n(\mu_{n+1}(u)\, | \, U_{n+1} = \breve{u})\, \du\\
& = \int_\Uset\En((\mu_{n+1}(u)-\mu_n(u))^2  \, | \, U_{n+1} = \breve{u})\, \du,  
  \end{split}
\end{equation}
the last equality being deduced from the Matheron update formula, which implies that under $\Pn$ (and given $U_{n+1} = \breve{u}$), $\mu_{n+1}$ is a GP with mean $\mu_n$.

\medbreak

\paragraph{Variance of the volume of the excursion} 
Similarly to the derivation for the integrated variance criterion
\begin{equation}
  \begin{split}
    G_n(\breve{u}) & = \Var_n(\P_\Uset(\Lambda(\xi))) - \En(\Var_{n+1}(\P_\Uset(\Lambda(\xi)))\,|\, U_{n+1} = \breve{u})\\
    & = \Var_{n}(\Enn(\P_\Uset(\Lambda(\xi)))\,|\, U_{n+1} = \breve{u})\\
    & = \Var_n(\gamma_{n+1} \,|\, U_{n+1} = \breve{u})\\
    & = \En((\gamma_{n+1}-\gamma_n)^2 \,|\, U_{n+1} = \breve{u}),
  \end{split}
\end{equation}
where $\gamma_n = \int_\Uset \P_n(u \in \lambda(\xi))\,\du$.

\medbreak

  \paragraph{Expected improvement} 
  From \cite{bect:2019:supermartingale}, we know that choosing an evaluation point according to the expected improvement can be formulated as a SUR criterion. More specifically, it is equivalent to minimizing the criterion
\begin{equation}
  J_n(\breve{u}) = \En(\Hn\, | \, U_{n+1} = \breve{u}),
\end{equation}
where $\Hn = \En(\max\,\xi - M_n)$. Hence, we have
\begin{equation}
  \begin{split}
G_n(\breve{u}) & = \En(\max\,\xi - M_n) - \En(\Enn(\max\,\xi - M_{n+1}) \, | \, U_{n+1} = \breve{u})\\
& = \En(\max\,\xi - M_n - \max\, \xi + M_{n+1}\, | \, U_{n+1} = \breve{u})\\
& = \En(|M_{n+1}-M_n| \, | \, U_{n+1} = \breve{u}),
  \end{split}
\end{equation}
where the absolute value arises from the fact that, by definition, $M_{n+1} \ge M_n$.

\section{Metropolis-Hasting with adaptive variance}
\label{eem:app:MH}
In this section, we give details regarding the move step. We use a Metropolis-Hastings algorithm \citep{hasting:1970:MH}, with a Gaussian kernel whose variance is adapted in the same way as in \cite{bect:2017:bss}.
\medbreak
Denote by $(X_j^0)_j$ the points obtained after the resampling step and $q_n^k$ the density targeting $\Gamma^k(f)$ after $n$ evaluations. We repeat the following instructions for a given number of steps. First, we generate new candidates points by applying a Gaussian perturbation
\begin{equation}
\bar{X_j^i} = X_j^{i-1} + \Sigma_{k}^iZ_j,
\end{equation}
with $Z_j$ independent standard Gaussian random vectors of length $d_\X$ and $\Sigma_k^i$ a diagonal matrix. Then, we set
\begin{equation}
X_j^i = \begin{cases}
  \bar{X_j^i} \quad \text{with probability} \quad \mathds{1}_\X(\bar{X_j^i})\min\left(1, \frac{q_{n}^k(\bar{X_j^i})}{q_n^k(X_j^{i-1})}\right)\\
  X_j^{i-1} \quad \text{with probability} \quad 1-\mathds{1}_\X(\bar{X_j^i})\min\left(1, \frac{q_{n}^k(\bar{X_j^i})}{q_n^k(X_j^{i-1})}\right)
\end{cases}
\end{equation}
Consider now $A_j = \frac{1}{m}\sum_{i=1}^m\mathds{1}_{X_j^i=X_j^{i-1}}$ the empirical acceptation rate, and $\sigma_{k,l}^i$ the $l$-th diagonal element of $\Sigma_{k,j}$. In order to calibrate the variance of the Gaussian perturbations, and hence sequentially adapt to sampling in smaller quantile sets, we set 
\begin{equation}
  \log\sigma_{k,j}^{i+1} = 
  \begin{cases}
    \log(\sigma_{k,j}^i)+\epsilon \quad \text{if} \quad A_j > A_{target}\\
    \log(\sigma_{k,j}^i)-\epsilon \quad \text{if} \quad A_j \le A_{target},
  \end{cases}
  \end{equation}
where $A_{target}$ is a specified target acceptation rate. 

\medbreak
In our implementation, for all $l=1,\dotsc,d_\X$, we fix $\sigma_{0,l}^0$ as the standard deviation of the uniform distribution over the $l$-th dimension of $\X$ and $\epsilon = \log2$

\medbreak 
Note that this adaptation scheme is quite rudimentary, but sufficient for the examples displayed in this article. In the context of hard shaped quantile sets, for example sets concentrated around a subspace, more complex adaptation schemes, notably full covariance adaptation \citep[see, e.g.,][]{haario:2001:adapt_MH}, could prove useful. For more information regarding adaptive SMC strategies, the interested reader can refer, for instance, to \cite{roberts:2009:example_smc,vihola:2010:thesis,utkin:2020:adapt_smc} and references therein.
\medbreak

\section{Complements to the experiments of Section 4}
\subsection{Implementation details}
\label{eem:app:details_implementation}
 We provide here details regarding the implementation of the QSI-SUR criterion and other competitor methods used in \autoref{eem:sec:sanity_check}. We also precise the implementation of the QSI-EEM criterion for these experiments.
 \medbreak

 \paragraph{Bayesian model} %
The function $f$ is represented by a GP with a constant
mean function and an anisotropic Matérn covariance. %
The parameters of the GP~prior are estimated, at each step, by
restricted maximum likelihood (ReML) method \citep{stein1999interpolation}. Notice that we enforce the constraint that the
regularity parameter $\nu$ of the kernel should belong
to~$\left\{\frac{1}{2}, \frac{3}{2}, \frac{5}{2}, +\infty \right\}$. %
To limit the occurrence of
ill-conditioned covariance matrices, a nugget of value~$10^{-6}$ is
added. %
Regarding the initial training points, we use an approximated maximin Latin Hypercube of size $n_0 = 10(d_\X+d_\S)$ \citep{loeppky2009choosing}.

\medbreak

\paragraph{Implementation of the QSI-EEM and QSI-SUR criteria}
At each step, $100$ points are sampled according to $\Ps$ as $\tS_n$. A set of
$n_\X = 500\, d_\X$ points is uniformly sampled in $\X$. %
On these points, we approach the misclassification
probability $\min(\pi_n(x), 1-\pi_n(x))$ using Monte Carlo simulations
of $\xi(x,\cdot)$ on $\tS$ (with respect to $\Pn$). A subset $\tX_n$ of $n_{\tX} = 40$~points is defined as the point with the highest misclassification probability, and $n_{\tX} - 1 = 39$~points drawn (without replacement) according to the discrete distribution $p_\X(x) \propto \min(\pi_n(x), 1-\pi_n(x))$, with $\pi_n(x) = \Pn(x\in\Gxi)$. %
The set $\tX_n\times\tS_n$ is used to approximate the integrals involved in the QSI-EEM and QSI-SUR criteria. In the case of the QSI-SUR criterion, the integrand arising is
approximated using a Gauss-Hermite quadrature of $10$~points coupled
with $100$~conditioned sample paths. The QSI-EEM is approached using a similar quadrature and the formulas of \autoref{eem:sec:QSI-EEM-approx}.
The approximated sampling criteria are optimized using a simple search over a subset of $\tX_n \times \tS_n$ of $250$
candidate points. This subset is constructed by sampling according to
$\min(p_n(x,s), 1 - p_n(x,s))$, where $p_n(u) = \Pn(\xi(u) \in C)$.%

\medbreak

\paragraph{Other competitors}
As in \cite{ait:2024:qsi}, we also compare the performance of the
QSI-EEM criterion against methods that aim at approximating the
set~$\Lambda(f)$ in the joint space~$\X \times \S$. The results are compared to the ``Joint-SUR'' criterion of
\cite{bect:2012:sur_failureproba, chevalier:2014:fast_parallel_kriging}.
We sample $5d_\X\times10^4$ points according to the uniform distribution
on $\Uset$. %
The criterion is approximated using a subset of $4000$ points
constructed from the points with the highest misclassification
probability $\min(p_n(u), 1-p_n(u))$, and a sample (without
replacement) according to this probability. %
It is then optimized using an exhaustive search on $250$ of those
points.

\medbreak
We also include the ``Ranjan'' criterion \citep{ranjan:2008:contour}
with $\kappa = 1.96$ and the misclassification probability criterion of
\citep{bryan:2005:active}. Those are evaluated on $5d_\X\times10^4$ points
sampled uniformly on $\Uset$. %

\medbreak

As a baseline, we consider uniform random sampling in the space $\Uset$.

\subsection{Description of the test cases}
\label{eem:app:sec4_examples}
We provide here a description of the synthetic functions used in \autoref{eem:sec:sanity_check}.
\medbreak

\paragraph{Case 1 - modified Branin-Hoo function}
\begin{align*}
f_1(x,s) = \frac{1}{12}\, b(x,s) + 3\sin\left(x^{\frac{5}{4}}\right)
+ \sin\left(s^{\frac{5}{4}}\right)
\end{align*} where $b$ is the Branin-Hoo function \citep{branin:1972:branin-hoo}
\begin{align*}
  b(x,s) = \left(s - \frac{5.1x^2}{4\pi^2} + \frac{5x}{\pi} - 6
\right)^2 + 10 \left(1-\frac{1}{8\pi}\right)\cos(x) + 10.
\end{align*}
We set $C = (-\infty; T]$ with $T = 7.5$, $\alpha = 0.05$ and
$\Ps$~the Beta distribution with parameters $(7.5, 1.9)$, rescaled
from~$\left[ 0, 1 \right]$ to~$\S$. %

\medbreak

\paragraph{Case 2}
\begin{align*}f_2\left( x,s \right) = \frac{1}{2}\, c(x_1,s_1) + \frac{1}{2}\,
c(x_2,s_2)
\end{align*} on~$\X = [-2;2]^2$ and~$S = [-1;1]^2$, with $c$ the
``six-hump camel'' function \citep{dixon:1978:intro}
\begin{align*}
c(x,s) = \left(4-2.1x^2+\frac{x^4}{3}\right)x^2 + xs + (4s^2-4)s^2.\end{align*} %
We set $C = (-\infty; T]$ with $T = 1.2$, $\alpha = 0.15$ and
$\Ps$~the uniform distribution on~$\S$.

\medbreak

\paragraph{Case 3 - Hartmann4 function}
We consider the Hartman4 function
\citep{picheny:2013:benchmark}, defined on~$\X = [0;1]^2$
and~$\S = [0;1]^2$ equipped with the uniform distribution. %
We set $\alpha = 0.6$ and $C = [T, +\infty)$, with
$T = -1.1$.

\subsection{Complementary numerical results (Section 4)}
\label{eem:SM:sec:results-experiences}

In this section, some complementary details on the numerical
experiments of \autoref{eem:sec:sanity_check} are given. %
This includes, for all the competitors, the 100~sample paths and the
quantiles of order~$75\%$ and~$95\%$ of the error (proportion of
misclassified points) as a function of the number of steps.

\subsubsection{\texorpdfstring{Example 1}{Example 1}}

See Figures~\ref{eem:SM:fig:stats_results_f1_2}--\ref{eem:SM:fig:trajs_results_f1}.

\begin{figure}[p]
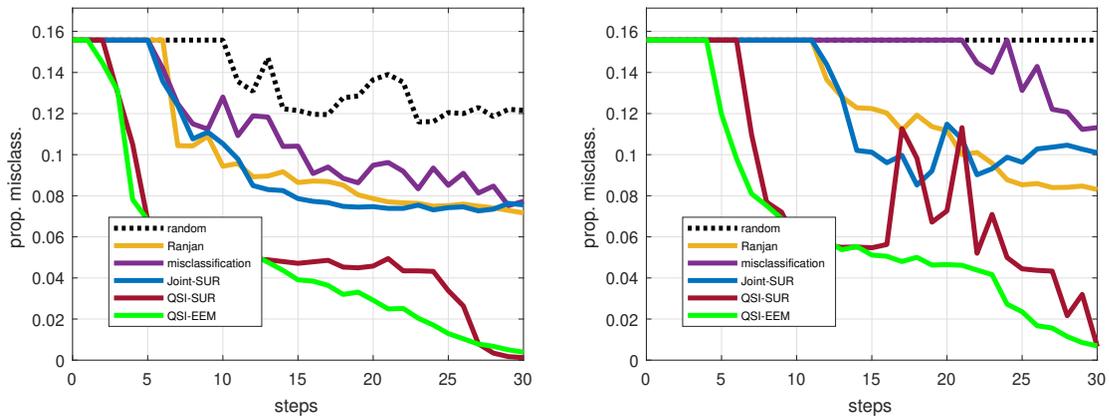

  \includegraphics[width=\toto]{graphs_eem/metric_75_branin_mod.eps}\hspace*{5mm}
  \includegraphics[width=\toto]{graphs_eem/metric_95_branin_mod.eps}
  \caption{%
    Quantiles of level~$0.75$ and~$0.95$ for the proportion of
    misclassified points vs.\ number of steps, for 100~repetitions of
    the algorithms on the test function~$f_1$.}
  \label{eem:SM:fig:stats_results_f1_2}
\end{figure}

\begin{figure}[p]
  \includegraphics[width=\toto]{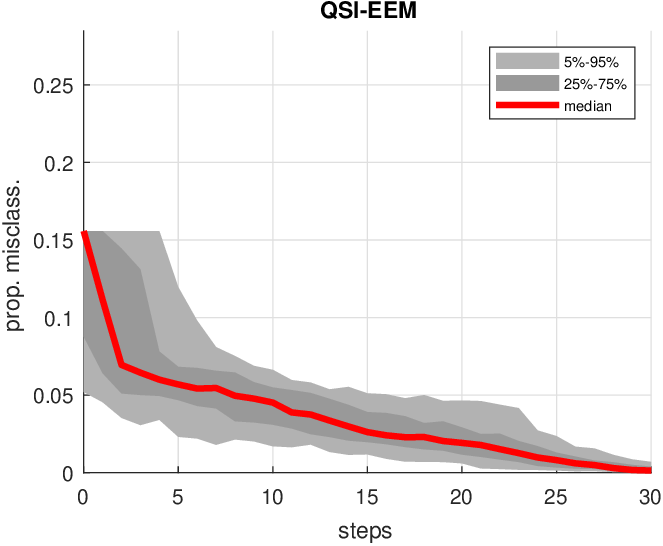}\hspace*{5mm}
  \includegraphics[width=\toto]{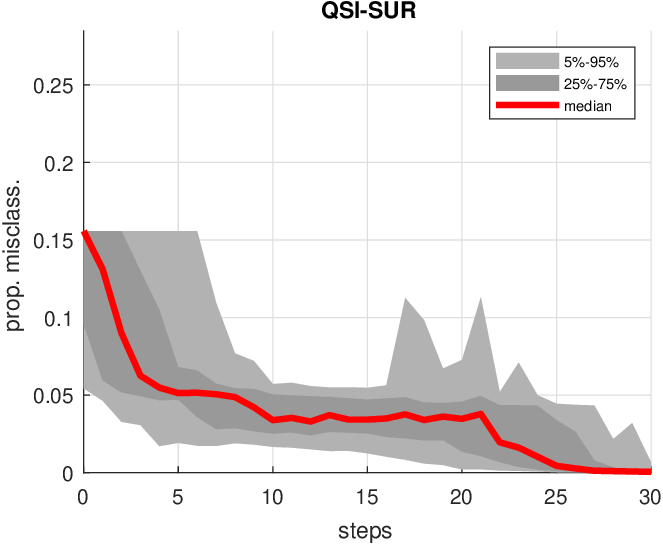}\\[5mm]
  \includegraphics[width=\toto]{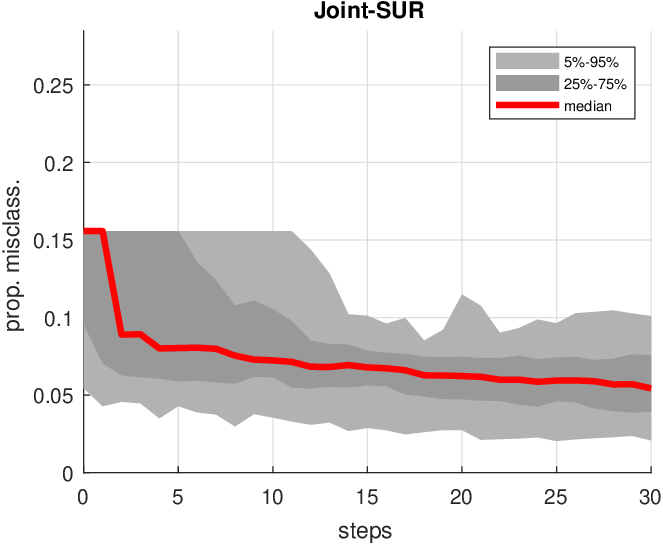}\hspace*{5mm}
  \includegraphics[width=\toto]{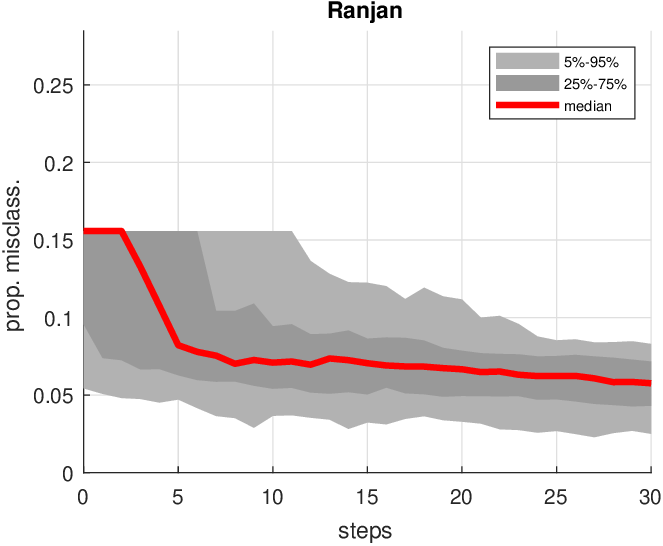}\\[5mm]
  \includegraphics[width=\toto]{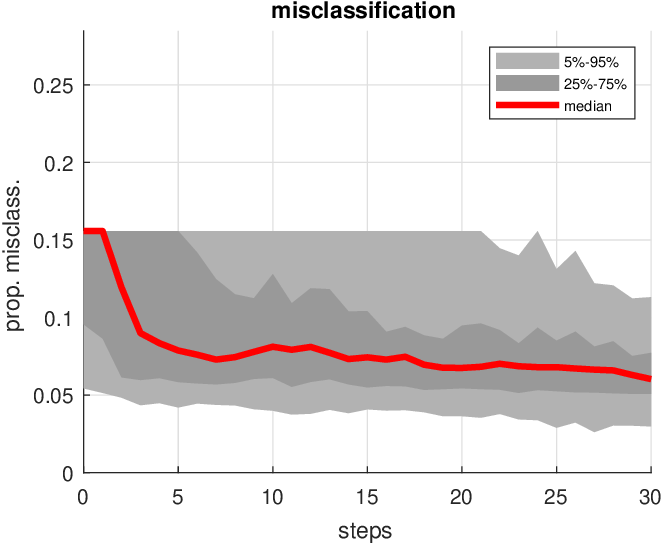}\hspace*{5mm}
  \includegraphics[width=\toto]{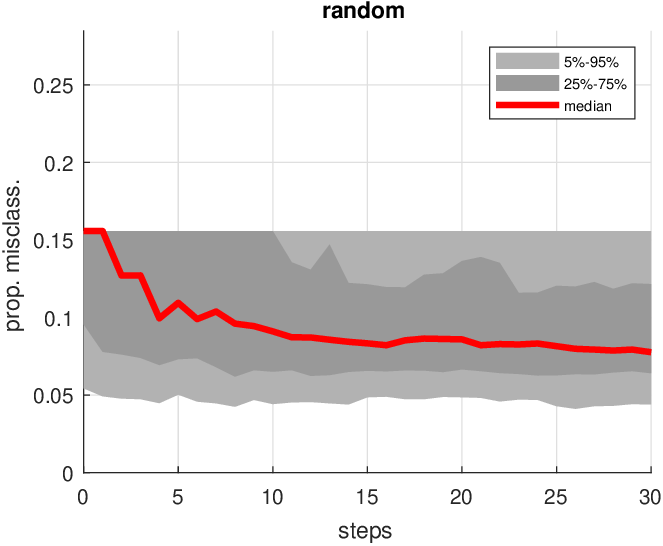}
  \caption{%
    Median and several quantiles of the proportion of misclassified
    points vs.\ number of steps, for 100~repetitions of the algorithms
    on the test function~$f_1$.}
  \label{eem:SM:fig:stats_results_f1}
\end{figure}

\begin{figure}[p]
  \includegraphics[width=\toto]{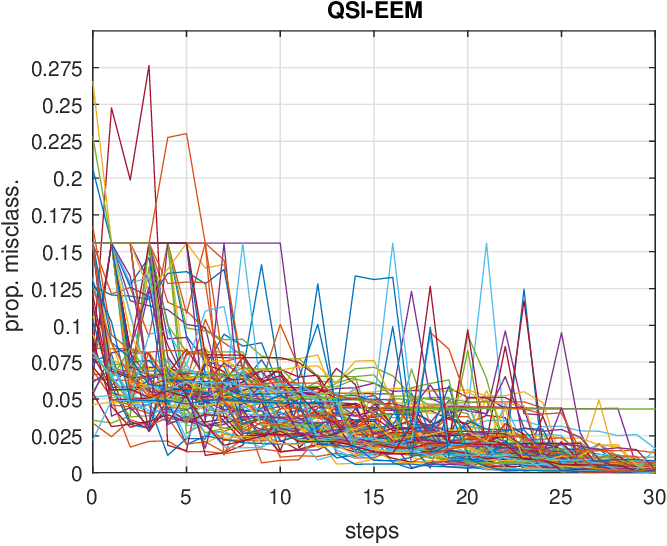}\hspace*{5mm}
  \includegraphics[width=\toto]{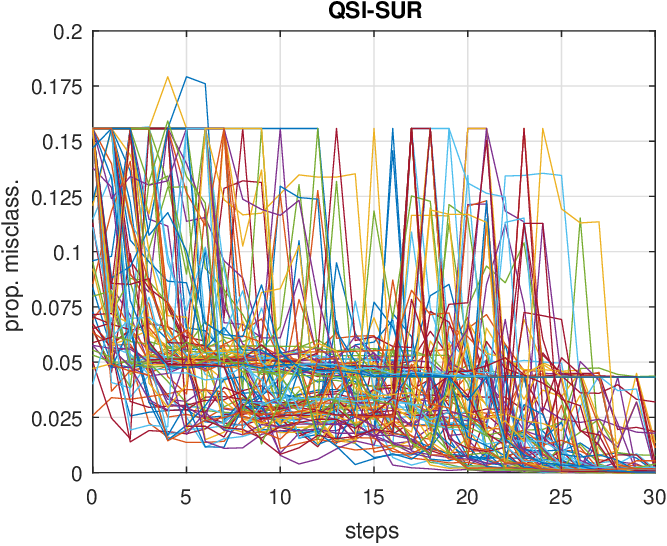}\\[5mm]
  \includegraphics[width=\toto]{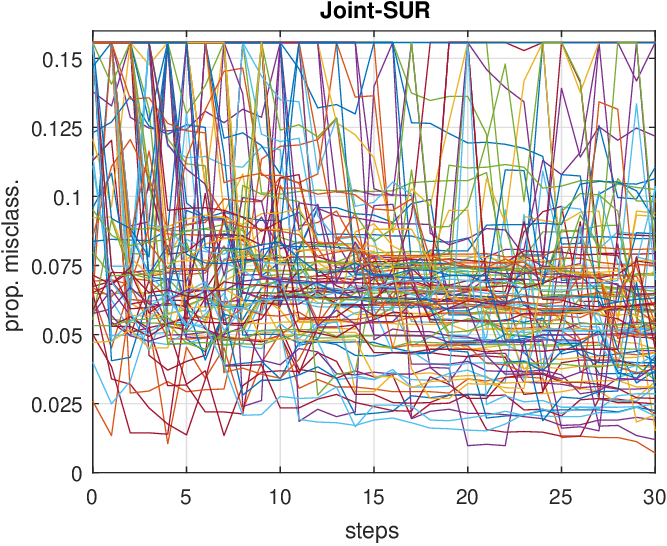}\hspace*{5mm}
  \includegraphics[width=\toto]{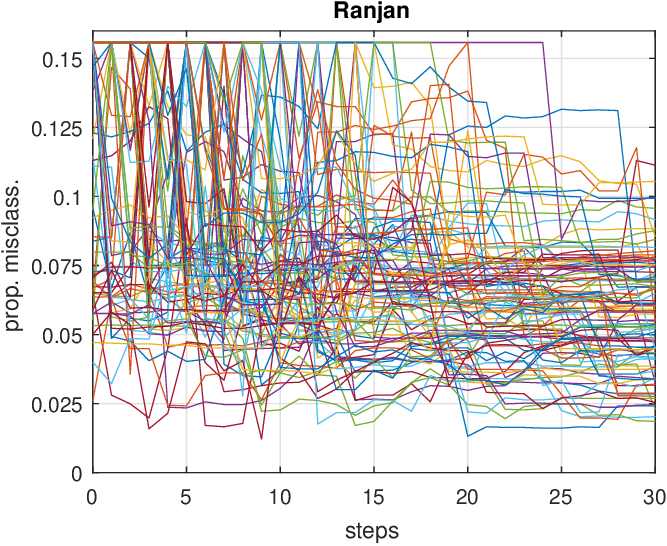}\\[5mm]
  \includegraphics[width=\toto]{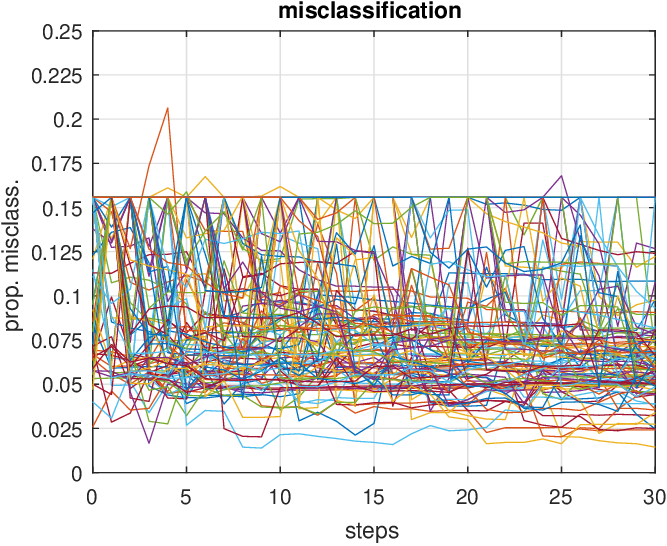}\hspace*{5mm}
  \includegraphics[width=\toto]{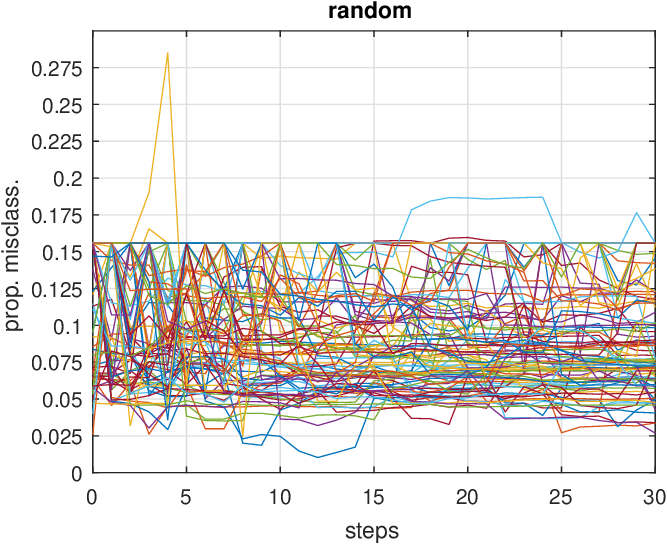}
  \caption{%
    Different sample paths of the proportion of misclassified points
    vs.\ number of steps, for 100~repetitions of the algorithms on the
    test function~$f_1$.}
  \label{eem:SM:fig:trajs_results_f1}
\end{figure}

\subsubsection{\texorpdfstring{Example 2}{Example 2}}

See Figures~\ref{eem:SM:fig:stats_results_f2_2}--\ref{eem:SM:fig:trajs_results_f2}.

\begin{figure}[p]

  \psfrag{0}     {}
  \psfrag{0.05}  {}
  \psfrag{50}    {}
  \psfrag{100}   {}

  \includegraphics[width=\toto]{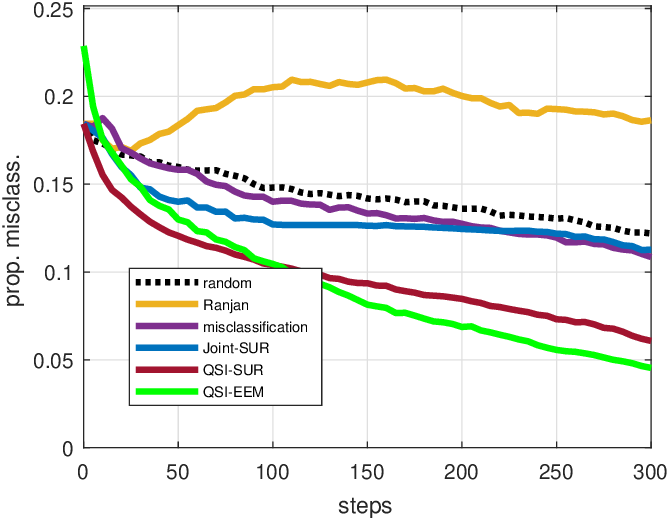}\hspace*{5mm}
  \includegraphics[width=\toto]{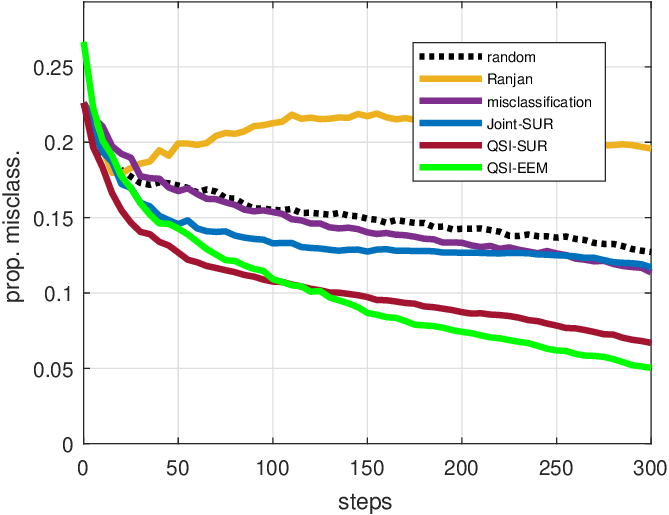}
  \caption{%
    Quantiles of level~$0.75$ and~$0.95$ for the proportion of
    misclassified points vs.\ number of steps, for 100~repetitions of
    the algorithms on the test function~$f_2$.}
  \label{eem:SM:fig:stats_results_f2_2}
\end{figure}

\begin{figure}[p]
  \includegraphics[width=\toto]{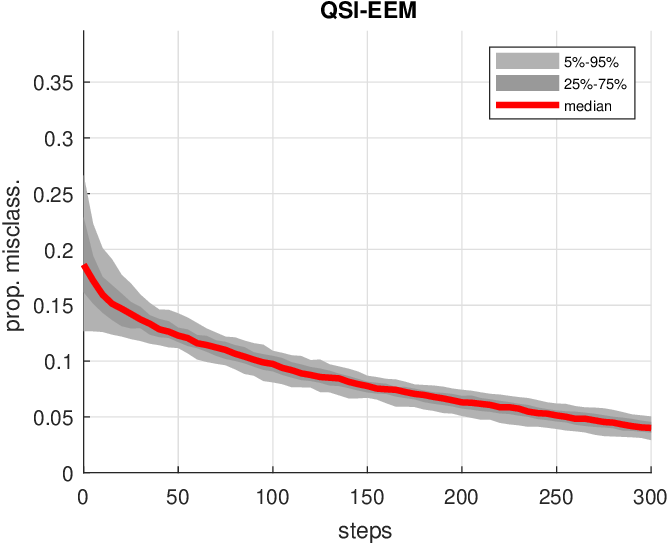}\hspace*{5mm}
  \includegraphics[width=\toto]{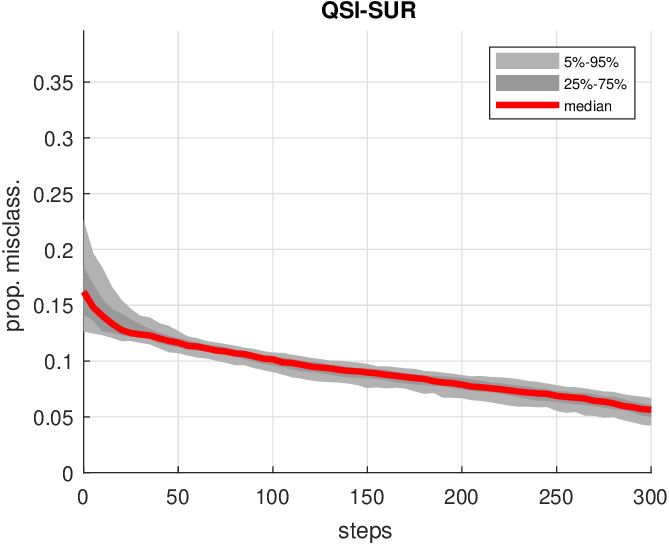}\\[5mm]
  \includegraphics[width=\toto]{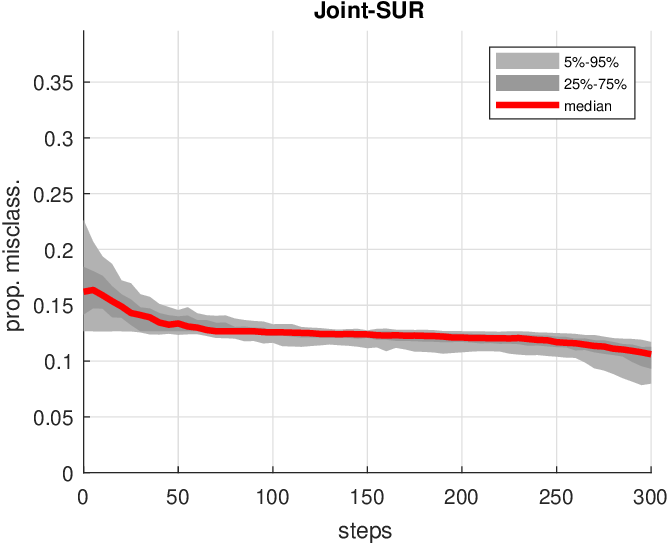}\hspace*{5mm}
  \includegraphics[width=\toto]{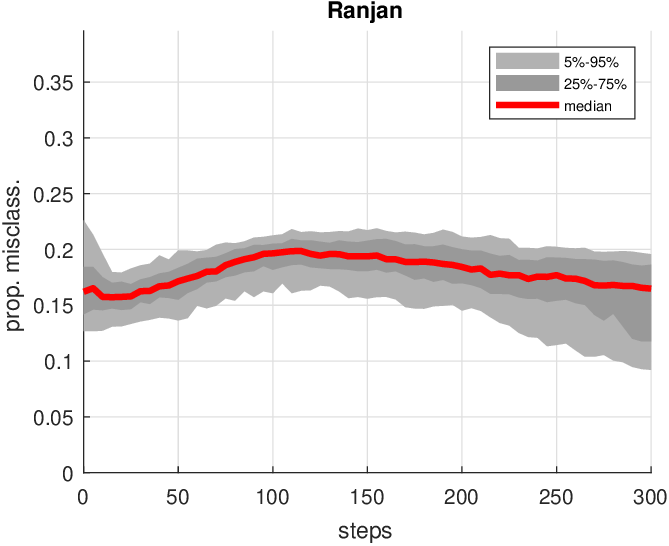}\\[5mm]
  \includegraphics[width=\toto]{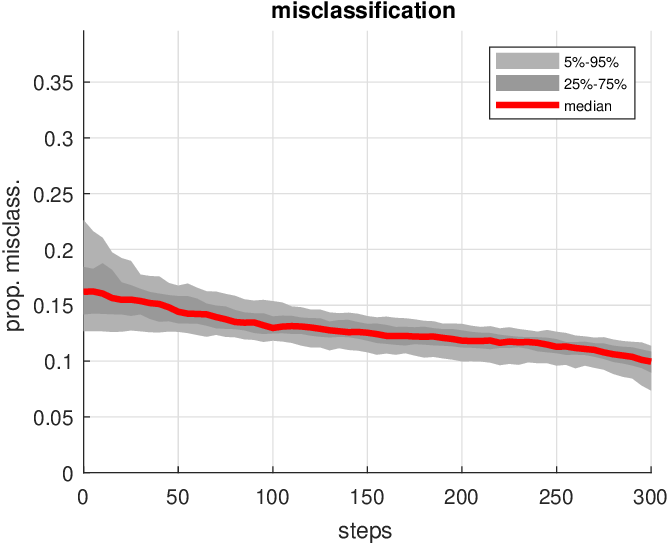}\hspace*{5mm}
  \includegraphics[width=\toto]{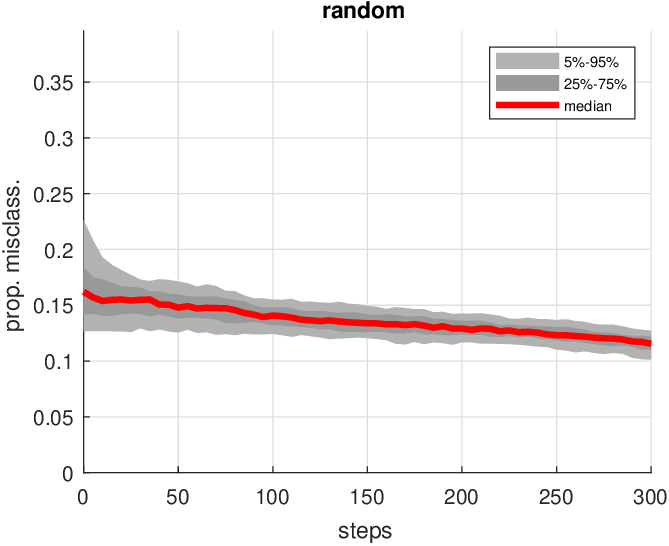}
  \caption{%
    Median and several quantiles of the proportion of misclassified
    points vs.\ number of steps, for 100~repetitions of the algorithms
    on the test function~$f_2$.}
  \label{eem:SM:fig:stats_results_f2}
\end{figure}

\begin{figure}[p]
  \includegraphics[width=\toto]{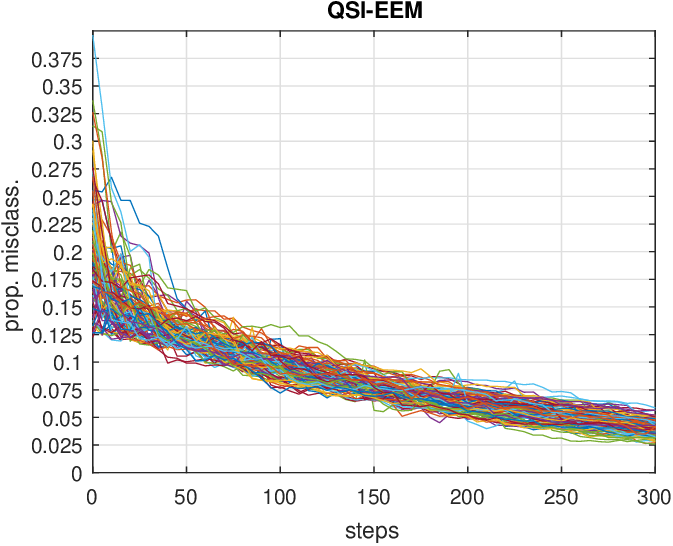}\hspace*{5mm}
  \includegraphics[width=\toto]{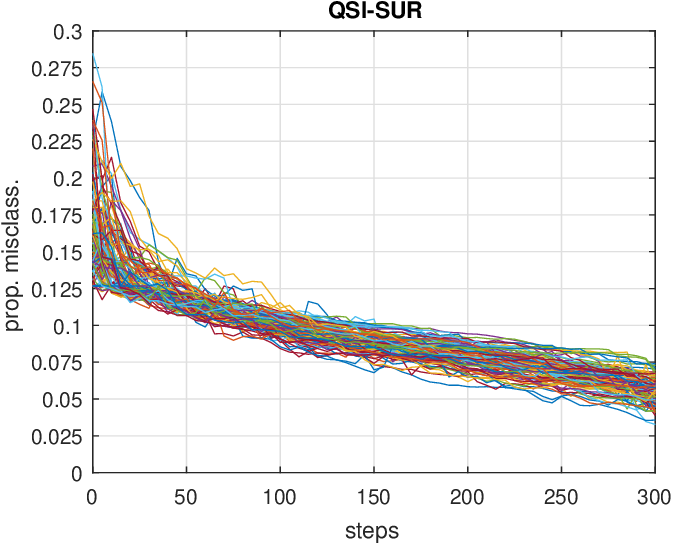}\\[5mm]
  \includegraphics[width=\toto]{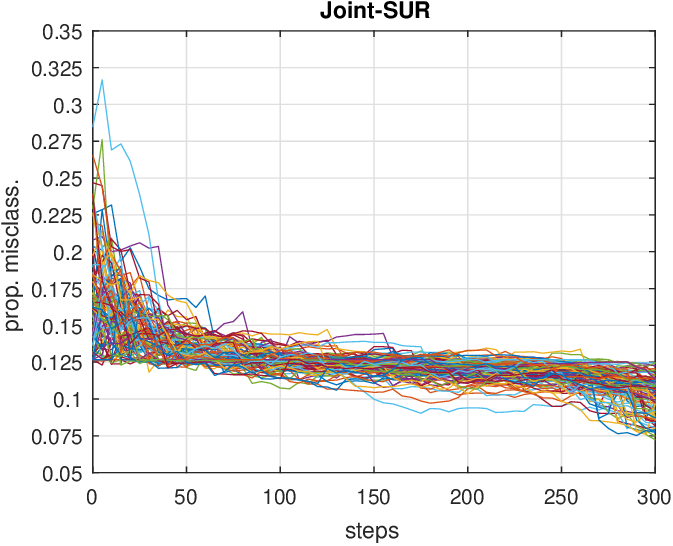}\hspace*{5mm}
  \includegraphics[width=\toto]{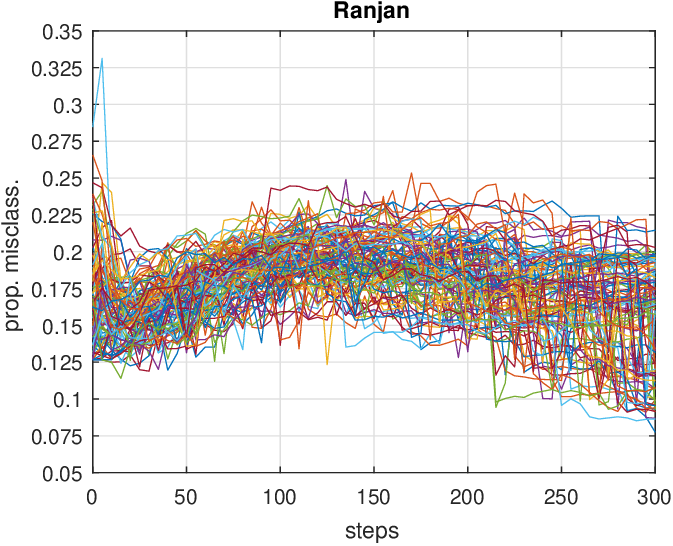}\\[5mm]
  \includegraphics[width=\toto]{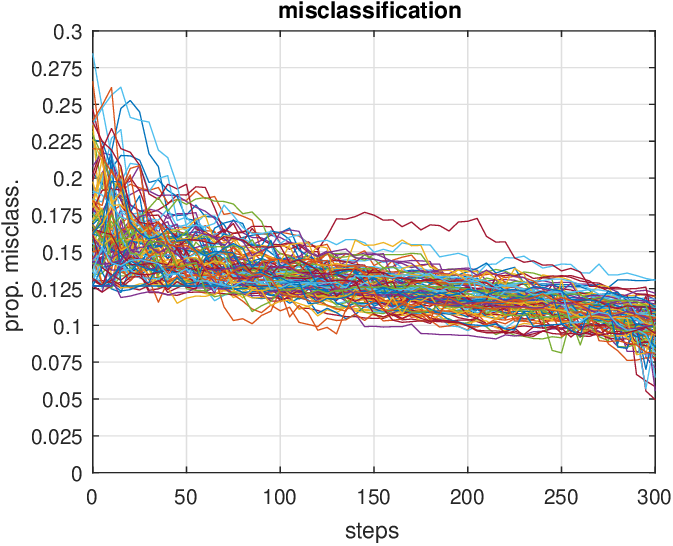}\hspace*{5mm}
  \includegraphics[width=\toto]{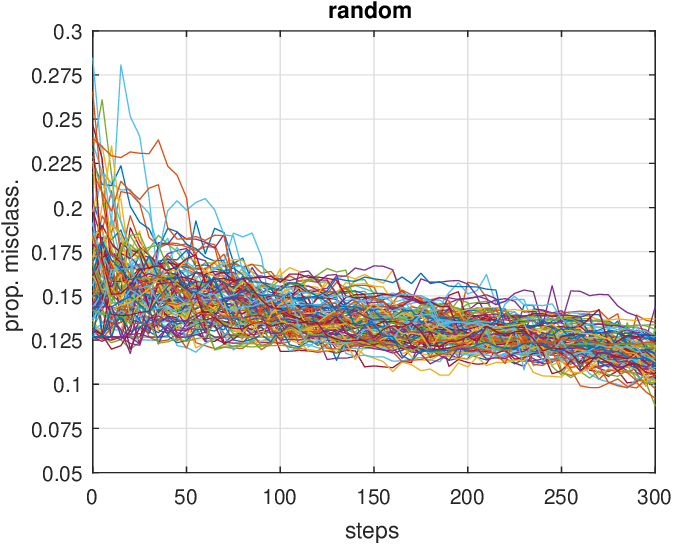}
  \caption{%
    Different sample paths of the proportion of misclassified points
    vs.\ number of steps, for 100~repetitions of the algorithms on the
    test function~$f_2$.}
  \label{eem:SM:fig:trajs_results_f2}
\end{figure}

\subsubsection{\texorpdfstring{Example 3}{Example 3}}

See Figures~\ref{eem:SM:fig:stats_results_f3_2}--\ref{eem:SM:fig:trajs_results_f3}.

\begin{figure}[p]
  \includegraphics[width=\toto]{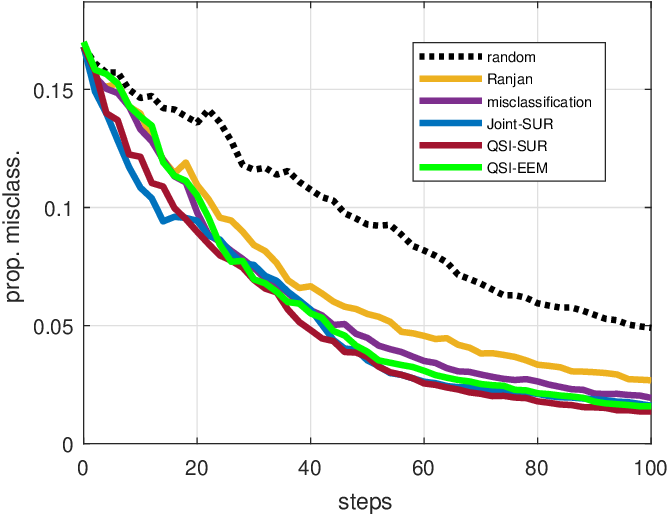}\hspace*{5mm}
  \includegraphics[width=\toto]{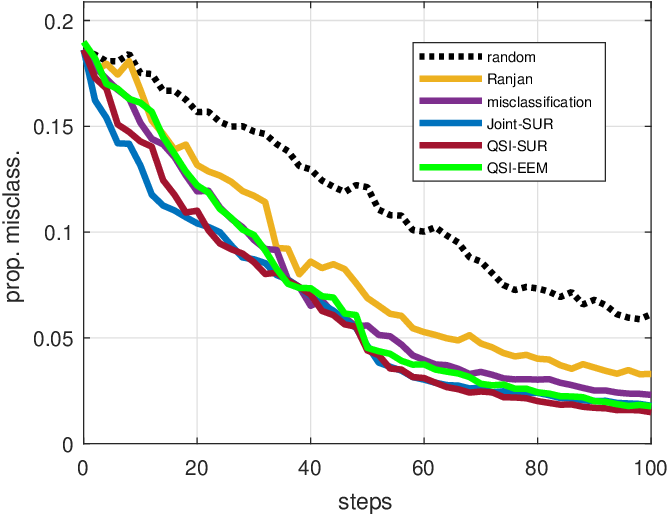}
  \caption{%
    Quantiles of level~$0.75$ and~$0.95$ for the proportion of
    misclassified points vs.\ number of steps, for 100~repetitions of
    the algorithms on the test function~$f_3$.}
  \label{eem:SM:fig:stats_results_f3_2}
\end{figure}

\begin{figure}[p]
  \includegraphics[width=\toto]{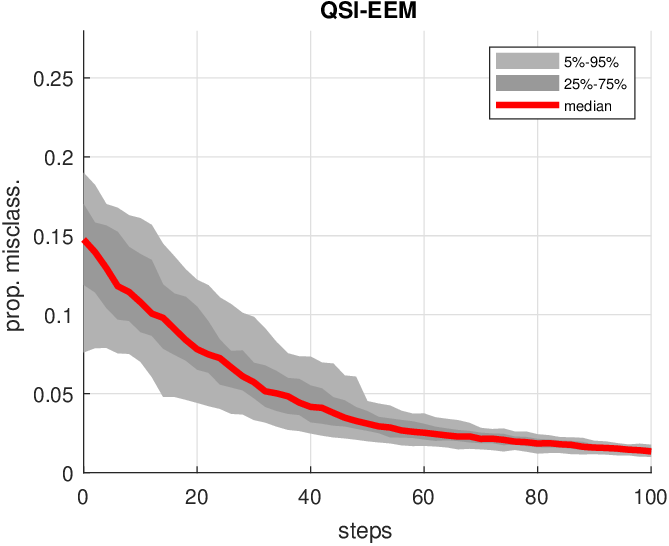}\hspace*{5mm}
  \includegraphics[width=\toto]{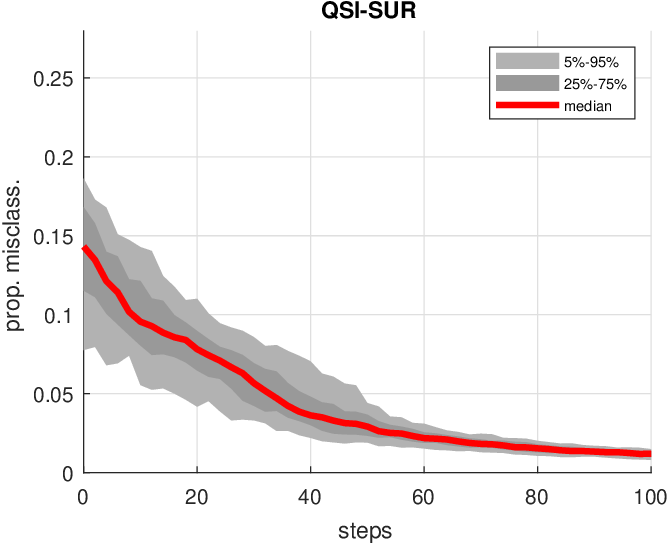}\\[5mm]
  \includegraphics[width=\toto]{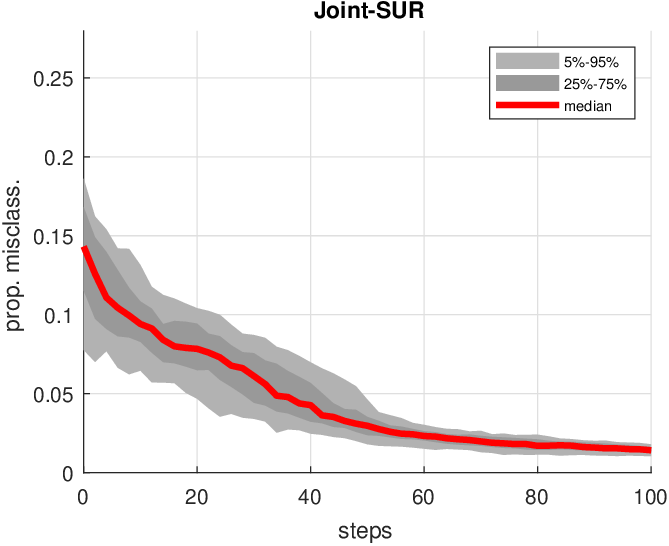}\hspace*{5mm}
  \includegraphics[width=\toto]{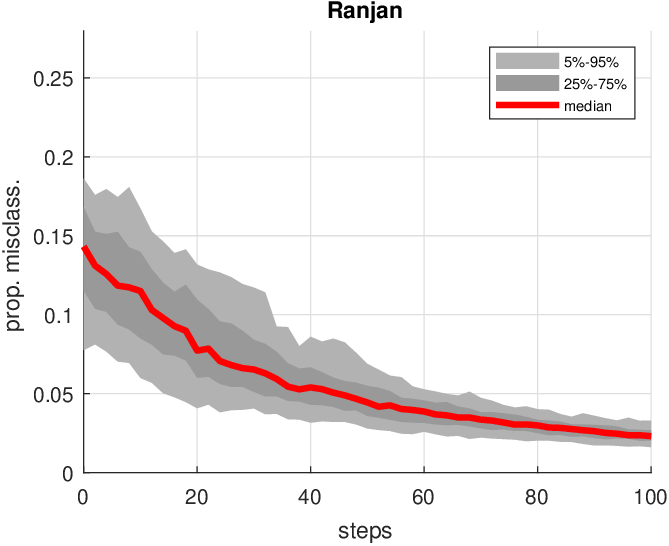}\\[5mm]
  \includegraphics[width=\toto]{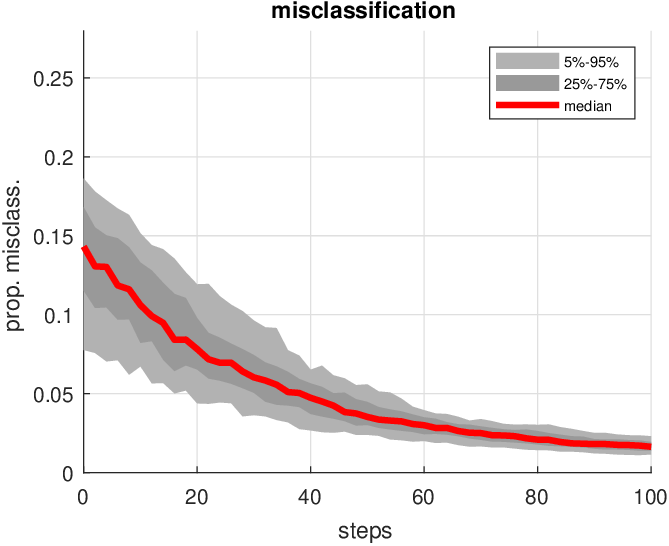}\hspace*{5mm}
  \includegraphics[width=\toto]{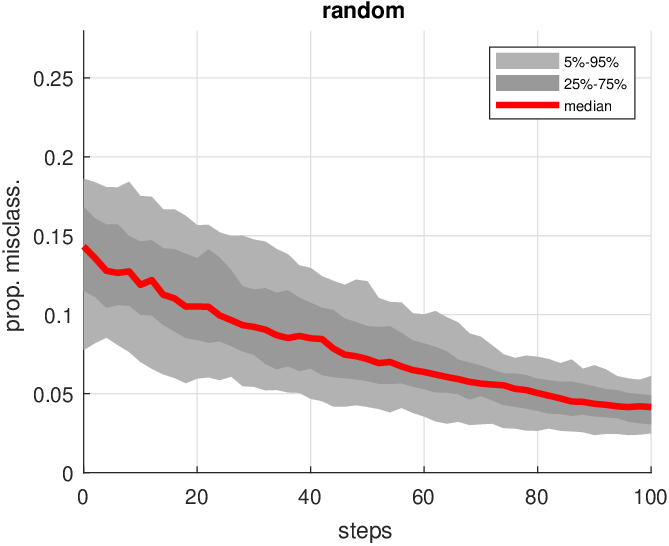}
  \caption{%
    Median and several quantiles of the proportion of misclassified
    points vs.\ number of steps, for 100~repetitions of the algorithms
    on the test function~$f_3$.}
  \label{eem:SM:fig:stats_results_f3}
\end{figure}

\begin{figure}[p]
  \includegraphics[width=\toto]{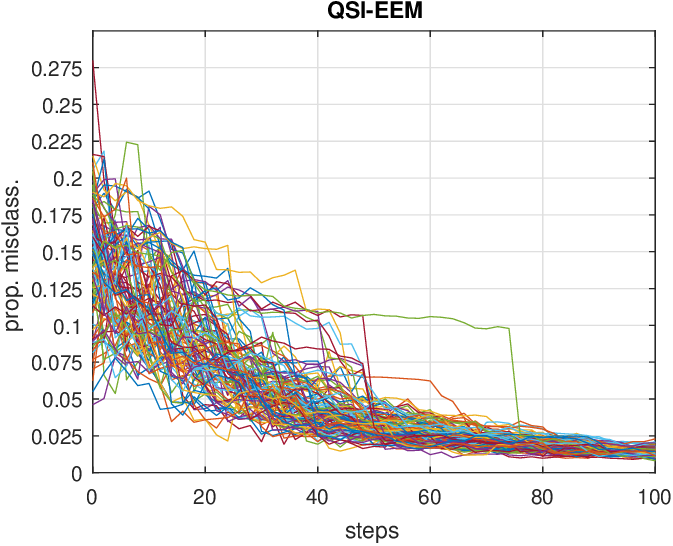}\hspace*{5mm}
  \includegraphics[width=\toto]{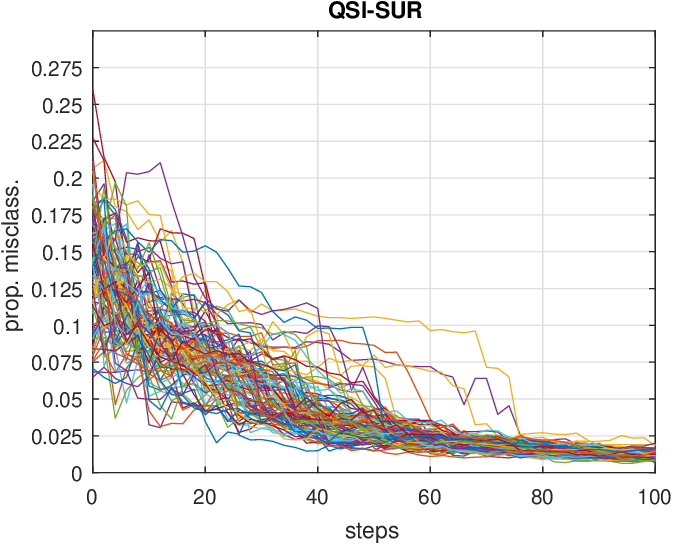}\\[5mm]
  \includegraphics[width=\toto]{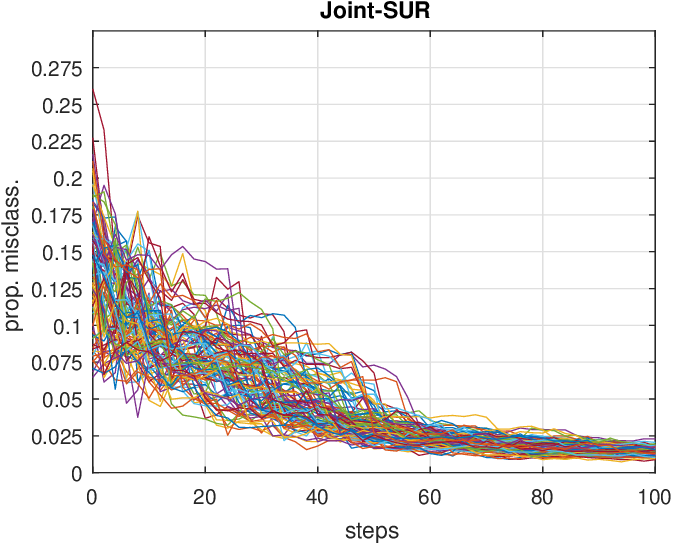}\hspace*{5mm}
  \includegraphics[width=\toto]{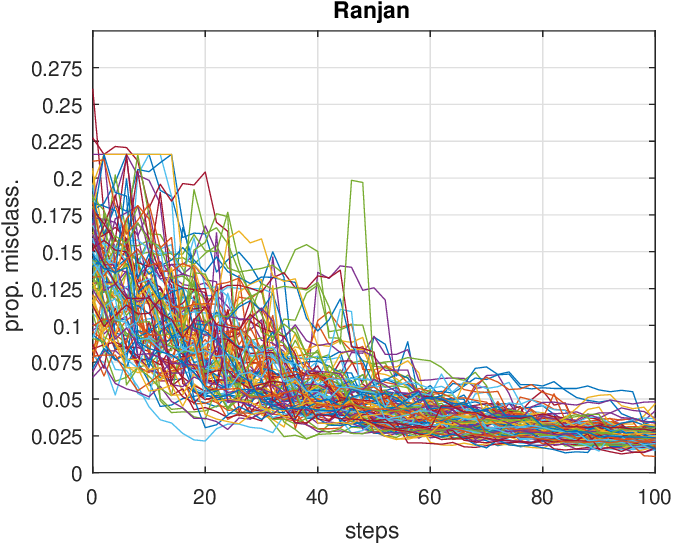}\\[5mm]
  \includegraphics[width=\toto]{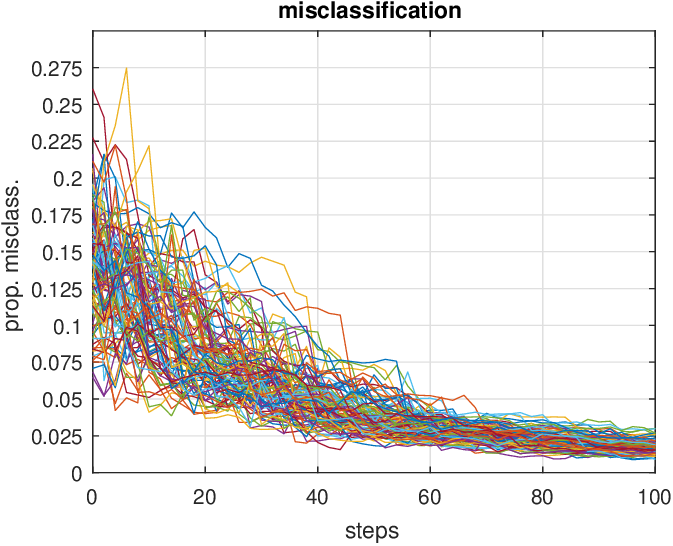}\hspace*{5mm}
  \includegraphics[width=\toto]{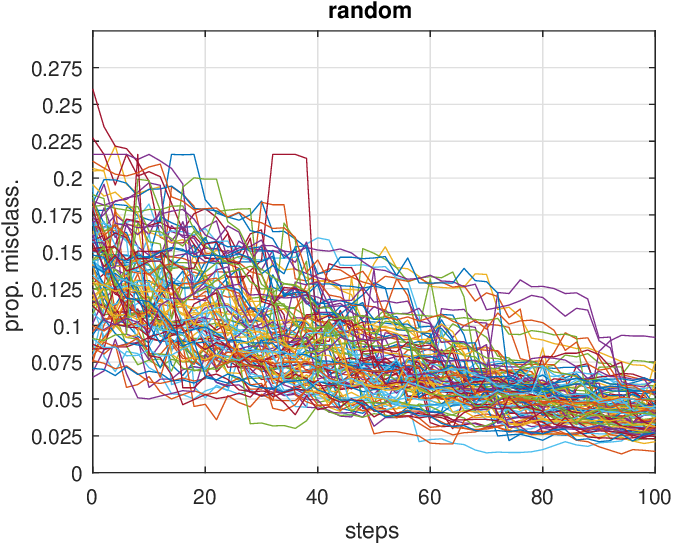}
  \caption{%
    Different sample paths of the proportion of misclassified points
    vs.\ number of steps, for 100~repetitions of the algorithms on the
    test function~$f_3$.}
  \label{eem:SM:fig:trajs_results_f3}
\end{figure}

\end{document}